\documentclass[11pt]{article}

\usepackage{amsfonts,latexsym,amsthm,amssymb,amsmath,amscd,euscript,tikz,mathtools}
\usepackage{framed}
\usepackage[normalem]{ulem}
\usepackage[margin=1in]{geometry}
\usepackage{color}
\usepackage[colorlinks=true,citecolor=blue,linkcolor=blue]{hyperref}
\usepackage{enumitem}
\usepackage{siunitx}
\usepackage{textcomp}
\usepackage{physics}
\usepackage{mathrsfs}
\usepackage{dsfont}
\usepackage[ruled,vlined,linesnumbered]{algorithm2e}
\usepackage{titlesec}
\usepackage{tikz-cd}
\usepackage{thmtools}
\usepackage{thm-restate}
\usepackage{cleveref}
\usepackage{stmaryrd}
\usepackage{subcaption}
\usepackage{graphicx}
\usepackage{quantikz}
\usetikzlibrary{positioning,chains,fit,shapes,calc}

\allowdisplaybreaks[1]

\newtheorem{theorem}{Theorem}[section]
\newtheorem{proposition}[theorem]{Proposition}
\newtheorem{lemma}[theorem]{Lemma}
\newtheorem{claim}[theorem]{Claim}
\newtheorem{corollary}[theorem]{Corollary}

\newtheorem{definition}[theorem]{Definition}

\newtheorem{remark}[theorem]{Remark}

\theoremstyle{remark}

\newcommand{\comment}[1]{}

\newcommand{\suggestion}[1]{}

\newcommand{\cA}{\mathcal{A}}\newcommand{\cB}{\mathcal{B}}
\newcommand{\cC}{\mathcal{C}}\newcommand{\cD}{\mathcal{D}}
\newcommand{\cE}{\mathcal{E}}\newcommand{\cF}{\mathcal{F}}

\newcommand{\cI}{\mathcal{I}}
\newcommand{\cL}{\mathcal{L}}
\newcommand{\cM}{\mathcal{M}}
\newcommand{\cO}{\mathcal{O}}
\newcommand{\cQ}{\mathcal{Q}}\newcommand{\cR}{\mathcal{R}}

\newcommand{\cV}{\mathcal{V}}

\newcommand{\bC}{\mathbb{C}}
\newcommand{\bF}{\mathbb{F}}

\newcommand{\bN}{\mathbb{N}}

\newcommand{\bZ}{\mathbb{Z}}

\newcommand{\1}{\mathds{1}}

\newcommand{\poly}{\operatorname{poly}}

\newcommand{\Enc}{\operatorname{Enc}}

\newcommand{\nc}{\newcommand}

\nc{\on}{\operatorname}
\nc{\Spec}{\on{Spec}}
\nc{\Aut}{\textit{Aut}}
\nc{\id}{\textit{id}}
\nc{\chr}{\on{char}}
\nc{\im}{\on{im}}
\nc{\Hom}{\on{Hom}}
\nc{\lcm}{\on{lcm}}
\nc{\dual}[1]{\prescript{t}{}{#1}}
\nc{\transpose}[1]{{#1}^{\intercal}}
\nc{\Sym}{\on{Sym}}
\nc{\End}{\on{End}}
\nc{\stab}{\on{stab}}
\nc{\Li}{\on{Li}}
\nc{\spn}{\on{span}}
\nc{\sgn}{\on{sign}}
\nc{\supp}{\on{supp}}
\nc{\Unif}{\on{Unif}}

\makeatletter
\newcommand\footnoteref[1]{\protected@xdef\@thefnmark{\ref{#1}}\@footnotemark}
\makeatother

\title{Constant-Overhead Addressable Gates \\ via Single-Shot Code Switching}
\author{
  Louis Golowich \\
  UC Berkeley \& IBM Quantum \\
  \href{mailto:lgolowich@berkeley.edu}{\texttt{lgolowich@berkeley.edu}}
  \and
  Kathleen (Katie) Chang \\
  Yale University \& IBM Quantum \\
  \href{mailto:katie.chang@yale.edu}{\texttt{katie.chang@yale.edu}}
  \and
  Guanyu Zhu \\
  IBM Quantum \\
  \href{mailto:Guanyu.Zhu@ibm.com}{\texttt{Guanyu.Zhu@ibm.com}}
}

\begin{document}

\maketitle

\begin{abstract}
  It is a major challenge to perform addressable and parallel logical operations on constant-rate quantum LDPC (qLDPC) codes. Indeed, the overhead of targeting specific logical qubits represents a crucial bottleneck in many quantum fault-tolerance schemes.

  We introduce fault-tolerant protocols for performing various addressable as well as parallel logical operations with constant space-time overhead, on a family of constant-rate and polynomial-distance qLDPC codes. Specifically, we construct gadgets for a large class of permutations of logical qubits. We apply these logical permutations to construct gadgets for applying a targeted Hadamard (or $CNOT$) gate on any chosen logical qubit (pair). We also construct gadgets for preparing logical code states, and for applying Hadamard gates on all logical qubits in a codeblock. All of our gadgets use constant quantum space-time overhead along with polynomially bounded classical computation. Prior protocols for such operations required larger overhead, or else relied on codes with certain symmetries that lack known asymptotic constructions.

  Our codes are given by tensor (i.e.~hypergraph) products of classical codes constructed from lossless expander graphs. Our core technical contribution is a constant-overhead code-switching procedure between 2- and 3-dimensional product codes, which generalizes Bomb\'{i}n’s dimensional jump (arXiv:1412.5079). We provide rigorous fault-tolerance proofs for our gadgets, and specifically prove a constant threshold under locally stochastic noise. Along the way, we develop a small-set flip decoder for high-dimensional product codes from lossless expanders. Our techniques yield additional interesting consequences, such as single-shot state preparation of 2-dimensional product codes with constant space-time overhead. We also propose a method for performing parallel non-Clifford gates by extending our techniques to codes supporting transversal application of such gates.

\end{abstract}

\newpage

\tableofcontents

\newpage

\section{Introduction}
\label{sec:intro}
Constant-rate quantum LDPC (qLDPC) codes provide fault-tolerant quantum memories with just constant space overhead. However, it remains a major challenge to efficiently perform logical computations on the encoded data. Existing approaches typically incur large (growing in the block length) space-time overheads for such gates that perform highly targeted or parallel operations on specific logical qubits within a codeblock.

In this work, we construct protocols for performing both addressable (i.e.~targeted) and parallel fault-tolerant logical operations with constant space-time overhead, on constant-rate codes. Constant overhead here means that the number of physical qubits and physical timesteps are within a constant factor of the number of logical qubits and logical timesteps, respectively.
Specifically, we show the following:

\begin{theorem}[Informal]\footnote{This theorem statement combines the gadgets in Proposition~\ref{prop:stateprep} and Proposition~\ref{prop:switchdown} (for item~\ref{it:mainprep}), Corollary~\ref{cor:logcyc} (for item~\ref{it:mainpermcyc}), Proposition~\ref{prop:permslab} (for item~\ref{it:mainpermgen}), Proposition~\ref{prop:hadsame} (for item~\ref{it:mainHa}), and Corollary~\ref{cor:target} (for item~\ref{it:mainHt} and item~\ref{it:mainCNOTt}). By Lemma~\ref{lem:redPauli} and Lemma~\ref{lem:threshold}, the threshold for these gadgets follows by combining the results listed above with the error-correction gadget in Proposition~\ref{prop:errcorr}.}
  \label{thm:main}
  There exists an infinite family of\footnote{Recall that an $[[n,k,d]]$ code encodes $k$ logical qubits into $n$ physical qubits with distance $d$.} $[[n,\;k=\Theta(n),\;d=\poly(n)]]$ qLDPC codes for which the following logical operations can be performed in constant quantum space-time overhead, while using polynomially bounded noiseless classical computation. These constructions exhibit a threshold under locally stochastic noise. 
  \begin{enumerate}
  \item\label{it:mainprep} Initialize a new codeblock with all logical qubits in the $\ket{0}$ or $\ket{+}$ state.
  \item\label{it:mainperm} Permute the logical qubits in a codeblock according to an appropriate permutation $\pi:[k]\rightarrow[k]$, meaning that logical qubit $j$ is moved to $\pi(j)$. For instance:
    \begin{enumerate}
    \item\label{it:mainpermcyc} For every $s\in[k]$, we can choose $\pi$ be the cyclic shift $\pi(j)\equiv j+s\pmod{k}$.
    \item\label{it:mainpermgen} For every constant $r\geq 3$, we can set $k=\ell^r$ for $\ell\in\bN$. Then for every $h\in[r]$ and every permutation $\sigma:[\ell]\rightarrow[\ell]$ we can choose $\pi=I^{\otimes h-1}\otimes\sigma\otimes I^{r-h}$ to be the permutation on $[\ell]^r\cong[k]$ that applies $\sigma$ to the $h^{\mathrm{th}}$ coordinate.
    \end{enumerate}
  \item\label{it:mainHa} Perform logical Hadamard gates on all logical qubits in a codeblock.
  \item\label{it:mainHt} Perform a targeted logical Hadamard gate on any single logical qubit.
  \item\label{it:mainCNOTt} Perform a targeted logical $CNOT$ gate between any pair of logical qubits within a codeblock, or across two codeblocks.
  \end{enumerate}
\end{theorem}

\begin{remark}
  Our qLDPC codes in Theorem~\ref{thm:main} are actually capable of encoding a larger number $k'=O(k)$ of logical qubits. However, due to the structure of our fault-tolerance scheme, we only encode our message into a fixed subset containing $k$ of these possible $k'$ logical qubits. 
\end{remark}

To our knowledge, our construction is the first achieving constant space-time overhead fault-tolerant (i.e.~exhibiting a threshold) gadgets for items~\ref{it:mainprep}--\ref{it:mainCNOTt} in Theorem~\ref{thm:main}. Many prior proposals for performing the fault-tolerant logical operations described in Theorem~\ref{thm:main} require space-time overhead growing polynomially in the code distance (or the protocol's fault distance) $d$, or else require specific code properties that are difficult to achieve with good parameters.

For instance, standard approaches for the state preparation in item~\ref{it:mainprep} incur polynomial overheads from repeated measurements (in space or time) or concatenation; see e.g.~\cite{gottesman_fault-tolerant_2014,bergamaschi_fault_2025}. While it was previously known that such state preparation can be performed with constant overhead using asymptotically good $[[n,\Theta(n),\Theta(n)]]$ quantum locally testable codes (qLTCs) \cite{pattison_personal_2024}, no such codes have yet been constructed. A state preparation scheme with low (though still super-constant) overhead was given by \cite{nguyen_quantum_2025} using nearly good qLTCs \cite{dinur_expansion_2024,kalachev_maximally_2025}. In contrast, we achieve the constant-overhead state preparation in item~\ref{it:mainprep} for codes as simple as (the disjoint union of multiple copies of) 2-dimensional hypergraph product codes \cite{tillich_quantum_2014}; see Section~\ref{sec:csinf} below.

Prior approaches for the constant-overhead cyclic permutations in item~\ref{it:mainpermcyc} used codes with strong symmetry properties for which good asymptotic families are not known \cite[Section~A.5]{xu_fast_2024}. Perhaps surprisingly, we obtain such cyclic permutations, as well as the more general permutations in item~\ref{it:mainpermgen}, on families of qLDPC codes without any special symmetry properties. Similarly, for product codes with an appropriate symmetry, the fold-transversal Hadamard gate \cite{breuckmann_fold-transversal_2024,quintavalle_partitioning_2023} achieves item~\ref{it:mainHa} up to some logical swaps, which then may require polynomial time to revert (e.g.~\cite[Section~A.5]{xu_fast_2024}). We avoid this polynomial overhead, as well as the symmetry assumption.

For the targeted logical Clifford gates in item~\ref{it:mainHt} and item~\ref{it:mainCNOTt}, prior approaches such as qLDPC-code lattice surgery (e.g.~\cite{cohen_low-overhead_2022,williamson_low-overhead_2024,cross_improved_2024,swaroop_universal_2025,he_extractors_2025}, which can be viewed as an application of stabilizer-weight reduction \cite{hastings_quantum_2023}) or concatenation \cite{gottesman_fault-tolerant_2014} naively incur a $\poly(d)$ time overhead. As described in \cite{gottesman_fault-tolerant_2014}, this time overhead can be reduced by encoding the $k$ logical qubits into $\poly(k)$ smaller codeblocks, and then trading off time for space when performing gates on these smaller codeblocks. Such techniques may yield similar results as item~\ref{it:mainHt} and item~\ref{it:mainCNOTt} in Theorem~\ref{thm:main}, though with some caveats. For instance, \cite{gottesman_fault-tolerant_2014} still requires super-constant time overhead for preparing certain logical ancilla states.\footnote{It may be possible to reduce this time overhead in \cite{gottesman_fault-tolerant_2014} to constant for the special case of Clifford circuits; it is an interesting direction of future work to further develop such optimizations and compare them to our methods.} Meanwhile, to our knowledge it remains an open question to construct an asymptotically efficient decoder for qLDPC-code lattice surgery schemes. In contrast, all of our gadgets use just polynomial-sized classical computation, including the cost of decoding.

\cite{nguyen_quantum_2025} provide gadgets for targeted logical gates that have a low space-time overhead in an amortized sense, meaning that the cost per gate is lower when many gates are performed in parallel. However, even this amortized overhead is super-constant, while the non-amortized overhead to perform a single targeted gate is at least polynomial in the distance. Furthermore, \cite{nguyen_quantum_2025} use somewhat involved distillation techniques within nearly good quantum locally testable codes. In contrast, we use more elementary codes; see Section~\ref{sec:csinf} below.

While \cite{he_quantum_2025,he_asymptotically_2025} construct asymptotically good codes with fully addressable and transversal non-Clifford (as well as Clifford) gates, their codes are not LDPC, but rather have linear-weight stabilizers. Therefore additional techniques such as concatenation will be needed to perform fault-tolerant syndrome extraction and error correction with these codes, which introduces a growing space-time overhead.

\subsection{Product Codes with Single-Shot Code Switching}
\label{sec:csinf}
Our codes in Theorem~\ref{thm:main} are given by tensor products of (co)chain complexes associated to lossless expander graphs, which in turn can be constructed randomly (e.g.~\cite[Theorem 4.16]{hoory_expander_2006}) or explicitly \cite{hsieh_explicit_2025} (see Section~\ref{sec:lossless}). Two-dimensional such product codes are often called hypergraph product codes \cite{tillich_quantum_2014}. The key technical ingredient behind Theorem~\ref{thm:main} is a \emph{single-shot code switching} procedure that we provide, which allows us to fault-tolerantly switch between product codes of different dimensions in constant space-time overhead, while preserving the encoded logical qubits. Specifically, for arbitrary fixed $r\geq 3$, the gadgets in Theorem~\ref{thm:main} are based on (sometimes repeated) switching between an $r$-dimensional $[[\Theta(n^r),\Theta(n^r),\Omega(n)]]$ product code and a collection of $\Theta(n)$ copies of an $(r-1)$-dimensional $[[\Theta(n^{r-1}),\Theta(n^{r-1}),\Omega(n)]]$ product code. For item~\ref{it:mainprep} in Theorem~\ref{thm:main}, we also develop a gadget for single-shot state preparation, which shares many technical ingredients with the code switching.

Our single-shot code switching procedure, which we prove exhibits a threshold under locally stochastic noise, has many interesting consequences. For instance, Section~\ref{sec:tarinf} below describes how we can use such code switching to target individual logical qubits. Code switching also allows us to prove Theorem~\ref{thm:main} for the code given by $\Theta(\sqrt{n})$ disjoint copies of a $[[n,\Theta(n),\Theta(\sqrt{n})]]$ 2-dimensional product code. In particular, we prove item~\ref{it:mainprep} by first preparing logical $\ket{0}$ or $\ket{+}$ states in a 3-dimensional product code, and then switching down to a collection of 2-dimensional codeblocks. Yet such 2-dimensional codes typically do not support constant-overhead logical state preparation, at least when preparing a single codeblock. It is therefore perhaps surprising that we are able to circumvent this barrier by preparing $\Theta(\sqrt{n})$ codeblocks at once.

Our code switching procedure can be viewed as a generalization of Bombin's code switching, or ``dimensional jump,'' on color/toric codes (\cite{bombin_dimensional_2016}, see also \cite{bombin_single-shot_2015-1}). Whereas a toric code is a product of classical repetition codes and therefore has poor rate, we obtain constant space-time overhead gadgets by taking products of constant-rate classical LDPC codes.

However, the fundamental idea is similar: to switch down from an $r$-dimensional codeblock to a collection of $(r-1)$-dimensional codeblocks, we simply measure out a subset of the physical qubits, then we run error correction on this measurement outcome, based on which we apply an appropriate Pauli correction. To switch back up to the $r$-dimensional code, we run a logical teleportation circuit, which requires logical bell pairs between $r$- and $(r-1)$-dimensional codes. To construct such logical bell pairs, we first prepare two $r$-dimensional codeblocks, one with logical $\ket{0}$ states and one with logical $\ket{+}$ states. We then apply our downwards switching to one of the codeblocks, and apply transversal $CNOT$ gates from the $\ket{+}$ block to the $\ket{0}$ block. Note that this transversal $CNOT$ is applied between an $(r-1)$-dimensional code and an $r$-dimensional code (see Lemma~\ref{lem:CNOTdiff}).

\subsection{Constant-Overhead Targeted Gates via Code Switching}
\label{sec:tarinf}

\begin{figure}
  \centering 
  \begin{subfigure}{0.4\textwidth}
    \includegraphics[width=\textwidth]{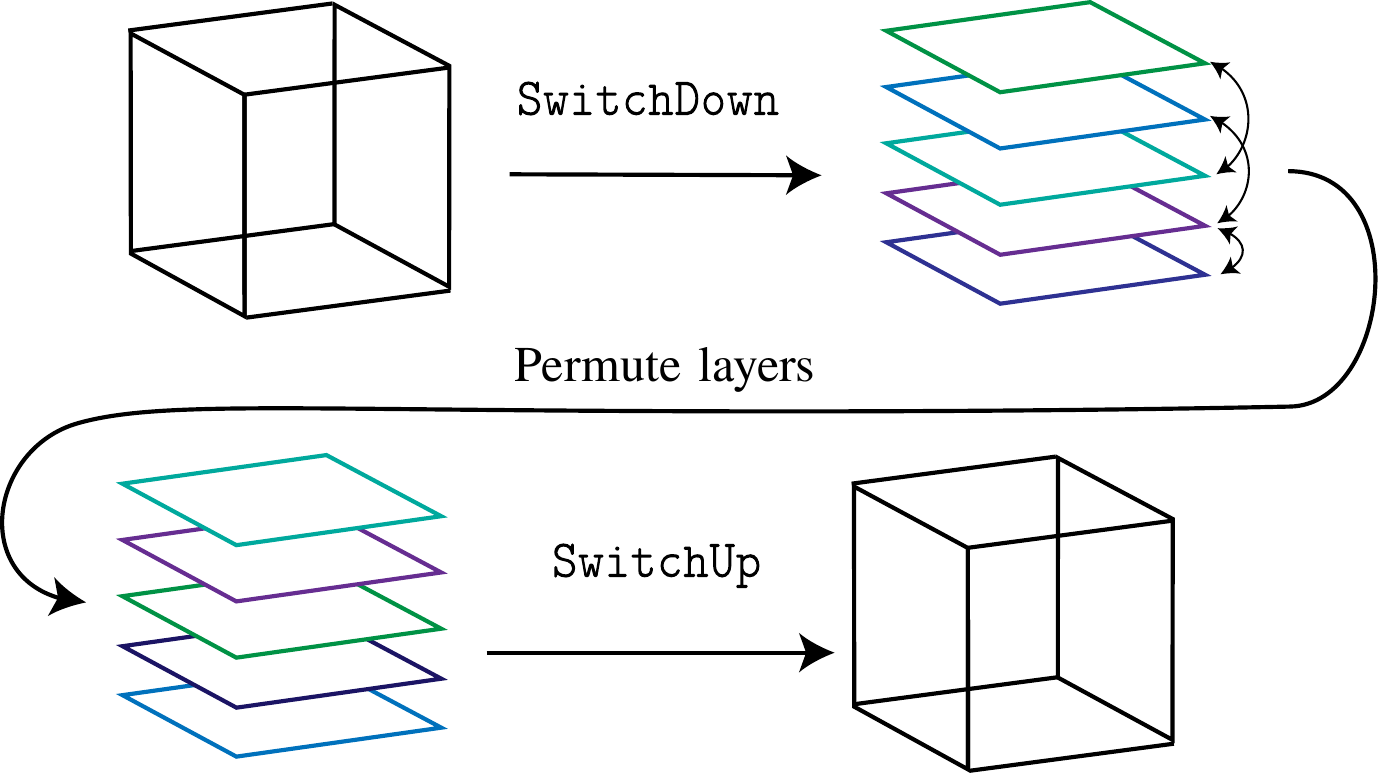}
    \captionsetup{width=\textwidth}
    \caption{\label{fig:swap} Gadget used in item~\ref{it:mainperm} of Theorem~\ref{thm:main} to permute ``slabs'' of logical qubits.}
  \end{subfigure}
  \begin{subfigure}{0.5\textwidth}
    \includegraphics[width=\textwidth]{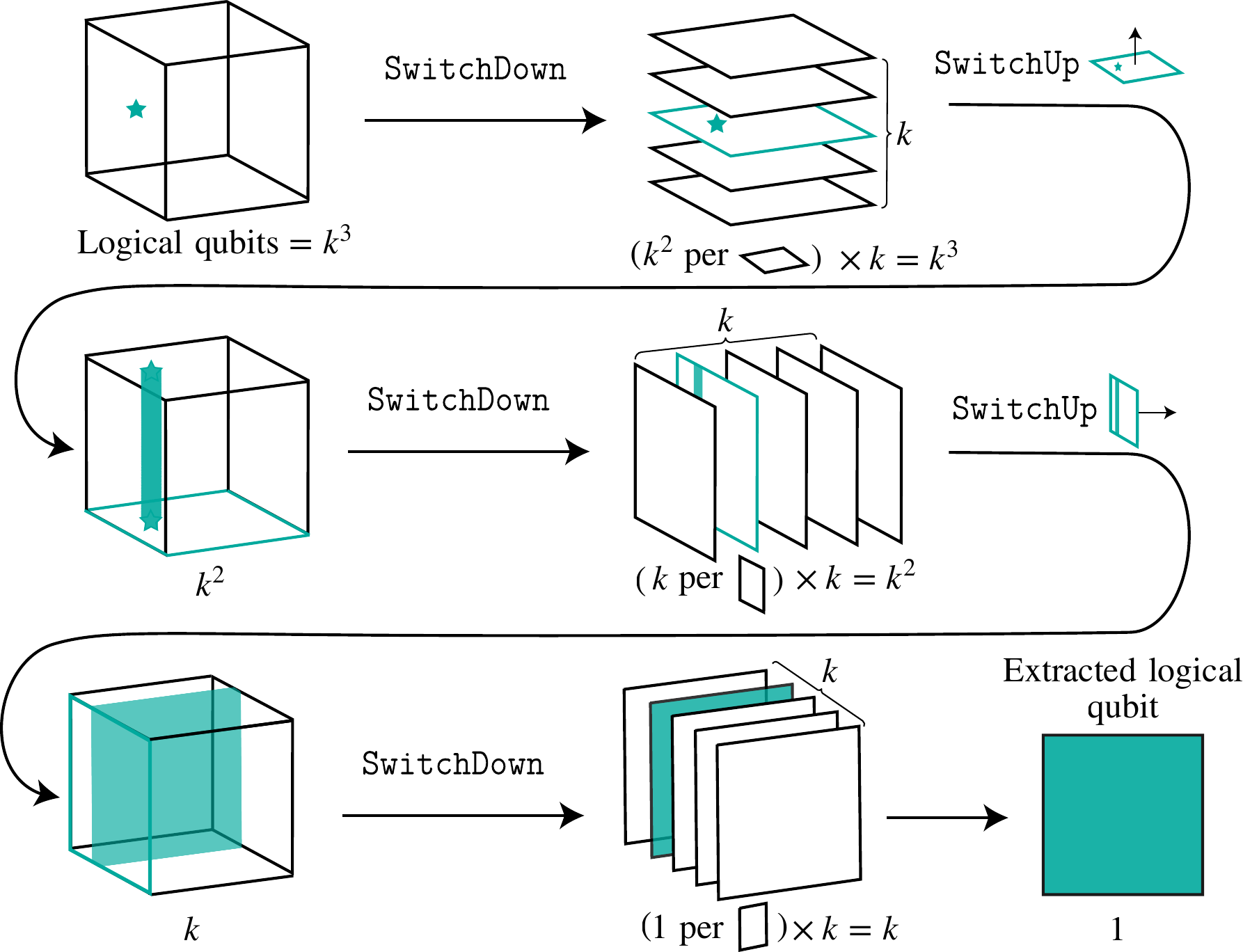}
    \captionsetup{width=\textwidth}
    \caption{\label{fig:target} Gadget used in item~\ref{it:mainHt} and item~\ref{it:mainCNOTt} of Theorem~\ref{thm:main} to target a logical qubit.}
  \end{subfigure}\hspace{0.05\textwidth}
  \caption{Applications of code switching to gadgets in Theorem~\ref{thm:main}.}
\end{figure}

We now outline how we use code switching to construct the gadgets in Theorem~\ref{thm:main}. For illustrative purposes here we consider the $r=3$ case; the generalization to arbitrary constant $r\geq 3$ is straightforward.

\textbf{Permuting logical qubits.} Figure~\ref{fig:swap} illustrates how we obtain item~\ref{it:mainpermgen} of Theorem~\ref{thm:main} by switching down, permuting the resulting 2-dimensional codeblocks, and then switching back up.

We then obtain item~\ref{it:mainpermcyc} by applying item~\ref{it:mainpermgen} in all $r$ directions. Specifically, letting $k=\ell_1\ell_2\cdots\ell_r$ for relatively prime $\ell_1,\dots,\ell_r\in\bN$, then $\bZ_k\cong\bZ_{\ell_1}\times\cdots\times\bZ_{\ell_r}$. Thus cyclic shifts of $k$ logical qubits can be realized by arranging the qubits in an ($r$-dimensional) $\ell_1\times\cdots\times\ell_r$ hypercube, and then performing cyclic shifts in each of the $r$ directions.

\textbf{Parallel logical Hadamard.} For item~\ref{it:mainHa} in Theorem~\ref{thm:main}, recall that applying transversal Hadamard gates across all physical qubits in a CSS code always induces logical Hadamard gates on all logical qubits, but also switches to the ``dual'' code, meaning that the $X$ and $Z$ stabilizers are swapped. We then use our code switching gadget (iteratively in each of the $r$ directions) to switch back to the original code.

\textbf{Targeting a single logical qubit in constant space and time.} Figure~\ref{fig:target} illustrates how we can begin with $k^3$ logical qubits encoded into a 3-dimensional $[[\Theta(k^3),k^3,\Omega(k)]]$ codeblock, and then perform code switching in each of the 3 directions interspersed with swap gates to extract any single desired logical qubit. Specifically, we begin with a 3-dimensional $k\times k\times k$ cube of logical qubits. By switching down and back up in the first direction, we are able to extract a $k\times k$ plane containing the desired logical qubit, which we encode into an ancilla 3-dimensional codeblock. Repeating this procedure in the other two directions, we ultimately extract a codeblock with a single logical qubit. We can then perform transversal $CNOT$ or Hadamard gates on this extracted qubit, before reverting the extraction procedure to insert the qubit back into its original location.

As each code switching step takes constant time, the entire procedure takes constant time. Although we extract one logical qubit into its own $\poly(k)$-sized codeblock, the scheme retains constant rate (i.e.~constant space overhead) because we do not parallelize the gadget. That is, we always have $k^3$ logical qubits collectively encoded into $\Theta(k^3)$ physical qubits. Thus we obtain item~\ref{it:mainHt} and item~\ref{it:mainCNOTt} of Theorem~\ref{thm:main}.

\subsection{Fault-Tolerance Analysis and Decoder}
\label{sec:ftinf}
An additional contribution of our paper lies in our rigorous proofs of fault-tolerance. Recall that a noise distribution is \emph{$\epsilon$-locally stochastic} if the probability that a set $S$ of physical qubits\footnote{We will not distinguish between physical qubit and measurement errors, as a measurement error is indistinguishable from qubit errors before and after the measurement.} lies within the support of the corruption is $\leq\epsilon^{|S|}$. As mentioned above, we prove that our gadgets exhibit fault-tolerance properties that imply a threshold under $\epsilon$-locally stochastic noise for sufficiently small constant $\epsilon>0$ (see Lemma~\ref{lem:redPauli} and Lemma~\ref{lem:threshold}). That is, under such physical noise, our logical error rate decays as $e^{-\poly(N)}$, where $N$ is the number of physical qubits. We prove such a threshold by adapting ideas from \cite{kovalev_fault_2013,gottesman_fault-tolerant_2014} to argue that errors cannot be highly concentrated inside appropriately-sized connected subgraphs of a constant-degree graph associated to our code.


We prove this threshold behavior under a new small-set flip decoder for high-dimensional product codes from lossless expanders. To construct this decoder, which uses constant quantum time and linear classical time, we generalize the decoder for 2-dimensional codes of \cite{leverrier_quantum_2015,fawzi_constant_2020}, while also applying ideas from \cite{dinur_good_2023,dinur_expansion_2024,kalachev_maximally_2025}. An informal version of our decoding result is stated below.

\begin{proposition}[Informal statement of Proposition~\ref{prop:ssflip} and Proposition~\ref{prop:errcorr}]
  \label{prop:decinf}
  For arbitrary constant $r\geq 2$, let $Q$ be a $[[\Theta(n^r),\Theta(n^r),\Omega(n)]]$ qLDPC code associated to an $r$-dimensional tensor product of $n$-vertex lossless expanders. Then there exists a decoder for $Q$ that uses constant quantum time along with $O(n^r)$ classical time, which provides:
  \begin{itemize}
  \item correction of up to $\delta n$ adversarial errors for a sufficiently small constant $\delta>0$, and
  \item a threshold (i.e.~logical error rate $\leq e^{-\poly(n)}$) under $\epsilon$-locally stochastic noise for a sufficiently small constant $\epsilon>0$.
  \end{itemize}
\end{proposition}

\subsection{Adding Magic to Achieve Universality}
While the results listed above provide constant-overhead gadgets for logical Clifford circuits, we can extend our techniques to obtain a scheme for universal fault-tolerant computation. For instance, we may add in a gadget for magic state distillation. Such gadgets with low (but super-constant) space-time overhead have previously been constructed, e.g.~in \cite{nguyen_quantum_2025}. 
We could alternatively try to apply our code switching gadgets to qLDPC codes supporting transversal non-Clifford gates \cite{bombin_exact_2007,bombin_topological_2007,bombin_gauge_2015,zhu_non-clifford_2023,lin_transversal_2024,golowich_quantum_2025-1,breuckmann_cups_2024,zhu_topological_2025}. 

We expand upon the latter approach in Appendix \ref{app:protocol}
by proposing a method for parallelizable logical non-Clifford gates using the codes of \cite{golowich_quantum_2025-1,zhu_topological_2025}. These codes support a transversal (i.e.~constant-depth) physical circuit inducing polynomially many logical $CCZ$ gates in a codeblock. However, we leave the problem of proving a threshold for single-shot code switching on such codes for future work; we discuss the challenges involved in Appendix \ref{app:protocol}. The independent and concurrent work of \cite{tan2025single} provides more results in this direction, for a class of codes based on those of \cite{zhu_topological_2025}.

\section{Technical Overview of Construction}
In this section, we provide an overview of the techinques underlying the fault-tolerant gadgets we use to prove Theorem~\ref{thm:main}, and we outline how the remainder of the paper is organized.

\subsection{Code Construction}
\label{sec:codeinf}
We first describe the codes we use to prove Theorem~\ref{thm:main}. For this purpose, we will need the notion of \emph{(co)chain complexes} and their products.

\begin{definition}[Informal statement of Definition~\ref{def:chaincom} and Definition~\ref{def:cctensor}]
  A \emph{$r$-dimensional cochain complex $\cC^*$ over $\bF+2$} is a sequence
  \begin{equation*}
    \cC^* = (C^0\xrightarrow{\delta_0}C^1\xrightarrow{\delta_1}\cdots\xrightarrow{\delta_{r-1}}C^r)
  \end{equation*}
  of $\bF_2$-vector spaces $\cC^i$ and \emph{coboundary maps} $\delta_i$ satisfying $\delta_i\delta_{i-1}=0$.

  For cochain complexes $\cA^*,\cB^*$ of respective dimensions $r^{\cA},r^{\cB}$, the \emph{tensor product $\cC^*=\cA^*\otimes\cB^*$} is the chain complex of dimension $r^{\cC}=r^{\cA}+r^{\cB}$ given by $\cC^i=\bigoplus_{j\in\bZ}\cA^j\otimes\cB^{i-j}$ and $\delta^{\cC}=\delta^{\cA}\otimes I+I\otimes\delta^{\cB}$.
\end{definition}

Recall that a length-$n$ quantum CSS code consists of a pair $Q=(Q_X,Q_Z)$ of classical linear codes $Q_X,Q_Z\subseteq\bF_2^n$ satisfying the orthogonality condition $Q_X^\perp:=\{y\in\bF_2^n:x\cdot y=0\;\forall x\in Q_X\}\subseteq Q_Z$. Therefore for an $r$-dimensional cochain complex $\cC^*$, for every level $1\leq i\leq r-1$ we have a naturally associated CSS code given by $Q_X=\ker(\delta_{i-1}^\top)$ and $Q_Z=\ker(\delta_i)$. Here the CSS orthogonality condition is equivalent to the cochain complex condition $\delta_i\delta_{i-1}=0$.

We prove Theorem~\ref{thm:main} using products of 1-dimensional cochain complexes that arise from a certain type of graph called a \emph{lossless expander}. Random bounded-degree graphs are lossless expanders with high probability, and explicit constructions are also known (see Section~\ref{sec:lossless}). In particular, for fixed $r\in\bN$ and for every $h\in[r]$, we let
\begin{equation*}
  {\cC^{(h)}}^* = \left(\bF_2^{V_L^{(h)}}\xrightarrow{\delta^{(h)}}\bF_2^{V_R^{(h)}}\right)
\end{equation*}
be a 1-dimensional cochain complex associated to a bounded-degree bipartite lossless expander $G^{(h)}=(V^{(h)}=V_L^{(h)}\sqcup V_R^{(h)},\;E^{(h)})$. The coboundary map $\delta^{(h)}$ is given by the bipartite adjacency matrix of $G^{(h)}$. We will choose such graphs for which all $|V^{(h)}|$ are the same. The quantum codes we use are those associated to some level $1\leq i\leq r-1$ of the product cochain complex
\begin{equation*}
  \cC^*={\cC^{(1)}}^*\otimes\cdots\otimes{\cC^{(r)}}^*.
\end{equation*}
Letting $n=|V^{(h)}|$ grow for fixed $r=O(1)$, then these codes are qLDPC with parameters $[[\Theta(n^r),\Theta(n^r),\Omega(n)]]$. The specific details of our instantiation are given in Section~\ref{sec:applycode}.

The $r=2$-dimensional case of these codes were introduced by \cite{tillich_quantum_2014}, where they were called ``hypergraph product codes.'' A key insight in our paper is that we can efficiently switch between product codes of different dimensions $r$.

\subsection{Error Correction, State Preparation, and Code Switching from Small-Set Flip Decoder}
\label{sec:decappinf}
The fault-tolerant procedures in Theorem~\ref{thm:main} are primarily based on three new gadgets that we construct for the product codes described in Section~\ref{sec:codeinf}. These gadgets, which perform error correction, state preparation, and code switching, respectively, all crucially rely on a new decoder that we develop for the codes in Section~\ref{sec:codeinf}. Specifically, in Section~\ref{sec:ssflip}, we construct a generalization of the \emph{small-set flip decoder}, which was originally introduced for $r=2$-dimensional tensor product codes \cite{leverrier_quantum_2015,fawzi_constant_2020}, to arbitrary $r\geq 2$.

Below, we define such small-set flip decoders, and state our result constructing them for product codes. Here we restrict attention to cochain complexes $\cC^*$ with a fixed basis $C^i$ for each $\cC^i$, so that $\cC^i=\bF_2^{C^i}$. We endow $C^0\sqcup\cdots\sqcup C^r$ with a partial order by letting $c^i\prec c^{i+1}$ for $c^i\in C^i,\;c^{i+1}\in C^{i+1}$ if $c^{i+1}\in\supp(\delta_i(\1_{c^i}))$, and then letting $c^i\prec c^j$ if $c^i\prec c^{i+1}\prec\cdots\prec c^{j-1}\prec c^j$ for some $c^{i+1}\in C^{i+1},\dots,c^{j-1}\in C^{j-1}$. We extend this notation to $c^i\in C^i$ and $c^j\in\cC^j=\bF_2^{C^h}$ by writing $c^i\preceq c^j$ if every $c'\in\supp(c^j)$ satisfies $c^i\preceq c'$. The cochain complexes we consider below have constant locality, meaning that for every $c^i\in C^i$, there are only a constant number of basis elements $c'$ satisfying $c^i\preceq c'$.

\begin{definition}[Abridged statement of Definition~\ref{def:ssflip}]
  \label{def:ssflipinf}
  For $m\in\bN$, we say a cochain complex $\cC^*$ has a \emph{$m$-small-set error-flip decoder at level $i$} if for every nonzero $e\in\cC^i$ of weight $|e|\leq m$, there exists a basis element $c^0\in C^0$ for which at least one of the following holds:
  \begin{enumerate}
  \item\label{it:notmininf} There exists a cochain $c^{i-1}\in\cC^{i-1}$ with $c^0\preceq c^{i-1}$ such that $|e+\delta(c^{i-1})|<|e|$, or
  \item\label{it:flipinf} There exists a cochain $c^i\in\cC^i$ with $c^0\preceq c^i$ such that $|\delta(e+c^i)|<|\delta(e)|$.
  \end{enumerate}
\end{definition}

\begin{proposition}[Informal statement of Proposition~\ref{prop:ssflip}]
  \label{prop:ssflipinf}
  The cochain complex $\cC^*={\cC^{(1)}}^*\otimes\cdots\otimes{\cC^{(r)}}^*$ in Section~\ref{sec:codeinf} has an $m$-small-set flip decoder at every level $0\leq i\leq r-1$ for $m=\Theta(|V^{(h)}|)$.\footnote{Recall here that all $h\in[r]$ have the same $|V^{(h)}|$.}
\end{proposition}

To generalize the 2-dimensional argument of \cite{leverrier_quantum_2015,fawzi_constant_2020} to the higher-dimensional case in Proposition~\ref{prop:ssflipinf}, we adapt techniques used by \cite{dinur_expansion_2024,kalachev_maximally_2025} to construct quantum locally testable codes. Specifically, \cite{dinur_expansion_2024,kalachev_maximally_2025} (see also \cite{nguyen_quantum_2025}) proved properties similar to small-set flip decodability for high-dimensional product codes that impose certain local codes exhibiting a \emph{robustness} property on a global \emph{spectral expander graph}. As our codes in Section~\ref{sec:codeinf} are constructed from lossless expanders instead of spectral expanders, we apply expansion in a different way, but we ultimately still leverage the robustness of local codes, which in our case are simply repetition codes (see Lemma~\ref{lem:robust}).

We apply this small-set flip decoder for two purposes, namely to perform error correction, and to show a sort of ``small-set soundness'' for high-dimensional product codes. For the error correction application, we simply take $e$ taken to be a low-weight error on a code state, so that Definition~\ref{def:ssflipinf} gives a correction $c^i$ that reduces the syndrome weight. Our error correction gadget in Section~\ref{sec:errcorr} simply repeatedly applies such corrections while they exist. We generalize the percolation-based arguments in \cite{fawzi_constant_2020} to the high-dimensional case to show that this error correction succeeds even in the presence of random data and syndrome errors. We similarly apply such error correction within our state preparation gadget in Section~\ref{sec:stateprep}, our code switching gadget in Section~\ref{sec:switchdown}, and our measurement gadget in Section~\ref{sec:measure}.

Meanwhile, for the ``small-set soundness'' application, we apply Definition~\ref{def:ssflipinf} to a low-weight \emph{syndrome} $e=\delta(f)$ for some $f\in C^{i-1}$. Intuitively, because $\cC^*$ has constant locality, then by repeatedly applying Definition~\ref{def:ssflipinf} to such an $e$, we can conclude that every sufficiently low-weight syndrome $e$ arises from some error $f$ of proportionally low weight $O(|e|)$. Note that this property is stronger than the ``small-set coboundary expansion'' property (see e.g.~\cite{dinur_good_2023,hopkins_explicit_2022-1}), which only requires that low-weight $f$ have sufficiently high-weight syndromes $e=\delta_{i-1}(f)$ so that $|f|\leq O(|e|)$.

We crucially apply this small-set soundness property in our state preparation gadget in Section~\ref{sec:stateprep}. This gadget prepares logical code states by preparing physical $\ket{0^n}$ or $\ket{+^n}$, and then measuring all the code stabilizers. We then run error correction on these measurement outcomes to compute an estimate $\tilde{s}$ of the syndrome, and then compute some error $\tilde{f}$ with $\delta(\tilde{f})=\tilde{s}$, and apply a correction based on $\tilde{f}$. We want to show that $\tilde{f}$ is close to the true error $f$ on the code state, given that $\tilde{s}=\delta(\tilde{f})$ is close to the true syndrome $s=\delta(f)$. Small-set soundness precisely ensures this property; we also again apply a percolation argument to deal with large numbers of random errors.

In our code switching gadget outlined in Section~\ref{sec:csinf} above, we similarly apply small-set soundness to argue that the error-corrected measurement outcomes lead to an accurate Pauli correction.



\subsection{Fault-Tolerance Analysis}
A significant technical contribution of our work lies in our proof of threshold under a constant rate of locally stochastic noise. Although our codes have sublinear distance, we are still able to handle such random linear-weight errors using percolation arguments adapted from \cite{kovalev_fault_2013,gottesman_fault-tolerant_2014}. Specifically, if our logical qubits are encoded in a CSS code associated to a cochain complex $\cC^*$, we track the propagation of errors through a \emph{connectivity graph} $G^{\cC}$ associated to $\cC^*$. Formally, the vertices of this graph are basis elements in $C^0\sqcup\cdots\sqcup C^r$, and an edge connects every $c,c'$ for which either there is some $c^0\in C^0$ with $c^0\preceq c,c'$, or some $c^r\in C^r$ with $c,c'\preceq c^r$ (see Definition~\ref{def:cctodec}). We then argue that under faults that are not too dense within any large connected subgraph of $G^{\cC}$, our gadgets preserve such a low density of errors within connected subgraphs (see Definition~\ref{def:graphtobad}). By standard percolation arguments, locally stochastic noise will indeed have low density within such large connected subgraphs (see Lemma~\ref{lem:threshold}). We furthermore apply a standard decomposition to reduce general noise to Pauli noise (see Lemma~\ref{lem:redPauli}).

Our fault-tolerance proofs crucially rely on the fact that the small-set flip decoder described in Section~\ref{sec:decappinf} by definition performs updates within constant-sized neighborhoods of the connectivity graph $G^{\cC}$. Therefore both when performing error correction and applying small-set soundness to compute Pauli corrections (see Section~\ref{sec:decappinf}), we control the propagation of errors within connected subgraphs of $G^{\cC}$. An interesting consequence of this argument is that our proof may not immediately extend to more global decoders, such as a maximum-liklihood decoder that computes the most likely error.

\subsection{Roadmap}
The remainder of this paper is organized as follows. Section~\ref{sec:prelim} provides necessary preliminary notions and prior results. In Section~\ref{sec:ssflip}, we present our small-set flip decoder. We present our gadgets for downwards and upwards code switching in Section~\ref{sec:switchdown} and Section~\ref{sec:switchup}, respectively. We postpone the proof for the upwards code switching gadget to Appendix~\ref{app:suproof}, as it is a direct application of other gadgets in the paper. We present our gadget for state preparation in Section~\ref{sec:stateprep}, and our gadget for error correction in Section~\ref{sec:errcorr}; the analyses are somewhat similar to that of the downwards code switching, so we postpone them to Appendix~\ref{app:spproof} and Appendix~\ref{app:ecproof}, respectively. Section~\ref{sec:basicgadgets} presents some more basic gadgets, namely for transversal $CNOT$ and Hadamard gates, as well as logical measurements. While these gadgets in Section~\ref{sec:basicgadgets} were previously known, we state and analyze them in our specific fault-tolerance setup for completeness. In Section~\ref{sec:apply}, we combine all of our gadgets listed above to prove Theorem~\ref{thm:main}.



\section{Preliminaries}
\label{sec:prelim}
This section presents preliminary notions and relevant known results.

\subsection{Notation}
This section describes the basic notation that we use throughout. For $n\in\bN$, we let $[n]=\{1,2,\dots,n\}$. For a set $S$, we let $2^S=\{S':S'\subseteq S\}$ denote the power set of $S$. We denote by $\bF_2=\{0,1\}$ the finite field on two elements. For $x\in\bF_2^n$, we let $|x|=|\{i\in[n]:x_i\neq 0\}|$ denote the Hamming weight of $x$. For $x,y\in\bF_2^n$, we let $x\cdot y=\sum_{i=1}^nx_iy_i$ denote the standard bilinear form. For $i\in[n]$, we let $\1_i\in\bF_2^n$ denote the indicator vector for component $i$, so that $(\1_i)_j=1$ iff $i=j$. In a slight abuse of notation (that will be made clear from context), for an event $E$, we also sometimes write $\1_E$ to be the indicator function for $E$ occuring, so that for instance $\1_{i=j}$ equals $1$ if $i=j$ and $0$ otherwise.

An $n$-qubit pure quantum state is specified by a vector $\ket{\psi}\in(\bC^2)^{\otimes n}=\bC^{2^n}$. We use the standard notation $\bC^2=\spn\{\ket{0},\ket{1}\}$ and $\ket{+}=(\ket{0}+\ket{1})/\sqrt{2}$. For a set $S\subseteq\bF_2^n$, we let $\ket{S}=(1/\sqrt{|S|})\sum_{x\in S}\ket{x}$ denote the uniform superposition of elements in $S$. A $n$-qubit density operator is a self-adjoint positive semi-definite operator $\rho\in\bC^{2^n\times 2^n}$ of trace $1$. We refer to linear maps $\cO:\bC^{2^n\times 2^n}\rightarrow\bC^{2^n\times 2^n}$ as \emph{superoperators}. A superoperator that is completely positive and trace-preserving (CPTP) is called a quantum channel. We let $\cI_n:\bC^{2^n\times 2^n}\rightarrow\bC^{2^n\times 2^n}$ denote the identity channel $\cI_n(\rho)=\rho$.

Letting $I,X,Y,Z\in\bC^{2\times 2}$ denote the standard single-qubit Pauli matrices, then an $n$-qubit Pauli matrix is a tensor product of $n$ single-qubit Pauli matrices. For $P\in\{I,X,Y,Z\}$ and $x\in\bF_2^n$, we write $P^x=\bigotimes_{i=1}^nP^{x_i}$. An $n$-qubit Pauli superoperator is a map of the form $\rho\mapsto A\rho B$ for $n$-qubit Paulis $A,B$.

We use ordinary big-$O$ notation $O(\cdot)$, $\Theta(\cdot)$, $\Omega(\cdot)$. Subscripts will be used to denote variables that are held fixed, so that for example $f(r,n)=O_r(g(n))$ if for every fixed $r$, there exists some $c(r)>0$ for which $f(r,n)\leq c(r)\cdot g(n)$.

\subsection{Classical and Quantum Codes}
This section presents standard definitions of classical and quantum codes.

\begin{definition}
  A \emph{classical (binary linear) code of length $n$} is a linear subspace $C\subseteq\bF_2^n$. The code's \emph{dimension} is $k=\dim(C)$, and the distance is $d=\min_{c\in C\setminus\{0\}}|c|$. We summarize these parameters by saying that $C$ is an $[n,k,d]$ code.

  The \emph{dual} of $C$ is the code
  \begin{equation*}
    C^\perp=\{x\in\bF_2^n:x\cdot c=0\;\forall c\in C\}.
  \end{equation*}

  A \emph{parity-check matrix} for $C$ is a matrix $H\in\bF_2^{m\times n}$ such that $C=\ker(H)$. We say that $C$ with associated parity-check matrix $H$ is \emph{LDPC of locality $w$} if every row and column of $H$ has $\leq w$ nonzero entries. We say a family of classical codes is \emph{LDPC} if they are all LDPC of constant locality $w=O(1)$.
\end{definition}

\begin{definition}
  A \emph{quantum (binary CSS) code of length $n$} is a pair of classical codes $Q=(Q_X,Q_Z)$. The code's \emph{dimension} is $k=\dim(Q_Z)-\dim(Q_X^\perp)$, and the distance is
  \begin{equation*}
    d=\min_{c\in(Q_X\setminus Q_Z^\perp)\cup(Q_Z\setminus Q_X^\perp)}|c|.
  \end{equation*}
  We summarize these parameters by saying that $Q$ is an $[[n,k,d]]$ code.

  We say that $Q$ is \emph{LDPC of locality $w$} if $Q_X,Q_Z$ are classical LDPC codes of locality $w$. We similarly say that a family of quantum codes is \emph{LDPC} if they are all LDPC of constant locality $w=O(1)$.
\end{definition}

\subsection{Classical Codes from Lossless Expanders}
\label{sec:lossless}

In this section, we describe constructions of classical codes from constant-degree lossless expander graphs. The quantum codes we study arise from tensor products of cochain complexes associated to these classical codes, and will inherit many desirable properties from these expander graphs, such as constant encoding rate and decodability from the expansion.

First, we define lossless expander graphs.

\begin{definition}
\label{def:losslessexpandergraph}
  Let $G=(V=V_L\sqcup V_R,E\subseteq V_L\times V_R)$ be a bipartite graph of maximum left and right degrees $\Delta_L$ and $\Delta_R$ respectively. For a set of vertices $S\subseteq V$, we let $N_G(S)\subseteq V$ denote the neighborhood of $S$, that is, the set of vertices sharing an edge with some element of $S$. We say $G$ is a \emph{$(\mu,\epsilon)$-lossless expander} if it holds for every $S\subseteq V_L$ of size $|S|\leq\mu|V_L|$ that $|N_G(S)|\geq(1-\epsilon)\Delta_L|S|$, and for every $S\subseteq V_R$ of size $|S|\leq\mu|V_R|$ that $|N_G(S)|\geq(1-\epsilon)\Delta_R|S|$.
\end{definition}

The following basic lemma shows that lossless expansion implies that every small set $S$ contains some vertex with many \emph{unique neighbors}, meaning many neighbors that avoid the neighborhoods of all other vertices in $S$.

\begin{lemma}
  \label{lem:losslesstoun}
  Let $G=(V=V_L\sqcup V_R,E)$ be a $(\mu,\epsilon)$-lossless expander of maximum left degree $\Delta_L$. Then for every $S\subseteq V_L$ of size $|S|\leq\mu|V_L|$, there exists some $v\in S$ such that
  \begin{equation}
    \label{eq:losslesstoun}
    |N_G(v)\setminus N_G(S\setminus\{v\})| \geq (1-2\epsilon)\Delta_L.
  \end{equation}
  (An analogous statement also holds for subsets $S\subseteq V_R$ of size $|S|\leq\mu|V_R|$.)
\end{lemma}
\begin{proof}
  Let $E'\subseteq E$ be a set of $|E'|=|N_G(S)|\geq(1-\epsilon)\Delta_L|S|$ edges such that for each $v\in N_G(S)$, $E'$ contains a single edge from a vertex in $S$ to $v$. As there are $\leq \Delta_L|S|$ outgoing edges from $S$, it follows that at most $\Delta_L|S|-|E'|\leq\epsilon \Delta_L|S|$ of these edges do not lie in $E'$, so at most $\epsilon \Delta_L|S|$ vertices in $N_G(S)$ have $\geq 2$ edges coming from $S$. Therefore there are at least $|N_G(S)|-\epsilon \Delta_L|S|\geq(1-2\epsilon)\Delta_L|S|$ vertices in $N_G(S)$ that are incident to exactly $1$ vertex in $S$, that is,
  \begin{equation*}
    \sum_{v\in S}|N_G(v)\setminus N_G(S\setminus\{v\})| \geq (1-2\epsilon)\Delta_L|S|.
  \end{equation*}
  Thus by the pigeonhole principle, there is some $v\in S$ satisfying
  \begin{equation*}
    |N_G(v)\setminus N_G(S\setminus\{v\})| \geq (1-2\epsilon)\Delta_L,
  \end{equation*}
  as desired.  
\end{proof}


The following result shows that there exist explicit constant-degree expander graphs.

\begin{proposition}[\cite{hsieh_explicit_2025}]
  \label{prop:expanders}
  For every fixed $\epsilon>0$ and $\beta_2>\beta_1>0$, there exists $\mu>0$ and $\Delta_L,\Delta_R\in\bN$ with $\beta_1<\Delta_R/\Delta_L<\beta_2$ and $\Delta_L,\Delta_R\geq 1/\epsilon$ such that there is an infinite explicit sequence $(G^i=(V^i=V^i_L,V^i_R,E^i))_{i\in\bN}$ of $(\Delta_L,\Delta_R)$-biregular $(\mu,\epsilon)$-lossless expanders, such that each $|V^i|<|V^{i+1}|\leq 2|V^i|$.\footnote{\cite{hsieh_explicit_2025} only stated their result with $|V^{i+1}|\leq O(|V^i|)$. However, the constant in the big-$O$ can without loss of generality be reduced to $2$ by considering disjoint unions of copies of the graphs, at the cost of a constant-factor reduction in $\mu$.}
\end{proposition}

\begin{remark}
\label{rmk:explicitexpanders}
  The explicitness condition in Proposition~\ref{prop:expanders} simply means that each graph $G^i$ can be constructed by a $\poly(|V^i|)$-time algorithm. If this condition is dropped, then it is well-known that random biregular graphs also give lossless expanders that satisfy the other conditions in Proposition~\ref{prop:expanders} (e.g.~\cite[Theorem 4.16]{hoory_expander_2006}).
\end{remark}

We construct classical LDPC codes by defining parity-check matrices from bipartite adjacency matrices of lossless expanders, as defined below.

\begin{definition}
\label{def:codefromgraph}
   For a bipartite graph $G=(V=V_L\sqcup V_R,E)$, the associated \emph{parity-check matrix} (also called the \emph{bipartite adjacency matrix}) $H_G\in\bF_2^{V_L\times V_R}$, is given by $(H_G)_{u,v}=\1_{(u,v)\in E}$.
\end{definition}

Our code switching gadgets will require identifying certain code components that contain the encoded logical information. The following lemma provides equivalent conditions for a set of bits in a classical code to contain such logical information.





\begin{lemma}
  \label{lem:extendable}
  For a code $C\subseteq\bF_2^n$ and a set $S\subseteq[n]$, the following are equivalent:
  \begin{enumerate}
  \item\label{it:ext} Every $x\in\bF_2^S$ can be extended to a (not necessarily unique) $\tilde{x}\in\bF_2^n$, so that $\tilde{x}|_S=x$.
  \item\label{it:dual} Every $y\in C^\perp\setminus\{0\}$ satisfies $\supp(y)\not\subseteq S$.
  \item\label{it:coset} The cosets $\1_v+C^\perp\in\bF_2^n/C^\perp$ for $v\in S$ are all linearly independent.
  \end{enumerate}
\end{lemma}
\begin{proof}
  Condition~\ref{it:ext} simply says that $C|_S=\bF_2^S$, which is in turn equivalent to the requirement that $C$ satisfies no nontrivial linear constraints supported inside $S$; this statement in turn is precisely Condition~\ref{it:dual}.

  Meanwhile, Condition~\ref{it:coset} says that every nonzero linear combination of elements $\1_v$ for $v\in S$ (i.e.~every nonzero vector supported inside $S$) lies outside of $C^\perp$, which is precisely Condition~\ref{it:dual}.
\end{proof}

\begin{definition}
  \label{def:extendable}
  A set $S\subseteq[n]$ is an \emph{extendable set} for a code $C\subseteq\bF_2^n$ if $S$ satisfies the equivalent conditions in Lemma~\ref{lem:extendable}. A maximal extendable set, that is an extendable set $S$ which becomes unextendable upon adding any addition element in $[n]\setminus S$, is called an \emph{information set}.
\end{definition}

Note that by Condition~\ref{it:coset} in Lemma~\ref{lem:extendable}, an information set for $C$ must have size $\dim(\bF_2^n/C^\perp)=\dim(C)$. That is, the bits in an information set contain the entire encoded message (up to some basis change); all bits outside of the information set are redundant bits used for error correction. An extendable set contains a subset of the bits in the message.

Below, we present a corollary of Proposition~\ref{prop:expanders}, in which we show that there exists lossless expanders that remain lossless expanders after removing the vertices associated to a linear-sized extendable set. This additional property will be crucial for our code switching gadget in Section~\ref{sec:switchdown} that switches from $r$-dimensional to $(r-1)$-dimensional product codes.

Specifically, to perform this downwards code switching, we measure out qubits that (in one of the $r$ directions) lie outside of a fixed extendable set. As our measurement outcomes may be noisy, we run them through our decoder described in Section~\ref{sec:ssflip}. However, this decoder requires lossless expansion, but we have only performed measurements on positions that lie outside of our fixed extendable set. Hence we need lossless expansion to be preserved under removing the extendable set.

Logical qubits outside of this extendable set are destroyed during the code switching, and hence should not be used to encode the message. Therefore we want this fixed extendable set to be linear-sized as in Corollary~\ref{cor:losslessfamily} below, so that we still have a constant encoding rate.

\begin{corollary}
  \label{cor:losslessfamily}
  For every fixed $0<\epsilon<1$, there exist $\mu,R>0$ and $\Delta_L,\Delta_R\in\bN$ with $1/4<\Delta_R/\Delta_L<1/2$ for which there is an infinite explicit sequence $(G^i=(V^i=V^i_L\sqcup V^i_R,E^i))_{i\in\bN}$ of $(\Delta_L,\Delta_R)$-biregular graphs satisfying:
  \begin{enumerate}
  \item $|V^i|<|V^{i+1}|\leq 2|V^i|$.
  \item $G^i$ is a $(\mu,\epsilon)$-lossless expander.
  \item There exists a set $S_i\subseteq V^i_R$ of size $|S_i|\geq R|V^i_R|$ that is extendable for the code $\ker(H_{G^i})$, such that $G^i$ remains a $(\mu,\epsilon)$-lossless expander if all vertices in $S$ are removed (along with any edges with a vertex in $S$).
  \end{enumerate}
\end{corollary}
\begin{proof}
  For a given value of $\epsilon>0$, we let $(G^i)_{i\in\bN}$ be a family of $(\Delta_L,\Delta_R)$-biregular $(\mu,\epsilon')$-lossless expanders from Proposition~\ref{prop:expanders} with expansion parameter $\epsilon'=\epsilon/2$, and with $\beta_1=1/4$ and $\beta_2=1/2$. 


  For a given $i\in\bN$, we define the extendable set $S_i$ as follows. Let $\ell=\Delta_L\Delta_R+1$, and choose a partition $V_R=U_1\sqcup\cdots\sqcup U_\ell$ such that no length-2 path in $G^i$ (i.e.~no edge in $(G^i)^2$) connects a pair of distinct vertices in the same set $U_j$. Because $(G^i)^2$ has degree $\Delta_L\Delta_R$, there exists such a partition of size $\ell=\Delta_L\Delta_R+1$.

  Let $T\subseteq V_R$ denote an information set for $\ker(H_{G^i})$, so that $|T|\geq\dim(\ker(H_{G^i}))\geq|V_R^i|-|V_L^i|\geq|V_R^i|/2$. By the pigeonhole principle, there exists some $j$ with $|U_j\cap T|\geq|T|/r\geq|V_R^i|/2(\Delta_L\Delta_R+1)$. We then let $S_i=U_j\cap T$, so that $|S_i|\geq R|V_R^i|$ for $R=1/2(\Delta_L\Delta_R+1)$.

  By definition $S_i\subseteq T$ is extendable for $\ker(H_{G^i})$. It remains to be shown that the graph $G^i\setminus S_i$ obtained by removing all vertices in $S_i$ from $G^i$ is a $(\mu,\epsilon)$-lossless expander. For right-to-left expansion, every set $W\subseteq V_R\setminus S_i$ of size $|W|\leq\mu|V_R^i\setminus S_i|\leq\mu|V_R^i|$ by definition satisfies
  \begin{equation*}
    |N_{G^i\setminus S_i}(W)| = |N_{G^i}(W)| \geq (1-\epsilon')\Delta_R|W| \geq (1-\epsilon)\Delta_R|W|.
  \end{equation*}
  For left-to-right expansion, consider a set $W\subseteq V_L$ of size $|W|\leq\mu|V_L^i|$. By construction, every $v\in W$ is incident to at most one element of $S_i\subseteq U_j$, so $|N_{G^i\setminus S_i}(W)|\geq|N_{G^i}(W)|-|W|$. Therefore
  \begin{equation*}
    |N_{G^i\setminus S_i}(W)| \geq (1-\epsilon')\Delta_L|W|-|W| = (1-\epsilon'-1/\Delta_L)\Delta_L|W| \geq (1-2\epsilon')\Delta_L|W|=(1-\epsilon)\Delta_L|W|,
  \end{equation*}
  where the second inequality above holds by the fact that $\Delta_L\geq 1/\epsilon'$ from Proposition~\ref{prop:expanders}. The two inequalities above imply that $G^i\setminus S_i$ is a $(\mu,\epsilon)$-lossless expander, as desired.
\end{proof}


\subsection{Chain Complexes}
This section describes the notion of (co)chain complexes, which provide a useful language for constructing quantum CSS codes.

\begin{definition}
  \label{def:chaincom}
  An \emph{$r$-dimensional chain complex $\cC_*$ over $\bF_2$} is a graded $\bF_2$-vector space $\cC=\bigoplus_{i=0}^r\cC_i$ together with a \emph{boundary map} $\partial:\cC\rightarrow\cC$ satisfying $\partial^2=0$ and $\partial(\cC_i)\subseteq\cC_{i-1}$. We write $\partial_i=\partial|_{\cC_i}$, and we summarize this data by writing
  \begin{equation*}
    \cC_* = (\cC_r\xrightarrow{\partial_r}\cC_{r-1}\xrightarrow{\partial_{r-1}}\cdots\xrightarrow{\partial_1}\cC_0).
  \end{equation*}
  We call $\cC_i$ the space of \emph{$i$-chains} of $\cC_*$, and we define the
  \begin{align*}
    i\text{-cycles} \; Z_i(\cC) &= \{z\in\cC_i:\partial(z)=0\} \\
    i\text{-boundaries} \; B_i(\cC) &= \{\partial(c):c\in\cC_{i+1}\} \\
    i\text{-homology} \; H_i(\cC) &= Z_i(\cC)/B_i(\cC).
  \end{align*}
  In a slight abuse of notation, for $i\in\bZ\setminus\{0,\dots,r\}$ we write $C_i=\emptyset$ and $\cC_i=\{0\}$.

  The \emph{cochain complex}
  \begin{equation*}
    \cC^* = (C^0\xrightarrow{\delta_0}C^1\xrightarrow{\delta_1}\cdots\xrightarrow{\delta_{r-1}}C^r)
  \end{equation*}
  associated to $\cC_*$ is the chain complex with the same vector space $\cC=\bigoplus_{i=0}^r\cC^i$ with each $\cC^i=\cC_i$, but whose boundary map is the \emph{coboundary map} $\delta:\cC\rightarrow\cC$ given by $\delta=\partial^\top$. We then write $\delta_i=\delta|_{\cC^i}=\partial_{i+1}^\top$. We call $\cC^i$ the space of \emph{$i$-cochains}, and we similarly define the
  \begin{align*}
    i\text{-cocycles} \; Z^i(\cC) &= \{z\in\cC^i:\delta(z)=0\} \\
    i\text{-coboundaries} \; B^i(\cC) &= \{\delta(c):c\in\cC_{i-1}\} \\
    i\text{-cohomology} \; H^i(\cC) &= Z^i(\cC)/B^i(\cC).
  \end{align*}

  We will assume our chain complexes are \emph{based}, meaning that $\cC_i=\bF_2^{C_i}$ for a specified set $C_i$. We then define Hamming weights of elements of $\cC_i=\cC^i$ with respect to this basis. For $c_{i-1}\in C_{i-1}$ and $c_i\in C_i$, we write $c_{i-1}\triangleleft c_i$ if $c_{i-1}\in\supp(\partial(\1_{c_i}))$. For $c_j\in C_j$ and $c_i\in C_i$, with $i<j$, we write $c_i\prec c_j$ if there exists a sequence $c_{i+1},\dots,c_{j-1}$ such that $c_i\triangleleft c_{i+1}\triangleleft\cdots\triangleleft c_{j-1}\triangleleft c_j$. This partial order provides the set $C=\bigsqcup_iC_i$ with the structure of a graded poset. For $i<j$, we extend the partial order notation to apply for a basis element $c_i\in C_i$ and a chain $c_j\in\cC_j=\bF_2^{C_j}$, so that $c_i\prec c_j$ if every $c'\in\supp(c_j)$ satisfies $c_i\prec c'$.


  The \emph{$i$-systolic distance $d_i(\cC)$} and \emph{$i$-cosystolic distance $d^i(\cC)$} are defined as
  \begin{align*}
    d_i(\cC) &= \min_{z\in Z_i(\cC)\setminus B_i(\cC)}|z| \\
    d^i(\cC) &= \min_{z\in Z^i(\cC)\setminus B^i(\cC)}|z|.
  \end{align*}

  We say that $\cC_*$ has \emph{locality $w$} if for every $c\in C$ there are $\leq w$ basis elements $c'\in C$ with $c'\preceq c$ or $c'\succeq c$.
\end{definition}

\begin{definition}
  For $r$-dimensional chain complexes $\cA_*,\cB_*$, a \emph{chain map $\phi:\cA_*\rightarrow\cB_*$} is a sequence of maps $(\phi_i:\cA_i\rightarrow\cB_i)_{i\in[r]}$ satisfying $\partial_i^{\cB}\circ\phi_i=\phi_{i-1}\circ\partial_i^{\cA}$. Every chain map naturally induces a map on homology $\tilde{\phi}:H_*(\cA)\rightarrow H_*(\cB)$.
\end{definition}

\begin{definition}
  \label{def:cctensor}
  For chain complexes $\cA_*,\cB_*$ of respective dimensions $r_{\cA},r_{\cB}$, the \emph{tensor product $\cC_*=\cA_*\otimes\cB_*$} is the chain complex of dimension $r_{\cC}=r_{\cA}+r_{\cB}$ for which the $i$-chain basis is $C_i=\bigsqcup_{j\in\bZ}A_j\times B_{i-j}$, so that $\cC_i=\bigoplus_{j\in\bZ}\cA_j\otimes\cB_{i-j}$, and the boundary map is given by $\partial^{\cC}=\partial^{\cA}\otimes I+I\otimes\partial^{\cB}$. The tensor product of cochain complexes is defined analogously, so that $\cA^*\otimes\cB^*=(\cA_*\otimes\cB_*)^*$.
\end{definition}

The following formula describing how homology behaves under tensor products is well known.

\begin{proposition}[K\"{u}nneth formula]
  \label{prop:kunneth}
  For chain complexes $\cA_*,\cB_*$ with tensor product $\cC_*=\cA_*\otimes\cB_*$, then there is an isomorphism
  \begin{align*}
    H_i(\cC) &\cong \bigoplus_{j\in\bZ}H_j(\cA)\otimes H_{i-j}(\cB),
  \end{align*}
  which for $a\in Z_i(\cA)$, $b\in Z_i(\cB)$ is given by
  \begin{align*}
    a\otimes b+B_i(\cC) &\mapsfrom (a+B_i(\cA))\otimes(b+B_i(\cB)).
  \end{align*}
\end{proposition}

For a product complex $\cC_*=\cC^{(1)}_*\otimes\cdots\otimes\cC^{(r)}_*$ with basis elements $c=(c_1,\dots,c_r),\;c'=(c_1',\dots,c_r')\in C^{(1)}\times\cdots\times C^{(r)}=C$, the graded poset structure of $C$ by definition has $c\preceq c'$ iff it holds for every $i\in[r]$ that $c_i\preceq c'_i$.

\begin{definition}
  \label{def:cctoCSS}
  The \emph{quantum CSS code associated to level $i$ of a cochain complex $\cC^*$} is given by $Q=(Q_X=\ker(\partial_i),\;Q_Z=\ker(\delta_i))$.
\end{definition}

The quantum code $Q$ in Definition~\ref{def:cctoCSS} by definition has length $n=\dim(\cC^i)$, dimension $k=\dim(H^i(\cC))$, and distance $d=\min\{d^i(\cC),d_i(\cC)\}$. Furthermore, $Q$ has naturally associated $X,Z$ parity-check matrices $\partial_i,\delta_i$ respectively, so if $\cC^*$ has locality $w$ then by definition $Q$ is LDPC of locality $w$.

\subsection{Fault-Tolerance Model}
\comment{for the purposes of reorg, I'm currently thinking of whether this should be in prelims or appendix depending on how many proofs that depend on these defns get moved to the appendix}

In this section, we describe the model of circuits, faults, and fault-tolerance that we use in this paper. Our model will largely follow that of \cite{nguyen_quantum_2025} (and the closely related \cite{he_composable_2025}). The most important distinction is that whereas \cite{nguyen_quantum_2025} consider noisy quantum circuits that are permitted to run a constant-depth noiseless classical circuit at each timestep, we allow a polynomial-depth noiseless classical circuit at each timestep. This polynomial-depth classical computation will ultimately be used for solving linear systems of equations needed in our constant-overhead state preparation and code switching gadgets.

We also introduce some notational differences compared to \cite{nguyen_quantum_2025}. For instance, as our focus in this paper is to provide certain new gadgets with reduced overhead rather than an end-to-end fault-tolerance scheme, we allow the inputs and outputs of circuits to be quantum states, rather than simply classical bit strings.

Below, we present our definition of a quantum circuit that may also perform (noisless) classical computation between timesteps.

\comment{might be worth splitting up the clifford v non clifford defn from the adaptive quantum circuit one (so one defn each)}

\begin{definition}
  \label{def:circuit}
  We define the \emph{Clifford gate set} to consist of the single-qubit untaries $I,X,Y,Z,H,S=\sqrt{Z}$, the 2-qubit unitary $CNOT$, the single-qubit Pauli $X,Y,Z$ measurement channels, and the single-qubit reset channels that reset a qubit to $\ket{0}$ or $\ket{+}$. For a non-Clifford gate $U$ such as the $3$-qubit unitary $U=CCZ$, the \emph{Clifford+$U$} gate set consists of the Clifford gates along with $U$.
  
  An \emph{adaptive quantum circuit $\cQ=(\cO,C)$ using quantum space $N$, classical space $M$, and time $T$} consists of a sequence $\cO=(\cO_1,\dots,\cO_T)$ of quantum channels acting on $N$ qubits, and a sequence $C=(C_1,\dots,C_T)$ of classical circuits acting on a classical register of some number $M\geq N$ of bits. Formally, we represent this classical register by an $M$-qubit quantum register that always remains in (a mixture of) computational basis states. If for each $t\in T$, the channel $\cO_t=\bigotimes_{i=1}^{g_t}\cO_{t,i}$ is a product of Clifford (resp.~Clifford+$U$) gates $\cO_{t,i}$ acting on disjoint sets of qubits, we say $\cQ$ is an adaptive \emph{Clifford (resp.~Clifford+$U$)} circuit.

  We now define an $(N+M)$-qubit superoperator $\cQ(\cdot)$ associated to $\cQ$. Given an $N$-qubit quantum input state $\rho\in\bC^{2^N\times 2^N}$ and an $M$-bit classical input state $\ket{x}\bra{x}\in\bC^{2^M\times 2^M}$ for some $x\in\bF_2^M$, we perform the following operators on our $N+M$ qubits for each $t=1,\dots,T$:
  \begin{enumerate}
  \item Apply $C_t$ to the current classical state $\ket{x}\bra{x}$.
  \item For every $i\in[g_t]$ such that either $x_i=1$ or $\cO_{t,i}$ is a single-qubit Pauli measurement, apply $\cO_{t,i}$ to the (appropriate qubits of the) current quantum state $\rho$. If $\cO_{t,i}$ is a measurement, update $x_i\in\bF_2$ to equal the measurement outcome.
  \end{enumerate}
  The output $\cQ(\rho,\ket{x}\bra{x})\in\bC^{2^{N+M}\times 2^{N+M}}$ is then defined to be the resulting $N+M$ qubit state. Note that in the construction above, the entire $N+M$ qubit state is always a mixture of states $\ket{\phi}\bra{\psi}$ whose partial trace down to the $M$-(qu)bit classical register is a computational basis state $\ket{x}\bra{x}$. Furthermore, if every $\cO_{t,i}$ is a channel, then $\cQ(\cdot)$ is a channel.
\end{definition}

\begin{remark}
  In all of the constructions in this paper, the number $M$ of classical bits and the size of the classical circuits $C_t$ will grow at most polynomially in the number of qubits $N$.
\end{remark}

\begin{remark}
  \label{remark:noclassical}
  For most of the adaptive quantum circuits considered in this paper, the classical input and output are not used, and rather the classical state $x$ is simply reset and then used as a ``scratchpad'' to perform some intermediate computation. In such cases, as a shorthand we often leave out the classical inputs and outputs, and for instance write $\cQ(\rho)\in\bC^{2^N\times 2^N}$ to denote the result $\tr_{[M]}(\cQ(\rho\otimes\ket{0^M}\bra{0^M}))$ from executing $\cQ$ with an all-$0$s classical input, and then discarding (i.e.~tracing out) the classical output. If the classical input (resp.~output) is nonempty but contains fewer than $M$ bits, we just assume the extra bits are initialized to $0$ (resp.~traced out). When constructing gadgets throughout the paper, we will explicitly state whenever we are not using this shorthand, and have meaningful classical inputs or outputs.
\end{remark}

\begin{definition}
  \label{def:fault}
  For an adaptive quantum circuit $\cQ=(\cO,C)$ using quantum space $N$ and time $T$, a \emph{fault $\cF$} is a sequence $\cF=(\cF_1,\dots,\cF_T)$ of $N$-qubit superoperators $\cF_t$. We say $\cF$ is a \emph{Pauli fault} if each $\cF_t$ is a Pauli superoperator, that is, $\cF_t(\rho)=F_t\rho F_t'$ for some $N$-qubit Paulis $F_t,F_t'$. We say such a Pauli fault is \emph{diagonal} if every $F_t=F_t'$.

  The \emph{$\cF$-corrupted circuit $\cQ[\cF]$} is the adaptive quantum circuit using quantum space $N$ and time $2T$ given by
  \begin{equation*}
    \cQ[\cF] = ((\cO_1,\cF_1,\cO_2,\cF_2,\dots,\cO_T,\cF_T),(C_1,I,C_2,I,\dots,C_t,I)),
  \end{equation*}
  where $I$ denotes the identity (i.e.~idling) circuit.\footnote{As a slight abuse of notation, when considering the superoperator $\cQ[\cF](\cdot)$ associated to the $\cF$-corruped circuit $\cQ[\cF]$, we assume that the $M$-(qu)bit classical register $x$ in Definition~\ref{def:circuit} is never updated in timesteps associated to corruptions (i.e.~even timesteps), even if some $\cF_t$ performs a single-qubit Pauli measurement. This convention is useful so that we can decompose arbitrary faults into superpositions of Pauli faults without changing the behavior of the associated superoperator. Note that this convention does not meaningfully change the error model, as the adversary can still corrupt a measurement outcome by corrupting the appropriate qubit in the timestep before the measurement.}

  The \emph{fault path} or \emph{support} $\supp(\cF)\subseteq[N]\times[T]$ of $\cF$ is the set of all space-time locations on which $\cF$ acts nontrivially, that is, $\supp(\cF)=\bigsqcup_{t\in[T]}\supp(\cF_t)\times\{t\}$.
\end{definition}

\begin{remark}
  One could consider a more general error model where an adversary has a private register that can perform entangled corruptions across fault locations. For simplicity, we do not pursue this direction in this paper, but refer the reader to \cite{he_composable_2025} for more details.
\end{remark}

Our primary aim is to provide gadgets that exhibit a threshold under a locally stochastic fault distribution:

\begin{definition}
  For an adaptive quantum circuit using quantum space $N$ and time $T$, a probability distribution over faults $\cF$ is \emph{$\epsilon$-locally stochastic} if for every $S\subseteq[N]\times[T]$,
  \begin{equation*}
    \Pr[S\subseteq\supp(\cF)] \leq \epsilon^{|S|}.
  \end{equation*}
\end{definition}

\begin{definition}
  For a family of sets $\cE\subseteq 2^{[n]}$, we say that a set $S\subseteq[n]$ is \emph{$\cE$-avoiding} if no element of $\cE$ is contained in $S$.

  We extend this definition to faults in the natural way, so that a fault $\cF$ on a circuit using quantum space $\leq N$ and time $\leq T$ is \emph{$\cE$-avoiding} for $\cE\subseteq 2^{[N]\times[T]}$ if no element of $\cE$ is contained in $\supp(\cF)$.

  For families $\cE_1\subseteq 2^{[n_1]}$, $\cE_2\subseteq 2^{[n_2]}$, we let $\cE_1\sqcup\cE_2\subseteq 2^{[n_1]\sqcup[n_2]}$ denote the disjoint union of $\cE_1$ and $\cE_2$ on separate sets of underlying elements. Therefore for $\cE\subseteq 2^{[n]}$ and $t\in\bN$, we write $\cE^{\sqcup t}\subseteq 2^{[n]\times[t]}$ to denote $\cE\sqcup\cdots\sqcup\cE$ (with $t$ copies of $\cE$).
\end{definition}

\begin{definition}
  \label{def:deccode}
  Let $Q$ be an $[[n,k]]$ quantum code with associated encoding isometry $\Enc:(\bC^2)^{\otimes k}\rightarrow(\bC^2)^{\otimes n}$, which we also view as a channel in the natural way. For $n'\geq n$, let $\cE\subseteq 2^{[n']}$ be a family of subsets of qubits, which we call \emph{bad sets for $Q$}. For a density operator $\rho\in\bC^{2^k\times 2^k}$, and for $N\geq n$, we say an operator $\sigma\in\bC^{2^N\times 2^N}$ is a \emph{(Pauli) $\cE$-deviation of $\Enc(\rho)$} if there exists a (Pauli) superoperator $\cF:\bC^{2^n\times 2^n}\rightarrow\bC^{2^n\times 2^n}$ whose support is $\cE$-avoiding such that $\tr_{[N]\setminus[n]}(\sigma)\propto\cF\circ\Enc(\rho)$.

  We say that the data $D=(Q,\Enc,\cE)$ forms a \emph{decorated (quantum) code}. We let $\emptyset$ denote the \emph{trivial decorated code} on $n=0$ qubits.

  For decorated codes $D=(Q,\Enc,\cE)$ and $D'=(Q',\Enc',\cE')$, we write
  \begin{equation*}
    D\sqcup D'=(Q\sqcup Q',\;\Enc\sqcup\Enc',\;\cE\sqcup\cE')
  \end{equation*}
  to denote the decorated code consisting of disjoint blocks of $D$ and $D'$.
\end{definition}

\begin{definition}
  For a 2-register state $\rho\in\bC^{2^{k+M}\times 2^{k+M}}$, we say that the second (size-$M$) register is a \emph{computational basis state} if $\rho$ is invariant under measuring the second register in the computational (Pauli $Z$) basis. Equivalently, the second register is a computational basis state if there exists a decomposition
  \begin{equation*}
    \rho = \sum_{x\in\bF_2^M}\rho_x\otimes\ket{x}\bra{x}.
  \end{equation*}
\end{definition}

\begin{definition}
  \label{def:gadget}
  For $\alpha\in\{\mathrm{in},\mathrm{out}\}$, let $D_\alpha=(Q_\alpha,\Enc_\alpha,\cE_\alpha)$ be a decorated $[[n_\alpha,k_\alpha]]$ CSS code, and let $M_\alpha\in\bZ_{\geq 0}$. Let $\bar{\cO}$ be a quantum channel with $k_{\mathrm{in}}+M_{\mathrm{in}}$ input qubits and $k_{\mathrm{out}}+M_{\mathrm{out}}$ output qubits, such that if the $M_{\mathrm{in}}$-qubit register of the input is a computational basis state, then so is the $M_{\mathrm{out}}$-qubit register of the output. Let $\cQ$ be an adaptive quantum circuit using quantum space $N\geq n_{\mathrm{in}},n_{\mathrm{out}}$, classical space $M\geq M_{\mathrm{in}},M_{\mathrm{out}}$, and time $T$ with an associated family of \emph{bad sets} $\cE_{\mathrm{run}}\subseteq 2^{[N]\times[T]}$.

  We say that the data $((\cQ,\cE_{\mathrm{run}}),D_{\mathrm{in}},D_{\mathrm{out}})$ provides a \emph{fault-tolerant gadget for $\bar{\cO}$} if for every $\ell\in\bN$, every density operator $\rho\in\bC^{2^{k+M_{\mathrm{in}}+\ell}\times 2^{k+M_{\mathrm{in}}+\ell}}$ for which the $M_{\mathrm{in}}$-qubit register is a computational basis state, every $\cE_{\mathrm{in}}$-deviation $\sigma\in\bC^{2^{N+M+\ell}\times 2^{N+M+\ell}}$ of $\Enc_{\mathrm{in}}\otimes\cI_{M_{\mathrm{in}}+\ell}(\rho)$ for which the $M$-qubit register is a computational basis state, and every $\cE_{\mathrm{run}}$-avoiding fault $\cF$ for $\cQ$, then $\cQ[\cF]\otimes\cI_\ell(\sigma)\in\bC^{2^{N+M+\ell}\times 2^{N+M+\ell}}$ is a $\cE_{\mathrm{out}}$-deviation of $(\Enc_{\mathrm{out}}\circ\bar{\cO})\otimes\cI_\ell(\rho)$.

  The gadget is said to be \emph{Pauli fault-tolerant} if the same statement holds when restricting attention to Pauli deviations and Pauli faults.

  Similarly as in Remark~\ref{remark:noclassical}, unless explicitly stated otherwise, we by default assume that there is no classical input/output, that is, $M_{\mathrm{in}}=M_{\mathrm{out}}=0$.
\end{definition}

\begin{remark}
  In Definition~\ref{def:deccode}, we allow $N$-qubit deviations of $n$-qubit code states for $N\geq n$, so that a gadget may use a greater number $N$ of qubits than taken by its input or output. For notational convenience, we also allow families of bad sets to be defined on a superset of the physical qubits; bad sets containing any extraneous elements outside of the physical qubits can be ignored.
\end{remark}

\begin{remark}
  \label{remark:refsys}
  The $\ell$-qubit ancilla system in Definition~\ref{def:gadget} that is never encoded can be similarly found in \cite[Definition 2.25]{nguyen_quantum_2025}, where it is called a \emph{reference system}. While this reference system does not meaningfully change the analysis of our gadgets, \cite{nguyen_quantum_2025} use the reference system to show that the parallel execution of multiple fault-tolerant gadgets remains fault-tolerant; see Proposition~\ref{prop:parcomp} below.
\end{remark}

\begin{proposition}[Proposition~2.30 of \cite{nguyen_quantum_2025}]
  \label{prop:parcomp}
  For (Pauli) fault-tolerant gadgets $((\cQ,\cE_{\mathrm{run}}),D_{\mathrm{in}},D_{\mathrm{out}})$, $((\cQ',\cE_{\mathrm{run}}'),D_{\mathrm{in}}',D_{\mathrm{out}}')$ for $\bar{\cO}$, $\bar{\cO}'$ respectively, then the parallel composition
  \begin{equation*}
    ((\cQ\sqcup\cQ',\; \cE_{\mathrm{run}}\sqcup\cE_{\mathrm{run}}'),\; D_{\mathrm{in}}\sqcup D_{\mathrm{in}}',\; D_{\mathrm{out}}\sqcup D_{\mathrm{out}}')
  \end{equation*}
  is a (Pauli) fault-tolerant gadget for $\bar{\cO}\sqcup\bar{\cO}'=\bar{\cO}\otimes\bar{\cO}'$.
\end{proposition}


A standard technique for proving that a computation is fault-tolerant is to prove that it is Pauli fault-tolerant, and then decompose a more general fault into a linear combination Pauli faults. Each term in this linear combination can then be analyzed using the Pauli fault-tolerance. For instance, the following lemma is (implicitly) shown by \cite{nguyen_quantum_2025} (see the subsection ``Fault analysis and reduction to Pauli noise'' of the proof of \cite[Theorem~3.20]{nguyen_quantum_2025}).

\begin{lemma}[Folklore; see e.g.~\cite{nguyen_quantum_2025}]
  \label{lem:redPauli}
  Let $((\cQ,\cE_{\mathrm{run}}),D_{\mathrm{in}},D_{\mathrm{out}})$ be a Pauli fault-tolerant gadget for a channel $\bar{\cO}$ with entirely classical input and output, so that $k_{\mathrm{in}}=k_{\mathrm{in}}=0$ (see Definition~\ref{def:gadget}). Then $((\cQ,\cE_{\mathrm{run}}),D_{\mathrm{in}},D_{\mathrm{out}})$ is also a fault-tolerant gadget for $\bar{\cO}$.
\end{lemma}


For simplicity, we have stated Lemma~\ref{lem:redPauli} for the case of gadgets with classical input and output. As we typically think of the ultimate circuits being run on quantum computers as having classical inputs and outputs, we may compile our smaller fault-tolerant gadgets into such end-to-end algorithms, and then apply Lemma~\ref{lem:redPauli} to the gadget corresponding to the entire algorithm. Furthermore, similar results as Lemma~\ref{lem:redPauli} can also be proven for quantum inputs/outputs, though additional assumptions may be necessary, such as needing the decorated code $D_{\mathrm{out}}$ to be able to correct all errors that avoid its specified bad sets.

Many of the above definitions closely mirror those in \cite{nguyen_quantum_2025,he_composable_2025}; the reader is referred to the discussion there for more context for these definitions.

The following family of bad sets will be particularly useful.

\begin{definition}
  \label{def:graphtobad}
  For a graph $G=(V,E)$ and real numbers $\eta,\gamma>0$, define a family $\cE(G,\eta,\gamma)\subseteq 2^V$ by
  \begin{equation*}
    \cE(G,\gamma,\eta) = \{S\subseteq V:\text{$S$ lies in a connected subgraph }G'\subseteq G\text{ with }|V(G')|\geq\eta\text{ and }|S|/|V(G')|\geq\gamma\}.
  \end{equation*}
\end{definition}

Lemma~\ref{lem:threshold} below, which (in a similar form) is for instance shown in \cite{kovalev_fault_2013,gottesman_fault-tolerant_2014}, implies that if $G$ has constant $O(1)$ maximum degree, $\eta\geq\poly(|V|)$, and $\gamma=\Omega(1)$, then for every sufficiently small constant $\epsilon>0$, an $\epsilon$-locally stochastic error will contain a bad set in $\cE(G,\gamma,\eta)$ with exponentially small probability $\leq e^{-\poly(|V|)}$.

Throughout this paper, we prove that our gadgets are Pauli fault-tolerant with respect to such families $\cE(G,\gamma,\eta)$ of bad sets. Then Lemma~\ref{lem:redPauli} and Lemma~\ref{lem:threshold} together imply that algorithms constructed using our gadgets exhibit a threshold under locally stochastic noise.

\begin{lemma}[Well known; see e.g.~\cite{kovalev_fault_2013,gottesman_fault-tolerant_2014}]
  \label{lem:threshold}
  Let $G=(V,E)$ be a graph of maximum degree $\Delta$, and let $\eta,\gamma,\epsilon>0$ be real numbers with $\epsilon$ is sufficient small such that\footnote{Here $e=2.718\dots$ denotes Euler's number.} $\epsilon\leq\gamma/e$ and $\Delta^2\cdot(e\epsilon/\gamma)^\gamma\leq 1/2$. Let $F$ be a random subset $F\subseteq V$ whose distribution satisfies $\Pr[S\subseteq F]\leq\epsilon^{|S|}$ for every $S\subseteq V$. Then
  \begin{align*}
    \Pr[F\text{ is not }\cE(G,\eta,\gamma)\text{-avoiding}] = \Pr[\exists S\in\cE(G,\eta,\gamma)\text{ with }S\subseteq F]
    &\leq |V| \cdot 2\left(\frac{e\epsilon}{\gamma}\right)^{\eta}.
  \end{align*}
\end{lemma}

We include a brief proof of Lemma~\ref{lem:threshold} for completeness.

\begin{proof}[Proof of Lemma~\ref{lem:threshold}]
  For $u\in\bN$, let $\cV_u\subseteq 2^V$ be the collection of sets $V'\subseteq V$ of size $|V'|=u$ that induce a connected subgraph of $G$. Then
  \begin{equation}
    \label{eq:cVubound}
    |\cV_u| \leq |V| \cdot \Delta^{2u}.
  \end{equation}
  Specifically, every $V'\in\cV_u$ contains a $u$-vertex spanning tree. In turn, every $u$-vertex subtree of $G$ can be traversed by starting at one of the $|V|$ vertices in $G$, and then taking a length-$2u$ walk in $G$, at each step of which there are $\Delta$ choices of a neighbor to step to. Thus there are at most $|V|\cdot\Delta^{2u}$ such subtrees, and therefore~(\ref{eq:cVubound}) holds.

  For every $V'\in\cV_u$,
  \begin{align*}
    \Pr[|V'\cap F|\geq\gamma u]
    &\leq {u \choose \lceil\gamma u\rceil}\cdot\epsilon^{\lceil\gamma u\rceil} \leq \left(\frac{eu}{\lceil\gamma u\rceil}\right)^{\lceil\gamma u\rceil}\cdot\epsilon^{\lceil\gamma u\rceil} \leq \left(\frac{e\epsilon}{\gamma}\right)^{\gamma u},
  \end{align*}
  where the final inequality above holds by the assumption that $\epsilon\leq\gamma/e$, so that $(e\epsilon/\gamma)^{\lceil\gamma u\rceil}\leq(e\epsilon/\gamma)^{\gamma u}$. 
  
  Then union-bounding over every $V'\in\cV_u$ and over every $u\geq\eta$ gives that
  \begin{align*}
    \Pr[\exists S\in\cE(G,\eta,\gamma)\text{ with }S\subseteq F]
    &\leq \sum_{u=\lceil\eta\rceil}^{|V|}\sum_{V'\in\cV_u}\Pr[|V'\cap F|\geq\gamma u] \\
    &\leq \sum_{u=\lceil\eta\rceil}^{|V|}|V|\cdot\Delta^{2u} \cdot \left(\frac{e\epsilon}{\gamma}\right)^{\gamma u} \\
    &\leq |V| \cdot 2\left(\frac{e\epsilon}{\gamma}\right)^{\eta},
  \end{align*}
  where the third inequality above holds by the assumption that $\Delta^2\cdot(e\epsilon/\gamma)^\gamma\leq 1/2$.
\end{proof}

\comment{I personally think that def~\ref{def:cctodec} and def.~\ref{lem:enc1dim}and def~\ref{def:encrdim} belong up after Kunneth's formula. Or somewhere in preliminaries after all the graph stuff. The fault tolerance definitions kind of break the flow, and I think its reasonable to want to see how the graphs construct classical+quantum codes before moving to fault tolerance}

We now describe how we decorate the quantum codes underlying Theorem~\ref{thm:main}, which come from products $\cC^*$ of 1-dimensional cochain complexes arising from bipartite graphs. For this purpose, we need to specify both an encoding map and a graph $G$ with which to define bad sets via Definition~\ref{def:graphtobad}. As defined below, the encoding map will arise naturally from encoding maps for the underlying classical codes. Meanwhile, we choose the graph $G$ to describe the incidence structure of $\cC^*$.


\begin{definition}[Encoding map and connectivity graph from cochain complex]
  \label{def:cctodec}
  Let $\cC^*$ be an $r$-dimensional cochain complex. We say a \emph{CSS encoding map} for the $[[n,k]]$ CSS code $Q$ associated to level $i$ of $\cC^*$ is an isomorphism $\Enc:\bF_2^k\xrightarrow{\sim}H^i(\cC)$. Such an encoding map naturally induces an encoding isometry $\Enc:(\bC^2)^{\otimes k}\rightarrow(\bC^2)^{\otimes n}$ given by $\Enc\ket{x}=\ket{\Enc(x)}$, which we may also view as a quantum channel.

  For $I\subseteq[r]$, we define a \emph{level-$I$ connectivity graph $G^{\cC}_I$ for $\cC^*$} to have vertex set $\bigsqcup_{i\in I}C^i\subseteq C$, and have an edge connecting every two vertices $c,c'\in C$ for which either there exists some $c^0\in C^0$ with $c,c'\succeq c^0$, or there exists some $c^r\in C^r$ with $c,c'\preceq c^r$.
\end{definition}

Given an $r$-dimensional cochain complex $\cC^*$ with a level $i\in[r]$ and real numbers $\eta,\gamma>0$, Definition~\ref{def:cctodec} then gives an associated decorated code
\begin{equation*}
  (Q,\Enc,\cE(G^{\cC}_i,\eta,\gamma)).
\end{equation*}
In our fault-tolerance analysis, we will often use such decorated codes, but with $G^{\cC}_i$ replaced by some larger graph that contains $G^{\cC}_i$ as a subgraph.

We now provide more details on how we construct CSS encoding maps for cochain complexes given by tensor products of 1-dimensional complexes. Note that there is a bijection between 1-dimensional cochain complexes and bipartite graphs, where the coboundary map of the complex corresponds to the parity-check/bipartite adjacency matrix of the graph.

\begin{lemma}
  \label{lem:enc1dim}
  Let
  \begin{equation*}
    \cC^* = \left(\bF_2^{V_L}\xrightarrow{\delta^{\cC}=H_G}\bF_2^{V_R}\right)
  \end{equation*}
  be a 1-dimensional cochain complex associated to a bipartite graph $G=(V_L\sqcup V_R,E)$. Let $M^0\subseteq V_L$ and $M^1\subseteq V_R$ be information sets for $\ker(\delta^{\cC})$ and $\ker(\partial^{\cC})$ respectively. Define a cochain complex
  \begin{equation*}
    \cM^* = (\bF_2^{M^0}\xrightarrow{\delta^{\cM}}\bF_2^{M^1})
  \end{equation*}
  with coboundary map $\delta^{\cM}=0$. Define the cochain map $\Enc:\cM^*\rightarrow\cC^*$ such that for $x\in M^0$, then $\Enc^0(x)$ equals the unique codeword $\tilde{x}\in\ker(\delta)$ whose restriction $\tilde{x}|_{S_L}=x$, and for $x\in M^1$, then $\Enc^1(x)=(x,0^{V_R\setminus M^1})$. Then $\Enc$ induces an isomorphism on cohomology.
\end{lemma}
\begin{proof}
  Note that $\Enc^0(x)$ is well-defined and unique by the definition of an information set. Then $\Enc$ is a well-defined chain map because for every $x\in\bF_2^{M^0}$, by definition $\delta^{\cC}(\Enc^0(x))=0=\Enc^1(\delta^{\cM}(x))$. Now $\Enc^0$ induces an isomorphism on $0$-cohomology by item~\ref{it:ext} of Lemma~\ref{lem:extendable}, while $\Enc^1$ induces an isomorphism on $1$-cohomology by item~\ref{it:coset} of Lemma~\ref{lem:extendable}.
\end{proof}

\begin{definition}[Product encoding map]
  \label{def:encrdim}
  For $r\in\bN$ and $h\in[r]$, let
  \begin{equation*}
    {\cC^{(h)}}^* = \left(\bF_2^{V_L^{(h)}}\xrightarrow{\delta^{(h)}=H_{G^{(h)}}^\top}\bF_2^{V_R^{(h)}}\right)
  \end{equation*}
  be a 1-dimensional cochain complex associated to a bipartite graph $G^{(h)}=(V_L^{(h)}\sqcup V_R^{(h)},E^{(h)})$. For $h\in[r]$, let ${M^{(h)}}^0\subseteq V_L^{(h)}$ and ${M^{(h)}}^1\subseteq V_R^{(h)}$ be information sets for $\ker(\delta^{(h)})$ and $\ker(\partial^{(h)})$ respectively, and let $\cM^{(h)}$ and $\Enc^{(h)}$ be the associated cochain complex and cochain map respectively defined in Lemma~\ref{lem:enc1dim}.

  Let $\cC^*=\cC^{(1)}\otimes\cdots\otimes\cC^{(r)}$, $\cM^*=\cM^{(1)}\otimes\cdots\otimes\cM^{(r)}$, and $\Enc=\Enc^{(1)}\otimes\cdots\otimes\Enc^{(r)}$. For $i\in[r]$, by Lemma~\ref{lem:enc1dim} along with the K\"{u}nneth formula (Proposition~\ref{prop:kunneth}), $\Enc$ induces an isomorphism $\widetilde{\Enc}^i:\cM^i=H^i(\cM)\xrightarrow{\sim}H^i(\cC)$ on cohomology. We call this isomorphism a \emph{product encoding map}.



\end{definition}

\subsection{Basic Subroutines}

\begin{lemma}[Vizing's theorem (e.g.~\cite{misra_constructive_1992})]
  \label{lem:vizing}
  For every graph $G=(V,E)$ of maximum degree $\Delta$, there exists a $O(|V|\cdot|E|)$-time algorithm to compute a coloring of the edges with $\Delta+1$ colors such that two edges sharing a vertex have the same color.
\end{lemma}

The following standard corollary applies Vizing's theorem to obtain a low-depth syndrome extraction circuit for a CSS code, as given in Algorithm~\ref{alg:syndext}.

\begin{algorithm}
  \caption{\label{alg:syndext} Syndrome extraction circuit in Corollary~\ref{cor:syndext}. Note that by Lemma~\ref{lem:vizing}, each step of the outer loop over $t$ in line~\ref{li:seouter} only performs gates with disjoint supports, and hence can be implemented entirely within timestep $t$. We call the $n$ qubits in the input $\rho$ \emph{data qubits}.}
  \SetKwInOut{Input}{Input}
  \SetKwInOut{Output}{Output}

  \SetKwFunction{FnSyndExt}{SyndExt}
  \SetKwProg{Fn}{Function}{:}{}

  \Input{$n$-qubit state $\rho$, Pauli $P\in\{X,Z\}$, parity-check matrix $H\in\bF_2^{m\times n}$}
  \Output{Outcomes of measuring $P^{H_i}$ on $\rho$ for $i\in[m]$, along with post-measurement state}

  \Fn{\FnSyndExt{$\rho;P,H$}}{
    Initialize $m$ ancilla qubits to $\ket{0}$ (if $P=Z$) or $\ket{+}$ (if $P=X$) \\
    For the bipartite graph $G=([m]\sqcup[n],E)$ given by the parity-check matrix $H=H_G$, compute a coloring $\chi:E\rightarrow\{2,\dots,w+3\}$ using Lemma~\ref{lem:vizing}. \\
    \ForEach{$t=2,\dots,w+3$}{ \label{li:seouter}
      \ForEach{$(i,j)\in E:\chi(i,j)=t$}{
        \If{$P=X$}{
          Apply $CNOT$ to ancilla qubit $i$ and data qubit $j$
        }
        \If{$P=Z$}{
          Apply $CNOT$ to data qubit $j$ and ancilla qubit $i$
        }
      }
    }
    Measure Pauli $P$ on all $m$ ancilla qubits
  }
\end{algorithm}

\begin{corollary}[Well known]
  \label{cor:syndext}
  Let $H\in\bF_2^{m\times n}$ be a parity-check matrix of locality $w$, and for $i\in[m]$ let $H_i$ denote the $i$th row of $H$. Then for $P\in\{X,Z\}$, there exists a Clifford circuit using quantum space $N=n+m$ and time $T=w+4$ that takes as input a $n$-qubit state $\rho$, and outputs $\rho$ following measurements of the $n$-qubit Pauli operators $P^{H_i}$ for $i\in[m]$, along with the $m$ measurement outcomes.
\end{corollary}

\section{Small-Set Flip Decoding of Product Codes}
\label{sec:ssflip}
In this section, we construct a small-set flip decoder for product codes from lossless expanders. This decoder is a key ingredient in many of our fault-tolerance gadgets.

\begin{definition}
  \label{def:ssflip}
  For $m\in\bN$, we say a cochain complex $\cC^*$ has a \emph{$m$-small-set error-flip decoder at level $i$} if for every nonzero $e\in\cC^i$ of weight $|e|\leq m$, there exists a basis element $c^0\in C^0$ for which at least one of the following holds:
  \begin{enumerate}
  \item\label{it:notmin} There exists a cochain $c^{i-1}\in\cC^{i-1}$ with $c^0\preceq c^{i-1}$ such that $|e+\delta(c^{i-1})|<|e|$, or
  \item\label{it:flip} There exists a cochain $c^i\in\cC^i$ with $c^0\preceq c^i$ such that $|\delta(e+c^i)|<|\delta(e)|$.
  \end{enumerate}
  For $0<\nu<1$, we say $\cC^*$ has a \emph{$(m,\nu)$-small-set error-flip decoder at level $i$} if in condition~\ref{it:flip} above, we impose the stronger requirement that $|\delta(e+c^i)|<|\delta(e)|-(1-\nu)|\delta(c^i)|$.

  We say $\cC^*$ has a \emph{$(m,\gamma)$-small-set syndrome-flip decoder at level $i$} if for every nonzero $e\in\cC^i$ of weight $|e|\leq m$ and every $f\in\cC^{i+1}$ of weight $|f|\leq \gamma|\delta(e)|$, there exists a basis element $c^0\in C^0$ and a cochain $c^i\in\cC^i$ with $c^0\preceq c^i$ such that $|\delta(e+c^i)+f|<|\delta(e)+f|$.

  We say that $\cC^*$ has a \emph{$(m,\gamma)$-small-set flip decoder at level $i$} if $\cC^*$ has both a $m$-small-set error-flip decoder and a $(m,\gamma)$-small-set syndrome-flip decoder at level $i$.
\end{definition}

In Definition~\ref{def:ssflip}, an error-flip decoder is given access to a corrupted codeword and hence can compute a perfect syndrome. However, a syndrome-flip decoder is only given access to a syndrome, which may be corrupted in $|f|\leq\gamma|\delta(e)|$ locations. In Algorithm~\ref{alg:ssflip}, we translate these abstract notions of decoders into concrete algorithms.

\begin{remark}
  \label{remark:gdl1}
  A $(m,\gamma)$-small-set syndrome-flip decoder must have $\gamma<1$, as otherwise setting $f=\delta(e)$ for arbitrary $e\neq 0$ would yield $\delta(e)+f=0$, whose weight cannot be reduced.
\end{remark}


\begin{algorithm}
  \caption{\label{alg:ssflip} (Classical) decoding algorithms associated to a small-set flip decoder from Definition~\ref{def:ssflip}.}
  \SetKwInOut{Input}{Input}
  \SetKwInOut{Output}{Output}
  \SetKwProg{Fn}{Function}{:}{}

  \SetKwFunction{FnSSFlipSyn}{SSFlipSyn}

  \Input{$(i+1)$-cochain $s\in\cC^{i+1}$ of a cochain complex $\cC^*$}
  \Output{$a^i\in\cC^i$ with $\delta(a^i)$ close to $s$}
  \Fn{\FnSSFlipSyn{$s;i,\cC^*$}}{
    Initialize $a^i\gets 0\in\cC^i$ \\
    \While{$\exists c^0\in\cC^0,c^i\in\cC^i$ with $c^0\preceq c^i$ and $|s+\delta(a^i+c^i)|<|s+\delta(a^i)|$}{
      $a^i\gets a^i+c^i$
    }
    \Return{$a^i$}
  }
  
  \SetKwFunction{FnSSFlipErr}{SSFlipErr}

  \Input{$i$-cochain $e\in\cC^i$ of a cochain complex $\cC^*$}
  \Output{$a^{i-1}\in\cC^{i-1},a^i\in\cC^i$ with $e=a^i+\delta(a^{i-1})$}

  \Fn{\FnSSFlipErr{$e;i,\cC^*$}}{
    Initialize $a^{i-1}\gets 0\in\cC^{i-1}$, $a^i\gets 0\in\cC^i$ \\
    \While{$e\neq a^i+\delta(a^{i-1})$}{ \label{li:fewhile}
      \If{$\exists c^0\in\cC^0,c^{i-1}\in\cC^{i-1}$ with $c^0\preceq c^{i-1}$ and $|e+a^i+\delta(a^{i-1}+c^{i-1})|<|e+a^i+\delta(a^{i-1})|$}{ \label{li:im1flip}
        $a^{i-1} \gets a^{i-1}+c^{i-1}$
      }\ElseIf{$\exists c^0\in\cC^0,c^i\in\cC^{i}$ with $c^0\preceq c^i$ and $|\delta(e+a^i+c^i)|<|\delta(e+a^i)|$}{ \label{li:iflip}
        $a^i \gets a^i+c^i$
      }\Else{
        \Return{FAIL}
      }
    }
    \Return{$a^{i-1},a^i$}
  }
\end{algorithm}

Lemma~\ref{lem:sslocal} and Lemma~\ref{lem:footprint} below apply the notion that \FnSSFlipSyn{} and \FnSSFlipErr{} in Algorithm~\ref{alg:ssflip} are ``local algorithms'' (see e.g.~\cite{fawzi_efficient_2018}), meaning that they act independently on different connected components of the subgraph of the connectivity graph defined in Definition~\ref{def:cctodec} induced by the flipped cochain basis elements. The statements of these lemmas require the following definition.

\begin{definition}
  \label{def:footprint}
  Fix an $r$-dimensional cochain complex $\cC^*$ and some level $0\leq i\leq r-1$.

  For a data error $e\in\cC^i$ and a syndrome error $f\in\cC^{i+1}$, we say the \emph{footprint at time $t$} of \FnSSFlipSyn{$\delta(e)+f;i,\cC^*$} is the subset of $C^i\sqcup C^{i+1}$ given by the union of the values of the sets $\supp(e+a^i)\subseteq C^i$ and $\supp(\delta(e+a^i)+f)\subseteq C^{i+1}$ across the first $t$ iterations of the while loop.

  Similarly, for $e\in\cC^i$, we say the \emph{footprint at time $t$} of \FnSSFlipErr{$e;i,\cC^*$} is the subset of $C^{i-1}\sqcup C^i\sqcup C^{i+1}$ given by the union of the values of the sets $a^{i-1}\subseteq C^{i-1}$, $\supp(e+a^i+\delta(a^{i-1}))\subseteq C^i$ and $\supp(\delta(e+a^i))\subseteq C^{i+1}$ across the first $t$ iterations of the while loop.

  In both cases, by the \emph{footprint} we mean the footprint at time $t=\infty$ after the respective while loop has terminated.
\end{definition}

In Lemma~\ref{lem:sslocal} below, we slightly abuse the standard notation of restriction of vectors. Specifically, for $x\in\bF_2^n$ and $I\subseteq[n]$, we let $x|_I\in\bF_2^n$ denote the vector obtained from $x$ by replacing the values of every component in $[n]\setminus I$ with $0$; this notation can be seen as a shorthand for writing $(x|_I,0^{[n]\setminus I})$.

\begin{lemma}
  \label{lem:sslocal}
  Fix an $r$-dimensional cochain complex $\cC^*$. For $0\leq i\leq r-1$, $e\in\cC^i$, $f\in\cC^{i+1}$, let $S_{\mathrm{syn}}\subseteq C^i\sqcup C^{i+1}$ and $a^i\in\cC^i$ denote the footprint and output, respectively, of \FnSSFlipSyn{$\delta(e)+f;i,\cC^*$}. For every connected component $V$ of the subgraph of $G^{\cC}_{i,i+1}$ induced by $S_{\mathrm{syn}}$, the output of \FnSSFlipSyn{$\delta(e|_{V\cap C^i})+(f|_{V\cap C^{i+1}});i,\cC^*$} must equal $a^i|_{V\cap C^i}$.

  Similarly, for $0\leq i\leq r-1$, $e\in\cC^i$, let $S_{\mathrm{err}}\subseteq C^{i-1}\sqcup C^i\sqcup C^{i+1}$ and $a^{i-1}\in\cC^{i-1},\;a^i\in\cC^i$ denote the footprint and output, respectively, of \FnSSFlipErr{$e;i,\cC^*$}. For every connected component $V$ of the subgraph of $G^{\cC}_{i-1,i,i+1}$ induced by $S_{\mathrm{err}}$, the output of \FnSSFlipErr{$e|_{V\cap C^i};i,\cC^*$} must equal $a^{i-1}|_{V\cap C^{i-1}},\;a^i|_{\cap C^i}$.
\end{lemma}
\begin{proof}
  By definition, every update $c^i$ or $c^{i-1}$ in any iteration of the while loop of either algorithm \FnSSFlipSyn{$\delta(e)+f;i,\cC^*$} or \FnSSFlipErr{$e;i,\cC^*$} must be supported within the neighborhood in $G^{\cC}_{i,i+1}$ or $G^{\cC}_{i-1,i,i+1}$ respectively of the footprint from the previous step. Thus different connected components of the subgraph induced by the footprint do not interact at all during the algorithm's execution, so the result follows.
\end{proof}

\begin{lemma}
  \label{lem:footprint}
  Fix an $r$-dimensional cochain complex $\cC^*$ of locality $w$.
  Then for every $0\leq i\leq r-1$, $e\in\cC^i$, $f\in\cC^{i+1}$, in Algorithm~\ref{alg:ssflip}, at least $\gamma_{\mathrm{syn}}=1/4w^3$-fraction of the vertices in every connected component in the subgraph of $G^{\cC}_{i,i+1}$ induced by the footprint at every time $t$ of \FnSSFlipSyn{$\delta(e)+f;i,\cC^*$} must lie inside $\supp(e)\sqcup\supp(f)$.

  Similarly, for every $0\leq i\leq r-1$, $e\in\cC^i$, at least $\gamma_{\mathrm{err}}=1/8w^4$-fraction of the vertices in every connected component in the subgraph of $G^{\cC}_{i-1,i,i+1}$ induced by the footprint at every time $t$ of \FnSSFlipErr{$e;i,\cC^*$} must lie inside $\supp(e)$.
\end{lemma}
\begin{proof}
  We first prove the claim regarding \FnSSFlipSyn{$\delta(e)+f;i,\cC^*$}. In this algorithm, by definition every update $c^i$ (that gets added in to $a^i$) at a given timestep has $\supp(c^i)$ within the neighborhood of the footprint at the prior timestep. Therefore if $S$ denotes the footprint at time $t$, then within every connected component $V$ of the subgraph of $G^{\cC}_{i,i+1}$ induced by $S$, the sequence of updates $c^i$ is the same as if all entries of $e,f$ outside of $V$ were replaced with $0$s. It also follows that each $\supp(c^i)$ either lies entirely inside or outside of $V$.

  Assume for a contradiction that $<\gamma_{\mathrm{syn}}$-fraction of the vertices in $V$ lie in $\supp(e)\sqcup\supp(f)$. Then because the value $|\supp(\delta(e+a^i)+f)\cap V|$ must decrease in each iteration of the while loop prior to time $t$ for which $\supp(c^i)\subseteq V$, starting from the value $|\supp(\delta(e)+f)\cap V|<w\cdot\gamma_{\mathrm{syn}}|V|$ at time $0$, and each $|c^i|\leq w$ and $|\delta(c^i)|\leq w^2$, we must have that
  \begin{align*}
    |V|
    &\leq |\supp(e\sqcup \delta(e)\sqcup f)\cap V| + |(\delta(e)+f)\cap V|\cdot(w+w^2) \\
    &< \gamma_{\mathrm{syn}}|V|((w+1) + w\cdot(w+w^2)) \\
    &\leq 4\gamma_{\mathrm{syn}}w^3|V| \\
    &= |V|.
  \end{align*}
  Specifically, the size of the intersection of the footprint and $V$ is at most $|\supp(e\sqcup \delta(e)\sqcup f)\cap V|$ at time $0$, and increases by at most $|c^i|+|\delta(c^i)|\leq w+w^2$ in at most $|\supp(\delta(e)+f)\cap V|$ iterations of the while loop prior to time $t$. But the above inequality $|V|<|V|$ is a contradiction, so the assumption that $<\gamma_{\mathrm{syn}}$-fraction of the vertices in $V$ lie in $\supp(e)\sqcup\supp(f)$ was false, as desired.

  We now similarly prove the claim regarding \FnSSFlipErr{$e;i,\cC^*$}. In this algorithm, by definition every update $c^{i-1}$ or $c^i$ at a given timestep has support within the neighborhood of the footprint at the prior timestep. Therefore if $S$ denotes the footprint at time $t$, then within every connected component $V$ of the subgraph of $G^{\cC}_{i,i+1,i+2}$ induced by $S$, the sequence of updates $c^{i-1}$ and $c^i$ is the same as if all entries of $e$ outside of $V$ were replaced with $0$s. It also follows that each $\supp(c^{i-1})$ and $\supp(c^i)$ either lies intirely inside or outside of $V$.

  Assume for a contradiction that $<\gamma_{\mathrm{err}}$-fraction of the vertices in $V$ lie in $\supp(e)$. Then either $|\supp(e+a^i+\delta(a^{i-1}))\cap V|$ or $|\supp(\delta(e+a^i))\cap V|$ must decrease in each iteration of the while loop prior to time $t$ for which $\supp(c^{i-1})\subseteq V$ or $\supp(c^i)\subseteq V$. By definition $|\supp(\delta(e+a^i))\cap V|$ can only decrease, starting from the value $|\supp(\delta(e))\cap V|\leq w|\supp(e)\cap V|$. Every time this value decreases, the value $|\supp(e+a^i+\delta(a^{i-1}))\cap V|$ increases by at most $|c^i|\leq w$, and hence there are at most $|\supp(e)\cap V|+w\cdot w|\supp(e)\cap V|=(1+w^2)|\supp(e)\cap V|$ iterations in which the algorithm adds in an update $c^{i-1}$ to decrease $|\supp(e+a^i+\delta(a^{i-1}))\cap V|$. Therefore we obtain a contradiction
  \begin{align*}
    |V|
    &\leq |\supp(e\sqcup\delta(e))\cap V| + w|\supp(e)\cap V|\cdot (w+w^2) + (1+w^2)|\supp(e)\cap V|\cdot(w+w^2) \\
    &\leq |\supp(e)\cap V|((1+w)+w(w+w^2)+(1+w^2)(w+w^2)) \\
    &< \gamma_{\mathrm{err}}|V|\cdot 8w^4 \\
    &\leq |V|,
  \end{align*}
  so the assumption that $<\gamma_{\mathrm{err}}$-fraction of the vertices in $V$ lie in $\supp(e)$ was false, as desired.
\end{proof}

We now prove the main result of this section, which is a small-set flip decoder for the tensor product of 1-dimensional cochain complexes associated to lossless expanders.


\begin{proposition}
  \label{prop:ssflip}
  For every $r\in\bN$ and $0<\beta\leq 1$, there exists $\epsilon=\epsilon(r,\beta)>0$ such that the following holds. For some $\mu>0$ and $\Delta_{\max}\in\bN$, for each $h\in[r]$ let $G^{(h)}=(V^{(h)}=V_L^{(h)}\sqcup V_R^{(h)},E^{(h)})$ be a $(\mu,\epsilon)$-lossless expander of maximum left-degree $\Delta_L^{(h)}\leq \Delta_{\max}$ and minimum left-degree $\geq\beta \Delta_{\max}$. Let
  \begin{equation*}
    {\cC^{(h)}}^* = \left(\bF_2^{V_L^{(h)}}\xrightarrow{\delta^{(h)}=H_{G^{(h)}}^\top}\bF_2^{V_R^{(h)}}\right)
  \end{equation*}
  be the associated 1-dimensional cochain complex. Then the tensor product $\cA^*:=\cC^{(1)}\otimes\cdots\otimes\cC^{(r)}$ has an $(m,\gamma)$-small-set flip decoder at every level $0\leq i\leq r-1$ for $m=\min_{h\in[r]}\mu|V_L^{(h)}|/(r\Delta_{\max}^{r+1}+1)$ and $\gamma=1/10$.
\end{proposition}


To prove Proposition~\ref{prop:ssflip}, we will use the notion of \emph{robustness} studied in \cite{kalachev_two-sided_2023,dinur_expansion_2024,kalachev_maximally_2025}. While this notion of robustness applies to arbitrary collections of classical codes, we will only need to consider it for collections of repetition codes (of possibly different lengths), as described below.

\begin{definition}
  \label{def:robust}
  For $r\in\bN$ and $n_1,\dots,n_r\in\bN$, consider the 1-dimensional cochain complexes
  \begin{equation*}
    \cR^{(h)} = \left(\bF_2\xrightarrow{\delta^{(h)}}\bF_2^{n_h}\right)
  \end{equation*}
  in which each coboundary map $\delta^{(h)}$ is simply given by the $n_h\times 1$ matrix of all $1$s. We say that the collection of $r$ repetion codes $(\im(\delta^{(1)}),\dots,\im(\delta^{(r)}))$ (of lengths $n_1,\dots,n_r$ respectively) is \emph{$\eta$-robust} if the $r$-dimensional product cochain complex
  \begin{equation*}
    \cQ^* = \cR^{(1)}\otimes\cdots\otimes\cR^{(h)}
  \end{equation*}
  satisfies the following: for every $0\leq i\leq r-1$ and every $c\in\cQ^i$, there exists $b\in B^i(\cQ)$ such that
  \begin{equation*}
    \eta\cdot|b+c| \leq |\delta(c)|.
  \end{equation*}
\end{definition}

Note that if all $r$ repetition codes have the same length $n=n_1=\cdots=n_r$, then our robustness parameter $\eta$ equals $n$ times the robustness parameter defined in prior works \cite{kalachev_two-sided_2023,dinur_expansion_2024,kalachev_maximally_2025}. This slightly different convention is convenient for us because we will consider codes of different lengths $n_1\neq\cdots\neq n_r$.

\begin{lemma}[Follows from \cite{dinur_expansion_2024,kalachev_maximally_2025}]
  \label{lem:robust}
  For every $r\in\bN$ and $\beta>0$, there exists $\kappa=\kappa(r,\beta)>0$ such that for every $n\in\bN$, every $r$-tuple of repetition codes whose lengths all lie in $[\beta n,n]$ is $\kappa n$-robust.
\end{lemma}
\begin{proof}[Proof sketch]
  We briefly explain how the lemma follows from results in \cite{dinur_expansion_2024,kalachev_maximally_2025}. \cite{dinur_expansion_2024} show that a collection of classical codes is robust if it satisfies a notion called \emph{product-expansion}. \cite{kalachev_maximally_2025} in turn showed that a collection of codes is product-expanding if all of the codes satisfy a more standard property called \emph{local testability}. But it is well known that classical repetition codes are locally testable, for instance with parity-check matrix given by the vertex-edge incidence matrix of a constant-degree spectral-expander graph $G=(V,E)$. Specifically, we place code bits on vertices in $V$, and parity-checks on edges in $E$ enforce that the bits assigned to the incident vertices are equal. Then for a repetition codeword with an error supported on vertices in a set $S\subseteq V$, the probability that a random parity-check fails equals the fraction of edges with one vertex in $S$ and one vertex outside of $S$, which by the expander mixing lemma (see e.g.~\cite{vadhan_pseudorandomness_2012}) is proportional to $\min\{|S|,|V|-|S|\}/|V|$. Therefore a random check fails with probability proportional to the error weight, so the repetition code is locally testable, and thus \cite{kalachev_maximally_2025} implies the desired product-expansion result.
\end{proof}

\begin{remark}
  While some of the results in \cite{dinur_expansion_2024,kalachev_maximally_2025} are stated for collections of codes that all have the same length, these results and their proofs extend naturally to our setting where all code lengths lie in an interval $[\beta n,n]$ for some constant $\beta>0$, at the cost of just worse constants in the expressions bounding product-expansion and robustness.
  
  For conciseness above we deduced that repetition codes are product-expanding by applying their local testability along with the results of \cite{kalachev_maximally_2025}. Alternatively, it can be shown more directly that collections of repetition codes are product-expanding.
\end{remark}

In Lemma~\ref{lem:sserrorflip} below, we construct the error-flip decoder for Proposition~\ref{prop:ssflip}. We will subsequently apply this error-flip decoder to construct a syndrome-flip decoder.

\begin{lemma}
  \label{lem:sserrorflip}
  For every $r\in\bN$, $0<\beta\leq 1$, and $0<\nu<1$, there exists $\epsilon=\epsilon(r,\beta,\nu)>0$ such that the following holds. Define $\mu,\Delta_{\max},G^{(h)},\Delta_L^{(h)},{\cC^{(h)}}^*,\cA^*$ as in Proposition~\ref{prop:ssflip}. Then $\cA^*$ has an $(m,\nu)$-small-set error-flip decoder at every level $0\leq i\leq r-1$ for $m=\min_{h\in[r]}\mu|V_L^{(h)}|$.
\end{lemma}
\begin{proof}
  Define $\kappa=\kappa(r,\beta/2)>0$ to be the value from Lemma~\ref{lem:robust}. Set
  \begin{equation}
    \label{eq:sseps}
    \epsilon = \min\left\{\frac{\beta\nu}{8},\frac{\kappa\nu}{16r}\right\}.
  \end{equation}
  
  Fix some $0\leq i\leq r-1$ and some nonzero $e\in\cA^i$ of weight $|e|\leq m$. Our goal is to show that there exists some basis element $a^0\in A^0$ satisfying one of the two conditions in Definition~\ref{def:ssflip} (with $\cC^*=\cA^*$).

  For this purpose, we first construct a sequence of vertices $(v_h\in V_L^{(h)})_{h\in[r]}$ and subsets $(U_h\subseteq N_{G^{(h)}}(v_h))_{h\in[r]}$ inductively as follows, where each $|U_h|\geq(1-\epsilon)\Delta_L^{(h)}$. We will then appeal to the fact that the restriction of $\cA^*$ to basis elements corresponding to vertices in $\{v_h\}\cup U_h$ for $h\in[r]$ has the structure of the product of repetition codes, as considered in Definition~\ref{def:robust} and Lemma~\ref{lem:robust}.

  For $t\in[r]$, assume we have already determined $v_1,\dots,v_{t-1}$ and $U_1,\dots,U_{t-1}$. Recall that $A=\bigsqcup_{j\in[r]}A^r=V^{(1)}\times\cdots\times V^{(r)}$, and for $h\in[r]$ let $\Pi^{(h)}:A\rightarrow V^{(h)}$ denote projection onto the $h$th component. Define $B_t\subseteq A$ by
  \begin{equation*}
    B_t = \{a\in A:a_h\in(\{v_h\}\cup U_h) \;\forall\; h\in[t-1]\},
  \end{equation*}
  and define $F_t\subseteq V^{(t)}$ by 
  \begin{equation*}
    F_t = \Pi^{(t)}(\supp(e)\cap B_t).
  \end{equation*}
  That is, $F_t$ contains vertices in $G^{(t)}$ that equal the $t$th component of some $a\in\supp(e)$ whose first $t-1$ components lie in the respective sets $\{v_h\}\cup U_h$ for $h\in[t-1]$.

  We will enforce the inductive hypothesis that $F_t\neq\emptyset$, and that
  for every $b=(b_1,\dots,b_r)\in B_t$ and $h\in[t-1]$, every $v\in V_L^{(h)}$ with $v\triangleleft b_h$ and $(b_1,\dots,b_{h-1},v,b_{h+1},\dots,b_r)\in\supp(e)$ must equal $v=v_h$.
  
  For the base case, when $t=1$ by definition $F_t=\Pi^{(1)}(\supp(e))$ is nonempty by the assumption that $e\neq 0$. Now for $t\geq 1$, assume the inductive hypothesis holds. If $F_t\cap V_L^{(t)}=\emptyset$, then we must have $F_t\cap V_R^{(t)}\neq\emptyset$, so we set $v_t$ to be any vertex in $V_L^{(t)}$ that is incident to any vertex in $F_t$, and we set $U_t=N_{G^{(t)}}(v_t)$. Otherwise, if $F_t\cap V_L^{(t)}\neq\emptyset$, then because by assumption $|F_t|\leq|e|\leq m\leq\mu|V_L^{(t)}|$, we may apply Lemma~\ref{lem:losslesstoun} with $S=F_t\cap V_L^{(t)}$ to choose some $v_t\in S$ such that the set
  \begin{equation*}
    U_t := N_{G^{(t)}}(v_t)\setminus N_{G^{(t)}}(S\setminus\{v_t\})
  \end{equation*}
  satisfies
  \begin{equation}
    \label{eq:Utbound}
    |U_t| \geq (1-2\epsilon)\Delta_L^{(t)}.
  \end{equation}

  In both cases above, we chose $v_t,U_t$ so that $(\{v_t\}\cup U_t)\cap F_t\neq\emptyset$. Therefore by the definition of $F_t$, there exists some $b\in\supp(e)\cap B_t$ with $b_t\in(\{v_t\}\cup U_t)\cap F_t$ and therefore $B_{t+1}=\{a\in B_t:a_t\in(\{v_t\}\cup U_t)\}$ also contains $b$, and thus is nonempty. To complete the proof of the inductive hypothesis, we must show that
  for every $b\in B_{t+1}$, every $v\in V_L^{(t)}$ with $v\triangleleft b_t$ and $(b_1,\dots,b_{t-1},v,b_{t+1},\dots,b_r)\in\supp(e)$ must equal $v=v_t$. But if $v\neq v_t$, then because $v\in F_t\cap V_L^{(t)}=S$ by definition, we have that $v\in S\setminus\{v_t\}$, so $b_t\in N_{G^{(t)}}(v)\subseteq N_{G^{(t)}}(S\setminus\{v_t\})$, and therefore $b_t\notin U_t$, contradicting the assumption that $b\in B_{t+1}$.
  Therefore the inductive hypothesis must indeed hold, as desired.

  Having completed the definition of $v_1,\dots,v_r$ and $U_1,\dots,U_r$, we let $B=B_r\subseteq A$, and we define an $r$-dimensional cochain complex $\cB^*$ with $j$-cochain basis $B^j=B\cap A^j$, and with $j$-coboundary map $\delta_j^{\cB}(b)=\delta_j^{\cA}(b)|_{B^{j+1}}$ (where we view $b\in\bF_2^{B^j}\subseteq\bF_2^{A^j}$ via the inclusion $B^j\hookrightarrow A^j$). As a point of notation, for $a\in\cA^j$, we let $a_B=a|_{B^j}$. By definition $\cB^*$ is the tensor product of $r$ 1-dimensional complexes obtained from the restriction of the graphs $G^{(h)}$ to the subgraphs induced by vertices in $\{v_h\}\cup U_h$ for $h\in[r]$. As $U_h\subseteq N_{G^{(h)}}(v_h)$, these subgraphs are all star graphs, so $\cB^*$ is precisely the cochain complex associated to $r$ repetition codes of lengths $|U_1|,\dots,|U_r|$ as described in Definition~\ref{def:robust}.

  Also note that the final $(t=r)$ iteration of the inductive hypothesis above implies the following:
  \begin{enumerate}
  \item $\supp(e)\cap B\neq\emptyset$, that is, $e_B\neq 0$, and
  \item For every $a\in\supp(e)$ and $b\in B$ satisfying $a\triangleleft b$, then $a\in B$, that is,
    \begin{equation*}
      a=(b_1,\dots,b_{h-1},v_h,b_{h+1},\dots,b_r)
    \end{equation*}
    for some $h\in[r]$.
  \end{enumerate}

  Now define $a^0=(v_1,\dots,v_r)\in A^0$. Our goal is to show that $a^0$ satisfies one of the two conditions in Definition~\ref{def:ssflip} (with $\cC^*=\cA^*$). We consider two cases separately:
  \begin{enumerate}
  \item $|\delta^{\cB}(e_B)|>4r\epsilon \Delta_{\max}|e_B|/\nu$: In this case, we show that condition~\ref{it:flip} holds in Definition~\ref{def:robust}. Intuitively, because $\delta^{\cB}(e_B)$ is large enough, we will show that flipping all bits of $e$ in $\supp(e_B)$ reduces the weight of the image under $\delta^{\cA}$, as we will zero out all bits in $B^{i+1}$, while at most flipping a small number of $0$s to $1$s outside of $B^{i+1}$.

    Specifically, define $a^i\in\cA^i$ by $a^i=e_B\in\bF_2^{B^i}\subseteq\bF_2^{A^i}$. Note that if $i=0$, we have defined $a^0$ twice, but both definitions agree in setting $a^0$ to be the indicator of the basis element $(v_1,\dots,v_r)\in A^0$ (as by construction $e_B\neq 0$ and $B^0=\{(v_1,\dots,v_r)\}$).

    We now must show that
    \begin{equation*}
      |\delta^{\cA}(e+a^i)| < |\delta^{\cA}(e)|-(1-\nu)|\delta^{\cA}(a^i)|.
    \end{equation*}
    For this purpose, we have that
    \begin{equation}
      \label{eq:resB}
      \delta^{\cA}(e)|_{B^{i+1}} = \delta^{\cA}(a^i)|_{B^{i+1}} = \delta^{\cB}(a^i) = \delta^{\cB}(e_B).
    \end{equation}
    Specifically, the latter two equalities above hold by definition, so we only need to show that $\delta^{\cA}(e)|_{B^{i+1}}=\delta^{\cB}(e_B)$. If instead some $b\in B^{i+1}$ had $\delta^{\cA}(e)_b\neq\delta^{\cB}(e_B)_b$, then there must be some $a\in\supp(e)\setminus B^i$ with $a\triangleleft b$. But as shown above from the inductive definition of $B$, the conditions $a\in\supp(e)$, $b\in B^{i+1}$, and $a\triangleleft b$ imply that $a\in B^i$, so $a\notin\supp(e)\setminus B^i$. Therefore~(\ref{eq:resB}) holds.
    Furthermore, we also must have that
    \begin{equation}
      \label{eq:cobnotB}
      \left|\delta^{\cA}(a^i)|_{A^{i+1}\setminus B^{i+1}}\right| = \left|\delta^{\cA}(e_B)|_{A^{i+1}\setminus B^{i+1}}\right| \leq 2r\epsilon \Delta_{\max}|e_B| < \frac{\nu}{2}|\delta^{\cB}(a^i)| \leq \frac{\nu}{2}|\delta^{\cA}(a^i)|,
    \end{equation}
    as for every $a\in\supp(e_B)=\supp(e)\cap B$, then every $b\in A^{i+1}\setminus B^{i+1}$ satisfying $a\triangleleft b$ must be of the form $b=(a_1,\dots,a_{h-1},u,a_{h+1},\dots,a_r)$ for some $h\in[r]$ with $a_h=v_h$ and some $u\in N_{G^{(h)}}(v_h)\setminus U_h$; for each of the $\leq r$ choices of $h$, there are $\leq 2\epsilon \Delta_L^{(h)}\leq 2\epsilon \Delta_{\max}$ such $u$ by~(\ref{eq:Utbound}). Thus there are at most $|e_B|\cdot r\cdot 2\epsilon \Delta_{\max}$ elements $b\in A^{i+1}\setminus B^{i+1}$ that can satisfy $a\triangleleft b$ for some $a\in\supp(e_B)$, which is a necessary condition for $b\in\supp(\delta^{\cA}(e_B)|_{A^{i+1}\setminus B^{i+1}})$, and therefore the first inequality in~(\ref{eq:cobnotB}) holds. The second inequality in~(\ref{eq:cobnotB}) follows by the assumption that $|\delta^{\cB}(e_B)|>4r\epsilon \Delta_{\max}|e_B|/\nu$, and the third inequality in~(\ref{eq:cobnotB}) follows from~(\ref{eq:resB}). Combining~(\ref{eq:resB}) and~(\ref{eq:cobnotB}), we obtain the desired bound
    \begin{align*}
      |\delta^{\cA}(e+a^i)|
      &= \left|\delta^{\cA}(e+a^i)|_{B^{i+1}}\right| + \left|\delta^{\cA}(e+a^i)|_{A^{i+1}\setminus B^{i+1}}\right| \\
      &\leq \left|\delta^{\cA}(e)|_{A^{i+1}\setminus B^{i+1}}\right| + \left|\delta^{\cA}(a^i)|_{A^{i+1}\setminus B^{i+1}}\right| \\
      &= |\delta^{\cA}(e)| - \left|\delta^{\cA}(e)|_{B^{i+1}}\right| + \left|\delta^{\cA}(a^i)|_{A^{i+1}\setminus B^{i+1}}\right| \\
      &= |\delta^{\cA}(e)| - \left|\delta^{\cA}(a^i)|_{B^{i+1}}\right| + \left|\delta^{\cA}(a^i)|_{A^{i+1}\setminus B^{i+1}}\right| \\
      &= |\delta^{\cA}(e)| - \left|\delta^{\cA}(a^i)\right| + 2\left|\delta^{\cA}(a^i)|_{A^{i+1}\setminus B^{i+1}}\right| \\
      &< |\delta^{\cA}(e)| - (1-\nu)|\delta^{\cA}(a^i)|.
    \end{align*}
    Specifically, the first inequality above holds by~(\ref{eq:resB}), the first inequality holds by the triangle inequality, the second equality holds by definition, the third equality holds by~(\ref{eq:resB}), the fourth equality holds by definition, and the final inequality holds by~(\ref{eq:cobnotB}).
  \item $|\delta^{\cB}(e_B)|\leq 4r\epsilon \Delta_{\max}|e_B|/\nu$: In this case, we show that condition~\ref{it:notmin} holds in Definition~\ref{def:robust}. Intuitively, because $|\delta^{\cB}(e_B)|$ is small enough, we apply robustness of $\cB^*$ (Lemma~\ref{lem:robust}) to show that $e_B$ is close to a coboundary $\delta^{\cB}(a^{i-1})$. We then argue that we can reduce the weight of $e$ by adding in the coboundary $\delta^{\cA}(a^{i-1})$.

    First, if $i=0$, then as $B^0=\{(v_1,\dots,v_r)\}$, we have $e_B=\1_{(v_1,\dots,v_r)}$ and
    \begin{equation*}
      |\delta^{\cB}(e_B)| = \sum_{h=1}^r|N_{G^{(h)}}(v_h)| \geq r\beta \Delta_{\max}.
    \end{equation*}
    The above inequality contradicts the assumption that $|\delta^{\cB}(e_B)|\leq 4r\epsilon \Delta_{\max}|e_B|/\nu=4r\epsilon \Delta_{\max}/\nu$ because $\epsilon<\beta\nu/4$ by~(\ref{eq:sseps}). Therefore we may assume that $i\geq 1$.

    Recall that $\cB$ is by definition the product of $r$ 1-dimensional cochain complexes associated to repetition codes of lengths at most $\Delta_{\max}$ and at least $|U_h|\geq(1-2\epsilon)\Delta_L^{(h)}\geq(1-2\epsilon)\beta \Delta_{\max}\geq(\beta/2)\Delta_{\max}$ (as $\epsilon\leq 1/4$ by~(\ref{eq:sseps})). Therefore Lemma~\ref{lem:robust} implies that this tuple of repetition codes is $\kappa \Delta_{\max}$-robust in the sense of Definition~\ref{def:robust} for $\kappa=\kappa(r,\beta/2)$ as defined above.

    Applying the definition of robustness to $e_B$, there exists some $i$-coboundary $e_B'\in B^i(\cB)$ such that
    \begin{equation*}
      \kappa \Delta_{\max}|e_B+e_B'| \leq |\delta^{\cB}(e_B)| \leq 4r\epsilon \Delta_{\max}|e_B|/\nu,
    \end{equation*}
    that is,
    \begin{equation*}
      |e_B+e_B'| \leq \frac{4r\epsilon}{\kappa\nu}|e_B|.
    \end{equation*}
    Now applying robustness to any element of $\cB^{i-1}$ whose coboundary equals $e_B'$, we obtain some $(i-1)$-cochain $a^{i-1}\in\cB^{i-1}\subseteq\cA^{i-1}$ with $\delta^{\cB}(a^{i-1})=e_B'$ and
    \begin{equation*}
      |a^{i-1}| \leq \frac{1}{\kappa \Delta_{\max}}|e_B'| \leq \frac{1}{\kappa \Delta_{\max}}\left(\frac{4r\epsilon}{\kappa\nu}+1\right)|e_B|
    \end{equation*}
    By definition $|e_B'|\geq(1-4r\epsilon/\kappa\nu)|e_B|>0$ because $e_B\neq 0$ and $\epsilon\leq\kappa\nu/16r$, so $a^{i-1}\neq 0$. As a point of notation, if $i=1$ then it appears we have overloaded the variable $a^0=a^{i-1}$, but indeed in this case $a^{i-1}$ must equal a nonzero element of $\cB^0$, and the unique such element is $a^0=\1_{(v_1,\dots,v_r)}$.

    Our goal is to show that
    \begin{equation*}
      \left|e+\delta^{\cA}(a^{i-1})\right| < |e|.
    \end{equation*}
    For this purpose, we have that
    \begin{align*}
      \left|e+\delta^{\cA}(a^{i-1})|_{B^i}\right|
      &= |e_B+e_B'| \leq \frac{4r\epsilon}{\kappa\nu}|e_B|
    \end{align*}
    and
    \begin{align*}
      \left|(e+\delta^{\cA}(a^{i-1}))|_{A^i\setminus B^i}-e|_{A^i\setminus B^i}\right|
      &= \left|\delta^{\cA}(a^{i-1})_{A^i\setminus B^i}\right| \leq 2r\epsilon \Delta_{\max}|a^{i-1}|,
    \end{align*}
    where the final inequality above holds by the definition of $B$, by similar reasoning used to show the first inequality in~(\ref{eq:cobnotB}) above. 
    Combining the above inequalities then gives the desired bound
    \begin{align*}
      \left|e+\delta^{\cA}(a^{i-1})\right|
      &= \left|e+\delta^{\cA}(a^{i-1})|_{B^i}\right| + \left|e+\delta^{\cA}(a^{i-1})|_{A^i\setminus B^i}\right| \\
      &\leq \frac{4r\epsilon}{\kappa\nu}|e_B| + \left|e|_{A^i\setminus B^i}\right| + 2r\epsilon \Delta_{\max}|a^{i-1}| \\
      &= \frac{4r\epsilon}{\kappa\nu}|e_B| + |e| - |e_B| + 2r\epsilon \Delta_{\max}|a^{i-1}| \\
      &\leq |e| - \left(1-\frac{4r\epsilon}{\kappa\nu}-\frac{2r\epsilon}{\kappa}\left(\frac{4r\epsilon}{\kappa\nu}+1\right)\right)|e_B| \\
      &\leq |e| - \left(1 - \frac{4r\epsilon}{\kappa\nu}\left(\frac{4r\epsilon}{\kappa\nu}+2\right)\right)|e_B| \\
      &< |e|,
    \end{align*}
    where the final inequality above holds because $\epsilon\leq\kappa\nu/16r$ by definition from~(\ref{eq:sseps}).
  \end{enumerate}
\end{proof}

We now apply the error-flip decoder from Lemma~\ref{lem:sserrorflip} to construct a syndrome-flip decoder for Proposition~\ref{prop:ssflip}.

\begin{lemma}
  \label{lem:sssyndromeflip}
  For every $r\in\bN$, $0<\beta\leq 1$, let $\nu=1/10$ and define $\epsilon=\epsilon(r,\beta,\nu)$ as in Lemma~\ref{lem:sserrorflip}. Then define $\mu,\Delta_{\max},G^{(h)},\Delta^{(h)}_L,\cC^*$ as in Proposition~\ref{prop:ssflip}. Then $\cA^*$ has an $(m,\gamma)$-small-set syndrome-flip decoder at every level $0\leq i\leq r-1$ for $m=\min_{h\in[r]}\mu|V_L^{(h)}|/(r\Delta_{\max}^{r+1}+1)$ and $\gamma=\nu$.
\end{lemma}
\begin{proof}
  Fix some $0\leq i\leq r-1$, some nonzero $e\in\cA^i$ of weight $|e|\leq m$, and some $f\in\cA^{i+1}$ of weight $|f|\leq\gamma|\delta(e)|$. Our goal is to show that there exists some basis element $\tilde{a}^0\in A^0$ and some cochain $\tilde{a}^i\in\cA^i$ with $\tilde{a}^0\preceq \tilde{a}^i$ and $|\delta(e+\tilde{a}^i)+f|<|\delta(e)+f|$. Note that in this proof, we let $\delta=\delta^{\cA}$ denote the coboundary map for $\cA^*$.

  For this purpose, we first consider running \FnSSFlipErr{$e;i,\cA^*$} from Algorithm~\ref{alg:ssflip}. At each execution of line~\ref{li:iflip}, there are potentially multiple valid choices of $c^0,c^i$; assume that we choose whichever pair yields the greatest value of $(|\delta(e)|-|\delta(e+c^i)|)/|\delta(c^i)|$. Then let $c^{i-1}_1,\dots,c^{i-1}_{t_{i-1}}\in\cA^{i-1}$ and $c^i_1,\dots,c^i_{t_i}\in\cA^i$ be the sequence of all values for $c^{i-1}$ and $c^i$ chosen in line~\ref{li:im1flip} and line~\ref{li:iflip} respectively, across all runs of the while loop.

  We begin with the following claim, which follows by the assumption that $|e|\leq m$.

  \begin{claim}
    When we run \FnSSFlipErr{$e;i,\cA^*$} as described above, the output is some valid $a^{i-1}\in\cA^{i-1},a^i\in\cA^i$ with $e=a^i+\delta(a^{i-1})$, and the while loop must execute for $t_{i-1}+t_i$ iterations with $t_{i-1}\leq(r\Delta_{\max}^{r+1}+1)|e|$ and $t_i\leq|\delta(e)|\leq r\Delta_{\max}|e|$. Furthermore, for every $j\in[t_i]$ we have
    \begin{equation}
      \label{eq:bigdec}
      \left|\delta\left(e+\sum_{\ell=1}^{j-1}c^i_\ell\right)\right|-\left|\delta\left(e+\sum_{\ell=1}^jc^i_\ell\right)\right| > (1-\nu)|\delta(c^i_j)|.
    \end{equation}
  \end{claim}
  \begin{proof}
    Recall that the algorithm initializes $a^{i-1}\gets 0$ and $a^i\gets 0$. Every time the if statement in line~\ref{li:iflip} is satisfied, we decrease $|\delta(e+a^i)|$ by at least $1$, so this if statement can be satisfied at most
    \begin{equation*}
      t_i\leq|\delta(e)|\leq r\Delta_{\max}|e|
    \end{equation*}
    times. Each such update from this if statement can increase $|e+a^i+\delta(a^{i-1})|$ by at most $|c^i|\leq\Delta_{\max}^r$. Meanwhile, every time the if statement in line~\ref{li:im1flip} is satisfied, we decrease $|e+a^i+\delta(a^{i-1})|$ by at least $1$. Therefore the if statement in line~\ref{li:im1flip} can be satisfied at most
    \begin{equation*}
      t_{i-1}\leq \Delta_{\max}^r\cdot t_i+|e|\leq (r\Delta_{\max}^{r+1}+1)|e|
    \end{equation*}
    times, and the RHS above is also an upper bound on the maximum value of $|e+a^i+\delta(a^{i-1})|$ at any point in the algorithm's execution. Therefore because $|e|\leq m=\min_{h\in[r]}\mu|V_L^{(h)}|/(r\Delta_{\max}^{r+1}+1)$, we always have $|e+a^i+\delta(a^{i-1})|\leq\min_{h\in[r]}\mu|V_L^{(h)}|$. Lemma~\ref{lem:sserrorflip} then implies that one of the two if statements is satisfied in every iteration of the while loop, and that when the if statement in line~\ref{li:iflip} is satisfied, we have $|\delta(e+a^i)|-|\delta(e+a^i+c^i)|<(1-\nu)|\delta(c^i)|$, and hence~(\ref{eq:bigdec}) holds. It also then follows that the algorithm successfully returns some final values $a^{i-1},a^i$ satisfying $e=a^i+\delta(a^{i-1})$ (and does not return FAIL).
  \end{proof}

  Therefore \FnSSFlipErr{$e;i,\cA^*$} returns $a^i=\sum_{j=1}^{t_i}c^i_j$ with
  \begin{equation}
    \label{eq:de}
    \delta(e) = \delta(a^i) = \sum_{j=1}^{t_i}\delta(c^i_j)
  \end{equation}
  and
  \begin{equation}
    \label{eq:desize}
    |\delta(e)| = \sum_{j=1}^{t_i}\left(\left|\delta\left(e+\sum_{\ell=1}^{j-1}c^i_\ell\right)\right|-\left|\delta\left(e+\sum_{\ell=1}^jc^i_\ell\right)\right|\right) > (1-\nu)\sum_{j=1}^{t_i}|\delta(c^i_j)|.
  \end{equation}
  For each $j\in[t_i]$, let $U_j=\supp(\delta(e))\cap\supp(\delta(c^i_j))\setminus\bigcup_{j'\neq j}\supp(\delta(c^i_{j'}))$ denote the components that lie in the support of $\delta(e)$ and $\delta(c^i_j)$ but not in any other $\delta(c^i_{j'})$. By~(\ref{eq:de}), we have
  \begin{equation*}
    \supp(\delta(e))\subseteq\bigcup_{j\in[t_i]}\supp(\delta(c^i_j)).
  \end{equation*}
  Therefore across all the sets $\supp(\delta(c^i_j))$ for $j\in[t_i]$, there are $|\delta(e)|>(1-\nu)\sum_{j\in[t_i]}|\delta(c^i_j)|$ distinct elements that lie in $\supp(e)$, plus less than $\nu\sum_{j\in[t_i]}|\delta(c^i_j)|$ additional elements that either do not lie in $\supp(e)$, or may equal some of these $|\delta(e)|$ distinct elements. Thus every element of $\supp(e)$ lies in some set $\supp(c^i_j)$, and $<\nu\sum_{j\in[t_i]}|\delta(c^i_j)|$ of these elements lie in more than one such set, so by~(\ref{eq:desize}),
  \begin{equation*}
    \sum_{j\in[t_i]}|U_j| > (1-2\nu)\sum_{j\in[t_i]}|\delta(c^i_j)|.
  \end{equation*}
  Thus as by assumption $|f|\leq\gamma|\delta(e)|\leq\gamma\sum_{j\in[t_i]}|\delta(c^i_j)|$ with $\gamma=\nu$, we have
  \begin{equation*}
    \sum_{j\in[t_i]}|U_j\setminus\supp(f)| \geq \sum_{j\in[t_i]}|U_j|-|f| > (1-3\nu)\sum_{j\in[t_i]}|\delta(c^i_j)|.
  \end{equation*}
  Therefore there must exist some $\tilde{j}\in[t_i]$ with
  \begin{equation}
    \label{eq:Ujmf}
    |U_{\tilde{j}}\setminus\supp(f)| > (1-3\nu)|\delta(c^i_{\tilde{j}})|.
  \end{equation}
  Letting $\tilde{a}^i=c^i_{\tilde{j}}$, then by definition there exists some basis element $\tilde{a}^0\in A^0$ with $\tilde{a}^0\preceq\tilde{a}^i$; specifically, $\tilde{a}^0$ equals the value of $c^0$ during the execution of line~\ref{li:iflip} in Algorithm~\ref{alg:ssflip} in which $c^i=c^i_{\tilde{j}}$. Furthermore,
  \begin{align*}
    |\delta(e+\tilde{a}^i)+f|
    &= \left|\delta(e+\tilde{a}^i)+f|_{U_{\tilde{j}}\setminus\supp(f)}\right| + \left|\delta(e+\tilde{a}^i)+f|_{A^i\setminus(U_{\tilde{j}}\setminus\supp(f))}\right| \\
    &= \left|\delta(e+\tilde{a}^i)+f|_{A^i\setminus(U_{\tilde{j}}\setminus\supp(f))}\right| \\
    &\leq |\delta(e)+f|-\left|\delta(e)+f|_{U_{\tilde{j}}\setminus\supp(f)}\right|+\left|\delta(\tilde{a}^i)|_{A^i\setminus(U_{\tilde{j}}\setminus\supp(f))}\right| \\
    &\leq |\delta(e)+f|-|U_{\tilde{j}}\setminus\supp(f)|+(|\delta(\tilde{a}^i)|-|U_{\tilde{j}}\setminus\supp(f)|) \\
    &< |\delta(e)+f|-(1-6\nu)|\delta(\tilde{a}^i)|.
  \end{align*}
  where the second equality above holds because by construction $\delta(e)_b=\delta(\tilde{a}^i)_b=1$ and $f_b=0$ for every $b\in U_{\tilde{j}}\setminus\supp(f)$, the first inequality holds by the triangle inequality, the second inequality holds because by definition $U_{\tilde{j}}\subseteq\supp(\delta(e))\cap\supp(\tilde{a}^i)$ and $U_{\tilde{j}}\setminus\supp(f)\subseteq\supp(\tilde{a}^i)$, and the third inequality holds by~(\ref{eq:Ujmf}). Thus as by definition $\nu=1/10$, we conclude that $|\delta(e+\tilde{a}^i)+f|<|\delta(e)+f|$, as desired.
\end{proof}

We have now completed the proof of Proposition~\ref{prop:ssflip}:

\begin{proof}[Proof of Proposition~\ref{prop:ssflip}]
  The result follows immediately from Lemma~\ref{lem:sserrorflip} and Lemma~\ref{lem:sssyndromeflip}.
\end{proof}

\section{Downwards Code Switching Gadget via Direct Measurement}
\label{sec:switchdown}
This section presents our gadget for switching from an $r$-dimensional code to an $(r-1)$-dimensional code by directly measuring out some of the code's data qubits, and then applying an appropriate Pauli correction. We introduce notation and state our result in Section~\ref{sec:csstate}. The gadget is given in Algorithm~\ref{alg:switchdown}. We describe the noiseless execution of this gadget in Section~\ref{sec:csnoiseless}, and then we prove fault-tolerance for the noisy execution in Section~\ref{sec:csanalysis}.

A complementary gadget for switching up from an $(r-1)$-dimensional code to an $r$-dimensional code is given in Section~\ref{sec:switchup}. As described in Section~\ref{sec:csinf}, this upwards code switching performs logical teleportation using logical bell pairs between $(r-1)$- and $r$-dimensional codeblocks. We require the downwards code switching gadget, as well as the $CNOT$ and measurement gadgets in Section~\ref{sec:basicgadgets} below, in order to construct these bell pairs and perform the teleportation.

\subsection{Result Statement}
\label{sec:csstate}
In this section, we state our result providing a gadget for switching to a lower-dimensional code. We will first need the following notation.

For $r\in\bN$, let $\cA$ be an $(r-1)$-dimensional cochain complex. Let\footnote{Here we write $M^{\cA,*}$ as a shorthand for $(M^{\cA})^*$ (and similarly for $M^{\cB}$, etc.).} $\cM^{\cA,*}$ be the $(r-1)$-dimensional cochain complex with $\cM^{\cA,i}=\bF_2^{\dim(H^i(\cA))}$ and $\delta^{\cM^{\cA}}=0$, and let $\Enc^{\cA}:\cM^{\cA,*}\rightarrow\cA^*$ be a cochain map inducing an isomorphism on cohomology\footnote{Such a cochain map can be constructed by mapping each basis element in $M^{\cA,i}$ to some basis element of $H^i(\cA)$.}.

Let $\cB^*$ be a 1-dimensional cochain complex. Fix information sets $M^{\cB,0}\subseteq B^0$ and $M^{\cB,1}\subseteq B^1$ for $\ker(\delta^{\cB})$ and $\ker(\partial^{\cB})$ respectively, let $\cM^{\cB,*}=(\bF_2^{M^{\cB,0}}\xrightarrow{0}\bF_2^{M^{\cB,1}})$, and let $\Enc^{\cB}:\cM^{\cB,*}\rightarrow\cB^*$ be the cochain map defined in Lemma~\ref{lem:enc1dim}. Also fix a subset $L^{\cB,1}\subseteq M^{\cB,1}$, and let $\bar{L}^{\cB,1}=B^1\setminus L^{\cB,1}$. We let $\cB_{\bar{L}}^*=(\bF_2^{B^0}\xrightarrow{\delta^{\cB}_{\bar{L}}}\bF_2^{\bar{L}^{\cB,1}})$ denote the 1-dimensional cochain complex with coboundary map given by $\delta^{\cB}_{\bar{L}}(b)=\delta^{\cB}(b)|_{\bar{L}^{\cB,1}}$, and we similarly let $\cM_{\bar{L}}^{\cB,*}=(\bF_2^{M^{\cB,0}}\xrightarrow{0}\bF_2^{M^{\cB,1}\setminus L^{\cB,1}})$.

Because $L^{\cB,1}$ is extendable for $\ker(\partial^{\cB})$, no element of $B^1(\cB)=\im(\delta^{\cB})$ is supported inside $L^{\cB,1}$, and therefore $Z^0(\cB_{\bar{L}})=\ker(\delta^{\cB})=Z^0(\cB)$. Furthermore, the cosets $\1_x+B^1(\cB_{\bar{L}})$ for $x\in M^{\cB,1}\setminus L^{\cB,1}$ must be linearly independent, as any nontrivial linear dependence among them would yield a nonzero coboundary in $B^1(\cB)$ supported inside $M^{\cB,1}$. These cosets must also span all of $\cB_{\bar{L}}^1/B^1(\cB_{\bar{L}})$ because the cosets $\1_x+B^1(\cB)$ for $x\in M^{\cB,1}$ span all of $\cB^1/B^1(\cB)$ by Lemma~\ref{lem:extendable}. It follows by Lemma~\ref{lem:extendable} that $M^{\cB,0}$ and $M^{\cB,1}\setminus L^{\cB,1}$ are information sets for $Z^0(\cB_{\bar{L}})$ and $Z_1(\cB_{\bar{L}})$ respectively. We then let $\Enc_{\bar{L}}^{\cB}:\cM_{\bar{L}}^{\cB,*}\rightarrow\cB_{\bar{L}}^*$ be the encoding map for $\cB_{\bar{L}}^*$ given by Lemma~\ref{lem:enc1dim}. It follows by definition that for every $i\in\{0,1\}$ and $x\in\cM^{\cB,i}$,
\begin{equation}
  \label{eq:BresLenc}
  \Enc^{\cB}(x)|_{B_{\bar{L}}^i} = \Enc_{\bar{L}}^{\cB}(x|_{M_{\bar{L}}^{\cB,i}}).
\end{equation}


Let $\cC^*=\cA^*\otimes\cB^*$, let $\cM^{\cC,*}=\cM^{\cA,*}\otimes\cM^{\cB,*}$, and let $\Enc^{\cC}=\Enc^{\cA}\otimes\Enc^{\cB}:\cM^{\cC,*}\rightarrow\cC^*$, which by the K\"{u}nneth formula induces an isomorphism on cohomology and hence provides a CSS encoding map for codes associated to $\cC^*$.

Let $\cC_{\bar{L}}^*=\cA^*\otimes\cB_{\bar{L}}^*$. For some $2\leq i\leq r-1$, assume that $\cC_{\bar{L}}^*$ has a $(m,0)$-small-set flip decoder at level $i$. We similarly define $\cM_{\bar{L}}^{\cC,*}=\cM^{\cA,*}\otimes\cM_{\bar{L}}^{\cB,*}$ and $\Enc_{\bar{L}}^{\cC}=\Enc^{\cA}\otimes\Enc_{\bar{L}}^{\cB}$. Then by~(\ref{eq:BresLenc}), for every $x\in\cM^{\cC,i}$,
\begin{equation}
  \label{eq:CresLenc}
  \Enc^{\cC}(x)|_{C_{\bar{L}}^i} = \Enc_{\bar{L}}^{\cC}(x|_{M_{\bar{L}}^{\cC,i}}).
\end{equation}

Let $Q_{\mathrm{in}}$ be the $[[n_{\mathrm{in}},k_{\mathrm{in}}]]$ CSS code associated to level $i$ of $\cC^*$, and let $Q_{\mathrm{out}}$ be the $[[n_{\mathrm{out}},k_{\mathrm{out}}]]$ CSS code consisting of $|L^{\cB,1}|$ copies of the CSS code associated to level $i-1$ of $\cA^*$. Given some graph $G_{\mathrm{run}}=(V_{\mathrm{run}},E_{\mathrm{run}})$ that contains $G^{\cC}_{i-1,i,i+1}\subseteq G_{\mathrm{run}}$ as a subgraph and some
\begin{align*}
  \eta_{\mathrm{run}} &\leq \min\{m,\; d_i(\cC_{\bar{L}})\} \\
  \gamma_{\mathrm{run}} &\leq 1/150w^7,
\end{align*}
let
\begin{align*}
  \cE_{\mathrm{run}} &= \cE(G_{\mathrm{run}},\; \eta_{\mathrm{run}},\; \gamma_{\mathrm{run}}) \\
  \cE_{\mathrm{out}} &= \cE(G_{\mathrm{run}},\; \eta_{\mathrm{run}},\; 300w^7\gamma_{\mathrm{run}}),
\end{align*}
and define the decorated codes\footnote{Here we write $(\Enc^{\cA})^{\sqcup L^{\cB,1}}$ (or simply $(\Enc^{\cA})^{L^{\cB,1}}$) to denote $\Enc^{\cA}$ applied to $|L^{\cB,1}|$ disjoint copies of $Q_{\mathrm{out}}$, labeled by elements of $L^{\cB,1}$.}
\begin{align*}
  D_{\mathrm{in}} &= (Q_{\mathrm{in}},\; \Enc^{\cC},\; \cE_{\mathrm{run}}) \\
  D_{\mathrm{out}} &= (Q_{\mathrm{out}},\; (\Enc^{\cA})^{\sqcup L^{\cB,1}},\; \cE_{\mathrm{out}})
\end{align*}
In these decorated codes, the $n_{\mathrm{in}}=|C^i|$ physical qubits of $Q_{\mathrm{in}}$ are naturally associated with the set $C^i\subseteq C^{i-1}\sqcup C^i\sqcup C^{i+1}=V(G^{\cC}_{i-1,i,i+1})\subseteq V_{\mathrm{run}}$. The $n_{\mathrm{out}}=|A^{i-1}|\cdot|L^{\cB,1}|$ physical qubits of $Q_{\mathrm{out}}$ are then naturally associated to the subset $A^{i-1}\times L^{\cB,1}\subseteq A^{i-1}\times B^1\subseteq C^i\subseteq V_{\mathrm{run}}$.

The $k_{\mathrm{in}}$ logical qubits of $Q_{\mathrm{in}}$ are naturally associated to the set $M^{\cC,i}=(M^{\cA,i}\times M^{\cB,0})\sqcup (M^{\cA,i-1}\times M^{\cB,1})$. The $k_{\mathrm{out}}$ logical qubits of $Q_{\mathrm{out}}$ are naturally associated to the set $M^{\cA,i-1}\times L^{\cB,1}\subseteq M^{\cA,i-1}\times M^{\cB,1}\subseteq M^{\cC,i}$. Under these associations, we let $\bar{\cO}:\bC^{2^{k_{\mathrm{in}}}\times 2^{k_{\mathrm{in}}}} \rightarrow \bC^{2^{k_{\mathrm{out}}}\times 2^{k_{\mathrm{out}}}}$ be the map that simply discards (i.e.~traces out) all $k_{\mathrm{in}}-k_{\mathrm{out}}$ logical qubits in $M^{\cC,i}\setminus (M^{\cA,i-1}\times L^{\cB,1})$. That is, $\bar{\cO}(\rho)=\tr_{M^{\cC,i}\setminus (M^{\cA,i-1}\times L^{\cB,1})}(\rho)$

\begin{proposition}
  \label{prop:switchdown}
  Define all variables as above in Section~\ref{sec:csstate}. Then there exists a Pauli fault-tolerant gadget $((\cQ,\cE_{\mathrm{run}}^{\sqcup T}),D_{\mathrm{in}},D_{\mathrm{out}})$ for $\bar{\cO}$, where $\cQ$ is an adaptive quantum circuit using quantum space $N=n_{\mathrm{in}}$ and time $T=2$.
\end{proposition}

The gadget $\cQ$ in Proposition~\ref{prop:switchdown} is given in Algorithm~\ref{alg:switchdown}.

\begin{remark}
  \label{remark:sdexchange}
  While we have stated Proposition~\ref{prop:switchdown} and presented Algorithm~\ref{alg:switchdown} to perform $Z$-basis measurements, an analogous result holds where we dualize to a chain complex, and exchange the roles of the Pauli $X$ and $Z$ bases everywhere in the construction and analysis.
\end{remark}

\begin{algorithm}
  \caption{\label{alg:switchdown} Gadget in Proposition~\ref{prop:switchdown} for code switching from $\cC^*=\cA^*\otimes\cB^*$ down to multiple copies of $\cA^*$. All variables are defined as in Section~\ref{sec:csstate}. The physical qubits of the circuit are labeled by the set $C^i=(A^i\times B^0)\sqcup (A^{i-1}\times B^1)$. Below, we state the ideal input and output under a noisless execution. Note that $\delta^{\cC}(c)\in \cC^i$ in line~\ref{li:sdcorr} is well-defined because $c\in\cC_{\bar{L}}^{i-1}\subseteq\cC^{i-1}$ via the inclusion $C_{\bar{L}}^{i-1}\subseteq C^{i-1}$.}
  \SetKwInOut{Input}{Input}
  \SetKwInOut{Output}{Output}

  \SetKwFunction{FnSwitchDown}{SwitchDown}
  \SetKwProg{Fn}{Function}{:}{}

  \Input{$\sigma=\Enc^{\cC,i}(\rho)$}
  \Output{$(\Enc^{\cA,i-1})^{L^{\cB,1}}(\tr_{M^{\cC,i}\setminus (M^{\cA,i-1}\times L^{\cB,1})}(\rho))$}

  \Fn{\FnSwitchDown{$\sigma;\; i,\; \cA^*,\; \cB^*$}}{
    Measure Pauli $Z$ on all qubits in $C_{\bar{L}}^i=(A^i\times B^0)\sqcup (A^{i-1}\times\bar{L}^{\cB,1})$, and let $z\in\cC_{\bar{L}}^i$ be the resulting outcome \\ \label{li:sdmeas}
    Run \FnSSFlipSyn{$\delta^{\cC_{\bar{L}}}(z);i,\cC_{\bar{L}}^*$} from Algorithm~\ref{alg:ssflip}, and let $a^i\in\cC_{\bar{L}}^i$ be the output \\ \label{li:sddec}
    Run Gaussian elimination to find some $c_{\bar{L}}\in\cC_{\bar{L}}^{i-1}$ with $z+a^i+\delta^{\cC_{\bar{L}}}(c_{\bar{L}})\in\im(\Enc_{\bar{L}}^{\cC,i-1})$ \\ \label{li:sdge}
    Apply $X^{\delta^{\cC}(c_{\bar{L}})|_{A^{i-1}\times L^{\cB,1}}}$ to $\sigma$ \\ \label{li:sdcorr}
    \Return{$\tr_{C^i\setminus(A^{i-1}\times L^{\cB,1})}(\sigma)$} \label{li:sdret}
  }
\end{algorithm}

\subsection{Noiseless Execution}
\label{sec:csnoiseless}
In this section, as a warm-up for Proposition~\ref{prop:switchdown}, we show that a noiseless execution of Algorithm~\ref{alg:switchdown} performs the desired code switching.

\begin{lemma}
  \label{lem:nesd}
  Define all variables as in Section~\ref{sec:csstate}. For every $\ell\in\bN$ and every $\rho\in\bC^{(\cM^{\cC}\otimes\bF_2^\ell)\times(\cM^{\cC}\otimes\bF_2^\ell)}$, the output of \FnSwitchDown{$\Enc^{\cC,i}\otimes\cI_\ell(\rho);i,\cA^*,\cB^*$} in Algorithm~\ref{alg:switchdown} is proportional to $(\Enc^{\cA,i-1})^{L^{\cB,1}}\otimes\cI_\ell(\tr_{M^{\cC,i}\setminus (M^{\cA,i-1}\times L^{\cB,1})}(\rho))$.
\end{lemma}
\begin{proof}
  To begin, we expand $\rho$ as
  \begin{align}
    \label{eq:sdexpand}
    \rho
    &= \sum_{x,x'\in\cM^{\cC,i}}\rho_{x,x'}\ket{x}\bra{x'}\otimes\nu_{x,x'}
  \end{align}
  for coefficients $\rho_{x,x'}\in\bC$ and reference-system matrices $\nu_{x,x'}\in\bC^{2^\ell\times 2^\ell}$. Then\footnote{\label{footnote:sdabuse} Note that here (and throughout this section) we abuse notation by letting $\Enc^{\cC}=\Enc^{\cC,i}$ denote the encoding chain map, its induced map on cohomology, the associated encoding isometry, and the associated encoding channel. We also sometimes write $\Enc^{\cC}$ as a shorthand for $\Enc^{\cC}\otimes\cI_\ell$. The meaning will always be made clear from the argument.}
  \begin{align*}
    \Enc^{\cC}(\rho)
    &= \Enc^{\cC}\rho{\Enc^{\cC}}^\dagger = \sum_{x,x'\in\cM^{\cC,i},\;b,b'\in B^i(\cC)}\frac{\rho_{x,x'}}{|B^i(\cC)|}\ket{\Enc^{\cC}(x)+b}\bra{\Enc^{\cC}(x)+b'}.
  \end{align*}
  The Pauli $Z$ measurements in line~\ref{li:sdmeas} then collapse the above state to
  \begin{equation}
    \label{eq:sdnecoll}
    \sum_{z\in\cC_{\bar{L}}^i}\ket{z}\bra{z} \otimes \sum_{x,x',b,b'}\rho_{x,x'}\ket{\Enc^{\cC}(x)+b|_{A^{i-1}\times L^{\cB,1}}}\bra{\Enc^{\cC}(x')+b'|_{A^{i-1}\times L^{\cB,1}}} \otimes \nu_{x,x'},
  \end{equation}
  where the second sum above is over all $x,x'\in\cM^{\cC,i}$ and $b,b'\in B^i(\cC)$ satisfying
  \begin{equation}
    \label{eq:zexpressions}
    \Enc^{\cC}(x)+b|_{C_{\bar{L}}^i}=z=\Enc^{\cC}(x')+b'|_{C_{\bar{L}}^i}.
  \end{equation}
  Note that here we only track the state up to global scalars/phases (such as the dropped factor of $1/|B^i(\cC)|$ above). Now by~(\ref{eq:CresLenc}),
  \begin{equation*}
    \Enc^{\cC}(x)+b|_{C_{\bar{L}}^i} = \Enc_{\bar{L}}^{\cC}(x|_{M_{\bar{L}}^{\cC,i}})+(b|_{C_{\bar{L}}^i}).
  \end{equation*}
  The following claim shows that the decomposition of $z$ on the RHS above is unique.
  
  \begin{claim}
    \label{claim:zdecomp}
    For every $z\in Z^i(\cC_{\bar{L}})$, there is a unique choice of $x_{\bar{L}}\in\cM_{\bar{L}}^{\cC,i}$ and $b_{\bar{L}}\in B^i(\cC_{\bar{L}})$ for which
    \begin{equation*}
      z = \Enc_{\bar{L}}^{\cC}(x_{\bar{L}})+b_{\bar{L}}.
    \end{equation*}
    Furthermore, for $z\in\cC_{\bar{L}}^i\setminus Z^i(\cC_{\bar{L}})$, there is no such $x_{\bar{L}},b_{\bar{L}}$.
  \end{claim}
  \begin{proof}
    Because $\Enc_{\bar{L}}^{\cC}:\cM_{\bar{L}}^{\cC,*}\rightarrow\cC_{\bar{L}}^*$ is a chain map inducing an isomorphism on cohomology, and $\cM_{\bar{L}}^{\cC,*}$ has coboundary map $\delta_{\bar{L}}^{\cC}=0$, it follows that $\Enc_{\bar{L}}^{\cC}$ maps every nonzero input to a nontrivial cohomology class in $H^*(\cC_{\bar{L}})$, and every such cohomology class has a unique representative with such a (unique) preimage under $\Enc_{\bar{L}}^{\cC}$. Thus the claim holds.
  \end{proof}

  Claim~\ref{claim:zdecomp} implies that every term with a nonzero contribution to the sum in~(\ref{eq:sdnecoll}) has $z\in Z^i(\cC_{\bar{L}})$, so line~\ref{li:sddec} has a trivial decoding problem and computes $a^i=0$.
  
  Furthermore, Claim~\ref{claim:zdecomp} implies that for~(\ref{eq:zexpressions}) to hold, we must have $z=\Enc_{\bar{L}}^{\cC}(x_{\bar{L}})+b_{\bar{L}}\in Z^i(\cC_{\bar{L}})$ for $x_{\bar{L}}:=x|_{M_{\bar{L}}^{\cC,i}}=x'|_{M_{\bar{L}}^{\cC,i}}$ and $b_{\bar{L}}:=b|_{C_{\bar{L}}^i}=b'|_{C_{\bar{L}}^i}$. Fix an arbitrary $c_{\bar{L}}(b_{\bar{L}})\in\cC_{\bar{L}}^{i-1}$ with $\delta_{\bar{L}}^{\cC}(c_{\bar{L}}(b_{\bar{L}}))=b_{\bar{L}}$. Then for~(\ref{eq:zexpressions}) to hold for a given $z$, we must have $x=(x_{\bar{L}},x_L)$ for some $x_L\in\cM^{\cA,i-1}\otimes\bF_2^{L^{\cB,1}}$, and $b=\delta^{\cC}(c_{\bar{L}}(b_{\bar{L}}))+b_L$ for some $b_L\in B^{i-1}(\cA)\otimes\bF_2^{L^{\cB,1}}$. Note that here we can apply $\delta^{\cC}$ to $c_{\bar{L}}(b_{\bar{L}})$ using the natural inclusion $\cC_{\bar{L}}^{i-1}\subseteq\cC_{\bar{L}}^{i-1}\oplus(\cA^{i-2}\otimes\bF_2^{L^{\cB,1}})=\cC^{i-1}$; the term $b_L$ is the image under $\delta^{\cC}$ of some $c_L\in\cA^{i-2}\otimes\bF_2^{L^{\cB,1}}$. By the same reasoning, (\ref{eq:zexpressions}) also implies that $x'=(x_{\bar{L}},x_L')$ for some $x_L'\in\cM^{\cA,i-1}\otimes\bF_2^{L^{\cB,1}}$ and $b=\delta^{\cC}(c_{\bar{L}}(b_{\bar{L}}))+b_L'$ for some $b_L'\in B^{i-1}(\cA)\otimes\bF_2^{L^{\cB,1}}$.
  
  Thus~(\ref{eq:sdnecoll}) can be expressed as
  \begin{align}
    \label{eq:sdneequiv}
    \begin{split}
      \hspace{1em}&\hspace{-1em} \sum_{x_{\bar{L}}\in\cM_{\bar{L}}^{\cC,i},\; b_{\bar{L}}\in B^i(\cC_{\bar{L}})}\ket{\Enc_{\bar{L}}^{\cC}(x_{\bar{L}})+b_{\bar{L}}}\bra{\Enc_{\bar{L}}^{\cC}(x_{\bar{L}})+b_{\bar{L}}} \\
      &\otimes \sum_{x_L,x_L'\in\cM^{\cA,i-1}\otimes\bF_2^{L^{\cB,1}},\; b_L,b_L'\in B^{i-1}(\cA)\otimes\bF_2^{L^{\cB,1}}}\rho_{x,x'} \ket{(\Enc^{\cA,i-1})^{L^{\cB,1}}(x_L)+b_L+(\delta^{\cC}(c_{\bar{L}}(b_{\bar{L}}))|_{A^{i-1}\times L^{\cB,1}})} \\
      &\hspace{19em} \bra{(\Enc^{\cA,i-1})^{L^{\cB,1}}(x_L')+b_L'+(\delta^{\cC}(c_{\bar{L}}(b_{\bar{L}}))|_{A^{i-1}\times L^{\cB,1}})} \\
      &\otimes \nu_{x,x'}.
    \end{split}
  \end{align}
  The $c_{\bar{L}}$ computed in line~\ref{li:sdge} of Algorithm~\ref{alg:switchdown} must have the same image under $\delta^{\cC}$ as $c_{\bar{L}}(b_{\bar{L}})$. Specifically, by definition we have $z+\delta^{\cC_{\bar{L}}}(c_{\bar{L}}),\; z+\delta^{\cC_{\bar{L}}}(c_{\bar{L}}) \in \im(\Enc_{\bar{L}}^{\cC,i-1})$. Because $\Enc_{\bar{L}}^{\cC,i-1}$ induces an isomorphism on cohomology, every nonzero element of $\im(\Enc_{\bar{L}}^{\cC,i-1})$ lies in a nontrivial cohomology class, and hence outside of $\im(\delta^{\cC_{\bar{L}}})$. Thus $\delta^{\cC_{\bar{L}}}(c_{\bar{L}})=\delta^{\cC_{\bar{L}}}(c_{\bar{L}}(b_{\bar{L}}))$, so if $\delta^{\cC}(c_{\bar{L}})\neq\delta^{\cC}(c_{\bar{L}}(b_{\bar{L}}))$, then $\delta^{\cC}(c_{\bar{L}}+c_{\bar{L}}(b_{\bar{L}}))\neq 0$ must be supported entirely outside of $C_{\bar{L}}^i$. But such an element cannot exist because $L^{\cB,1}$ is extendable for $\ker(\partial^{\cC})$ so that $\im(\delta^{\cB})$ has no nonzero elements supported entirely inside $L^{\cB,1}$, and therefore $\delta^{\cC}(\cC_{\bar{L}}^{i-1})$ has no nonzero elements supported entirely inside $A^{i-1}\times L^{\cB,1}$. Thus the correction in line~\ref{li:sdcorr} maps the above state to 
  \begin{align*}
    \hspace{1em}&\hspace{-1em} \sum_{x_{\bar{L}}\in\cM_{\bar{L}}^{\cC,i},\; b_{\bar{L}}\in B^i(\cC_{\bar{L}})}\ket{\Enc_{\bar{L}}^{\cC}(x_{\bar{L}})+b_{\bar{L}}}\bra{\Enc_{\bar{L}}^{\cC}(x_{\bar{L}})+b_{\bar{L}}} \\
                &\otimes \sum_{x_L,x_L'\in\cM^{\cA,i-1}\otimes\bF_2^{L^{\cB,1}},\; b_L,b_L'\in B^{i-1}(\cA)\otimes\bF_2^{L^{\cB,1}}}\rho_{x,x'} \ket{(\Enc^{\cA,i-1})^{L^{\cB,1}}(x_L)+b_L} \bra{(\Enc^{\cA,i-1})^{L^{\cB,1}}(x_L')+b_L'} \\
                &\otimes \nu_{x,x'}.
  \end{align*}
  Tracing out over qubits $M_{\bar{L}}^{\cC,i}$ (i.e.~the first register above) then gives
  \begin{align*}
    \hspace{1em}&\hspace{-1em}\sum_{x_{\bar{L}}\in\cM_{\bar{L}}^{\cC,i}}(\Enc^{\cA,i-1})^{L^{\cB,1}}\left(\sum_{x_L,x_L'\in\cM^{\cA,i-1}\otimes\bF_2^{L^{\cB,1}}}\rho_{x,x'}\ket{x_L}\bra{x_L}\right) \otimes \nu_{x,x'} \\
                &= (\Enc^{\cA,i-1})^{L^{\cB,1}}(\tr_{M^{\cC,i}\setminus (M^{\cA,i-1}\times L^{\cB,1})}(\rho)),
  \end{align*}
  as desired.
\end{proof}

\subsection{Noisy Execution}
\label{sec:csanalysis}
In this section, we prove Proposition~\ref{prop:switchdown}. Here we define all variables as in Section~\ref{sec:csstate}. We will perform a similar analysis as in Section~\ref{sec:csnoiseless}, but now allowing for Pauli errors.

\begin{proof}[Proof of Proposition~\ref{prop:switchdown}]
  The desired circuit $\cQ$ is given in Algorithm~\ref{alg:switchdown}. By definition line~\ref{li:sdmeas} and line~\ref{li:sdcorr} each use $1$ timestep, while line~\ref{li:sddec} and line~\ref{li:sdge} just perform classical computations and therefore use $0$ timesteps (in the sense of adaptive quantum circuits as defined in Definition~\ref{def:circuit}). Thus the entire circuit uses $T=2$ timesteps. Furthermore, the circuit by definition uses $N=|C^i|=n_{\mathrm{in}}$ physical qubits.

  Our goal is to show that for every $\ell\in\bN$ and every $\rho\in\bC^{(\cM^{\cC}\otimes\bF_2^\ell)\times(\cM^{\cC}\otimes\bF_2^\ell)}$, every Pauli $\cE_{\mathrm{run}}$-deviation\footnote{See Footnote~\ref{footnote:sdabuse}.} $\sigma=E_0\Enc^{\cC,i}(\rho)E_0'$ of $\Enc^{\cC,i}(\rho)$ (for Paulis $E_0,E_0'$ such that $\supp(E_0)\cup\supp(E_0')\subseteq C^i$ is $\cE_{\mathrm{run}}$-avoiding) and every $\cE_{\mathrm{run}}^{\sqcup T}$-avoiding Pauli fault $\cF$, the resulting output $\cQ[\cF](\sigma)$ is a Pauli $\cE_{\mathrm{out}}$-deviation of $(\Enc^{\cA,i-1})^{L^{\cB,1}}(\tr_{M^{\cC,i}\setminus (M^{\cA,i-1}\times L^{\cB,1})}(\rho))$.

  We will perform a similar analysis as in the proof of Lemma~\ref{lem:nesd}, but while tracking all errors that arise at any point during the execution of $\cQ$ using the graph $G_{\mathrm{run}}$. Let
  \begin{equation*}
    S_{\cF} = (\supp(E_0)\cup\supp(E_0'))\cup\supp(\cF_1)\cup\supp(\cF_2)\subseteq C^i
  \end{equation*}
  denote the set of all physical qubits that lie in the support of the input's Pauli error $E_0,E_0'$ or of the fault $\cF$ at any timestep. Note that because $\cQ$ has only single-qubit gates, these errors cannot propagate through gates to different qubits.

  \begin{claim}
    \label{claim:sdSFavoid}
    The set $S_{\cF}$ is $\cE(G_{\mathrm{run}},\; \eta_{\mathrm{run}},\; 3\gamma_{\mathrm{run}})$-avoiding.
  \end{claim}
  \begin{proof}
    The claim holds because by definition $S_{\cF}$ is the union of three $\cE_{\mathrm{run}}=\cE(G_{\mathrm{run}},\; \eta_{\mathrm{run}},\; \gamma_{\mathrm{run}})$-avoiding sets.
  \end{proof}

  To begin, we expand $\rho$ as in~(\ref{eq:sdexpand}). Writing $E_0=Z^{e_Z}X^{e_X}$ and $E_0'=X^{e_X'}Z^{e_Z'}$, then the input state just prior to the measurements in line~\ref{li:sdmeas} is
  \begin{align*}
    E_0\Enc^{\cC}(\rho)E_0'
    &\propto \sum_{x,x'\in\cM^{\cC,i},\;b,b'\in B^i(\cC)}\rho_{x,x'} Z^{e_Z}X^{e_X}\ket{\Enc^{\cC}(x)+b}\bra{\Enc^{\cC}(x)+b'}X^{e_X'}Z^{e_Z'} \otimes \nu_{x,x'}.
  \end{align*}
  For $P\in\{X,Z\}$, write $e_P=(e_{P,\bar{L}},e_{P,L})$ for $e_{P,\bar{L}}=e_P|_{C_{\bar{L}}^i}$ and $e_{P,L}=e_P|_{A^{i-1}\times L^{\cB,1}}$ (and analogously for the primed versions $e_P=(e_{P,\bar{L}}',e_{P,L}')$). The measurements in line~\ref{li:sdmeas} project the state above to
  \begin{align}
    \label{eq:sdcoll}
    \begin{split}
      \hspace{1em}&\hspace{-1em} \sum_{z\in\cC_{\bar{L}}^i}Z^{e_{Z,\bar{L}}}\ket{z}\bra{z}Z^{e_{Z,\bar{L}}'} \\
      &\otimes \sum_{x,x',b,b'}\rho_{x,x'}Z^{e_{Z,L}}X^{e_{X,L}}\ket{\Enc^{\cC}(x)+b|_{A^{i-1}\times L^{\cB,1}}}\bra{\Enc^{\cC}(x')+b'|_{A^{i-1}\times L^{\cB,1}}}X^{e_{X,L}'}Z^{e_{Z,L}'} \otimes \nu_{x,x'},
    \end{split}
  \end{align}
  where the second sum above is over all $x,x'\in\cM^{\cC,i}$ and $b,b'\in B^i(\cC)$ satisfying
  \begin{align}
    \label{eq:zexpressionsnoisy}
    \begin{split}
      \Enc^{\cC}(x)+b|_{C_{\bar{L}}^i}+e_{X,\bar{L}} &= z =\Enc^{\cC}(x')+b'|_{C_{\bar{L}}^i}+e_{X,\bar{L}}'.
    \end{split}
  \end{align}
  Let $S_{\mathrm{syn}}\subseteq C^i\sqcup C^{i+1}$ be the footprint (see Definition~\ref{def:footprint}) of the call to \FnSSFlipSyn{$\delta^{\cC_{\bar{L}}}(z);i,\cC_{\bar{L}}^*$} in line~\ref{li:sddec} of Algorithm~\ref{alg:switchdown}.

  \begin{claim}
    \label{claim:sdsynavoid}
    The set $S_{\cF}\cup S_{\mathrm{syn}}$ is $\cE(G_{\mathrm{run}},\; \eta_{\mathrm{run}},\; 15w^3\gamma_{\mathrm{run}})$-avoiding.
  \end{claim}
  \begin{proof}
    Assume for a contradiction that there exists some $V\subseteq V_{\mathrm{run}}$ of size $|V|\geq\eta_{\mathrm{run}}$ that induces a connected subgraph of $G_{\mathrm{run}}$ with $|V\cap(S_{\cF}\cup S_{\mathrm{syn}})|/|V|\geq 15w^3\gamma_{\mathrm{run}}$. We may repeatedly add points in $(S_{\cF}\cup S_{\mathrm{syn}})\setminus V$ that lie in the neighborhood of $V$, until no more such points exist. 
    The resulting set $V$ can only have larger size $|V|\geq\eta_{\mathrm{run}}$ and $(S_{\cF}\cup S_{\mathrm{syn}})$-density $|V\cap(S_{\cF}\cup S_{\mathrm{syn}})|/|V|\geq 15w^3\gamma_{\mathrm{run}}$, and contains all or none of every connected component of the subgraph of $G_{\mathrm{run}}$ induced by $S_{\cF}$.

    By~(\ref{eq:zexpressionsnoisy}), the call \FnSSFlipSyn{$\delta^{\cC_{\bar{L}}}(z);i,\cC_{\bar{L}}^*$} in line~\ref{li:sddec} of Algorithm~\ref{alg:switchdown} has $z=\Enc^{\cC}(x)+b|_{C_{\bar{L}}^i}+e_{X,\bar{L}}$, and by definition $\delta^{\cC_{\bar{L}}}(\Enc^{\cC}(x)+b|_{C_{\bar{L}}^i})=0$, so $\delta^{\cC_{\bar{L}}}(z)=\delta^{\cC_{\bar{L}}}(e_{X,\bar{L}})$. Then by Lemma~\ref{lem:footprint}, we have
    \begin{align*}
      |V\cap S_{\mathrm{syn}}|
      &\leq |V\cap\supp(e_{X,\bar{L}})|/\gamma_{\mathrm{syn}} \\
      &\leq |V\cap S_{\cF}| \cdot 4w^3 \\
      &< 3\gamma_{\mathrm{run}}|V| \cdot 4w^3,
    \end{align*}
    where the third inequality above holds by Claim~\ref{claim:sdSFavoid}. We then again apply Claim~\ref{claim:sdSFavoid} to conclude that
    \begin{equation*}
      |V\cap (S_{\cF}\cup S_{\mathrm{syn}})| \leq |V\cap S_{\cF}|+|V\cap S_{\mathrm{syn}}| < (4w^3+1)\cdot 3\gamma_{\mathrm{run}}|V|,
    \end{equation*}
    which contradicts the assumption that $|V\cap(S_{\cF}\cup S_{\mathrm{syn}})|\geq 15w^3\gamma_{\mathrm{run}}$, as desired.
  \end{proof}

  Recall that by assumption $\eta_{\mathrm{run}}\leq m$. Also by definition, $\gamma_{\mathrm{run}}\leq 1/15w^3$, so it follows by Claim~\ref{claim:sdsynavoid} that the subgraph of $G_{\mathrm{run}}$ induced by $S_{\cF}\cup S_{\mathrm{syn}}$ has no connected components containing $\geq\eta_{\mathrm{run}}$ vertices. Therefore within every connected component $V$ of this induced subgraph, by the assumption that $\cC_{\bar{L}}^*$ has a $(m,0)$-small-set syndrome-flip decoder at level $i$, the while loop in the restricted execution \FnSSFlipSyn{$\delta^{\cC_{\bar{L}}}(z|_{V\cap C_{\bar{L}}^i});i,\cC_{\bar{L}}^*$} will not terminate until the computed output $a^i|_{V\cap C_{\bar{L}}^i}$ satisfies $\delta^{\cC_{\bar{L}}}(z+a^i|_{V\cap C_{\bar{L}}^i})=0$. By definition $\supp(e_{X,\bar{L}})\subseteq S_{\cF}$, and by~(\ref{eq:zexpressionsnoisy}), it follows that $\delta^{\cC_{\bar{L}}}(e_{X,\bar{L}}+a^i|_{V\cap C_{\bar{L}}^i})=\delta^{\cC_{\bar{L}}}(z+a^i|_{V\cap C_{\bar{L}}^i})=0$. We analogously have $\supp(e_{X,\bar{L}}')\subseteq S_{\cF}$ and $\delta^{\cC_{\bar{L}}}(e_{X,\bar{L}}'+a^i|_{V\cap C_{\bar{L}}^i})=0$.

  We now define $S_{\mathrm{err}}$ and $\tilde{a}^{i-1},\tilde{a}^i$ (resp.~$S_{\mathrm{err}}'$ and $(\tilde{a}^{i-1})',(\tilde{a}^i)'$) to be the footprint and output of \FnSSFlipErr{$e_{X,\bar{L}}+a^i;i,\cC_{\bar{L}}^*$} (resp.~\FnSSFlipErr{$e_{X,\bar{L}}'+a^i;i,\cC_{\bar{L}}^*$}), respectively. We prove the following claim for the unprimed variables $e_{X,\bar{L}},\dots$, but it analogously applies to the primed variables $e_{X,\bar{L}}',\dots$.

  \begin{claim}
    \label{claim:sderravoid}
    The set $S_{\cF}\cup S_{\mathrm{syn}}\cup S_{\mathrm{err}}$ is $\cE(G_{\mathrm{run}},\; \eta_{\mathrm{run}},\; 150w^7\gamma_{\mathrm{run}})$-avoiding.
  \end{claim}
  \begin{proof}
    Assume for a contradiction that there exists some $V\subseteq V_{\mathrm{run}}$ of size $|V|\geq\eta_{\mathrm{run}}$ that induces a connected subgraph of $G_{\mathrm{run}}$ with $|V\cap(S_{\cF}\cup S_{\mathrm{syn}}\cup S_{\mathrm{err}})|\geq 150w^7\gamma_{\mathrm{run}}$. Similarly as in Claim~\ref{claim:sdsynavoid}, we may repeatedly add points in $(S_{\cF}\cup S_{\mathrm{syn}}\cup S_{\mathrm{err}})\setminus V$ that lie in the neighborhood of $V$, in order to assume that $V$ contains all or none of every connected component of the subgraph of $G_{\mathrm{run}}$ induced by $S_{\cF}\cup S_{\mathrm{syn}}\cup S_{\mathrm{err}}$.

    By Lemma~\ref{lem:footprint}, we have
    \begin{align*}
      |V\cap S_{\mathrm{err}}|
      &\leq |V\cap\supp(e_{X,\bar{L}}+a^i)|/\gamma_{\mathrm{err}} \\
      &\leq |V\cap(S_{\cF}\cup S_{\mathrm{syn}})| \cdot 8w^4 \\
      &< 15w^3\gamma_{\mathrm{run}}|V| \cdot 8w^4,
    \end{align*}
    where the third inequality above holds by Claim~\ref{claim:sdsynavoid}. We then again apply Claim~\ref{claim:sdsynavoid} to conclude that
    \begin{equation*}
      |V\cap(S_{\cF}\cup S_{\mathrm{syn}}\cup S_{\mathrm{err}})| \leq |V\cap(S_{\cF}\cup S_{\mathrm{syn}})| + |V\cap S_{\mathrm{err}}| < (8w^4+1)15w^3\gamma_{\mathrm{run}}|V|,
    \end{equation*}
    which contradicts the assumption that $|V\cap(S_{\cF}\cup S_{\mathrm{syn}}\cup S_{\mathrm{err}})|\geq 150w^7\gamma_{\mathrm{run}}$, as desired.
  \end{proof}

  Recall that we use $\tilde{a}^{i-1},\tilde{a}^i$ to denote the variables that are updated throughout the execution of and then returned by \FnSSFlipErr{$e_{X,\bar{L}}+a^i;i,\cC_{\bar{L}}^*$}. As shown above, we have $\delta^{\cC_{\bar{L}}}(e_{X,\bar{L}}+a^i)$, and hence the entire execution of \FnSSFlipErr{$e_{X,\bar{L}}+a^i;i,\cC_{\bar{L}}^*$} will have $\tilde{a}^i=0$; that is, every iteration of the while loop (in line~\ref{li:fewhile} of Algorithm~\ref{alg:ssflip}) that does not fail must add some $c^{i-1}$ in to $\tilde{a}^{i-1}$. By definition, every such $c^{i-1}$ must have $\supp(c^{i-1})$ contained within the neighborhood in $G^{\cC}_{i-1,i,i+1}\subseteq G_{\mathrm{run}}$ of the footprint from the previous step. As a consequence, if we restrict $e_{X,\bar{L}}+a^i$ to a connected component $V$ if $S_{\cF}\cup S_{\mathrm{syn}}\cup S_{\mathrm{err}}$ (and zero out all values outside of $V$), the output of \FnSSFlipErr{$e_{X,\bar{L}}+a^i|_{V\cap C_{\bar{L}}^i};i,\cC_{\bar{L}}^*$} must equal the restriction $\tilde{a}^{i-1}|_{V\cap C_{\bar{L}}^{i-1}},\tilde{a}^i|_{V\cap C_{\bar{L}}^i}$ of the output $\tilde{a}^{i-1},\tilde{a}^i$ of \FnSSFlipErr{$e_{X,\bar{L}}+a^i;i,\cC_{\bar{L}}^*$} (see Lemma~\ref{lem:sslocal}).

  Recall that by assumption $\eta_{\mathrm{run}}\leq m$. Also by definition, $\gamma_{\mathrm{run}}\leq 1/150w^7$, so it follows by Claim~\ref{claim:sderravoid} that the subgraph of $G_{\mathrm{run}}$ induced by $S_{\cF}\cup S_{\mathrm{syn}}\cup S_{\mathrm{err}}$ has no connected components containing $\geq\eta_{\mathrm{run}}$ vertices. Therefore within every connected component $V$ of this induced subgraph, by the assumption that $\cC_{\bar{L}}^*$ has a $m$-small-set error-flip decoder at level $i$, the while loop in the restricted execution \FnSSFlipErr{$e_{X,\bar{L}}+a^i|_{V\cap C_{\bar{L}}^i};i,\cC_{\bar{L}}^*$} will not terminate until the computed output $\tilde{a}^{i-1}|_{V\cap C_{\bar{L}}^{i-1}}$ satisfies $\delta^{\cC_{\bar{L}}}(\tilde{a}^{i-1}|_{V\cap C_{\bar{L}}^{i-1}})=e_{X,\bar{L}}+a^i|_{V\cap C_{\bar{L}}^i}$; recall from above that we will have $\tilde{a}^i|_{V\cap C_{\bar{L}}^i}=0$. It also follows that the output will always be such a valid $\tilde{a}^{i-1}|_{V\cap C_{\bar{L}}^{i-1}}$, and never FAIL. Therefore we conclude that
  \begin{equation}
    \label{eq:sdta}
    \delta^{\cC_{\bar{}}}(\tilde{a}^{i-1}) = e_{X,\bar{L}}+a^i,
  \end{equation}
  and by definition $\supp(\tilde{a}^{i-1})\subseteq S_{\mathrm{err}}$.

  As in Lemma~\ref{lem:nesd}, for every $b_{\bar{L}}\in B^i(\cC_{\bar{L}})$, fix an arbitrary $c_{\bar{L}}(b_{\bar{L}})\in\cC_{\bar{L}}^{i-1}$ satisfying $\delta^{\cC_{\bar{L}}}(c_{\bar{L}}(b_{\bar{L}}))=b_{\bar{L}}$. Then it follows by~(\ref{eq:zexpressionsnoisy}) that for every measurement result $z$, then $c_{\bar{L}}(b|_{C_{\bar{L}}^i})+\tilde{a}^{i-1}$ is a valid choice for $c_{\bar{L}}$ in line~\ref{li:sdge} of Algorithm~\ref{alg:switchdown}. Thus because no nonzero element of $\im(\delta^{\cC}_{\bar{L}})$ lies in $\im(\Enc_{\bar{L}}^{\cC,i-1})$ (as every nonzero element of the latter lies in a nontrivial cohomology class; see the proof of Lemma~\ref{lem:nesd}), there is a unique valid value of $\delta^{\cC_{\bar{L}}}(c_{\bar{L}})$, so we must have
  \begin{equation*}
    c_{\bar{L}} = c_{\bar{L}}(b|_{C_{\bar{L}}^i})+\tilde{a}^{i-1}+z^{i-1}
  \end{equation*}
  for some cocycle $z^{i-1}\in Z^{i-1}(\cC_{\bar{L}})$, and hence
  \begin{equation}
    \label{eq:sdsamec}
    \delta^{\cC}(c_{\bar{L}}) = \delta^{\cC}(c_{\bar{L}}(b|_{C_{\bar{L}}^i})+\tilde{a}^{i-1})
  \end{equation}
  because $\delta^{\cC}(z^{i-1})=0$ as every nonzero element of $\delta^{\cC}(\cC_{\bar{L}}^{i-1})$ has nontrivial support inside $C_{\bar{L}}^i$ (see the proof of Lemma~\ref{lem:nesd}).

  Recall that the entire analysis above also applies to the primed variables $e_{X,\bar{L}}',\dots$, so we also have
  \begin{equation}
    \label{eq:sdtap}
    \delta^{\cC_{\bar{L}}}((\tilde{a}^{i-1})') = e_{X,\bar{L}}'+a^i,
  \end{equation}
  with $\supp(\tilde{a}^{i-1})\subseteq S_{\mathrm{err}}'$ and
  \begin{equation}
    \label{eq:sdsamecp}
    \delta^{\cC}(c_{\bar{L}}) = \delta^{\cC}(c_{\bar{L}}(b'|_{C_{\bar{L}}^i})+(\tilde{a}^{i-1})').
  \end{equation}

  By~(\ref{eq:zexpressionsnoisy}) and Claim~\ref{claim:zdecomp}, for every $x_{\bar{L}}\in\cM_{\bar{L}}^{\cC,i},\; b_{\bar{L}}\in B^i(\cC_{\bar{L}})$, there exists a unique choice of $x_{\bar{L}}'\in\cM_{\bar{L}}^{\cC,i},\; b_{\bar{L}}'\in B^i(\cC_{\bar{L}})$ such that
  \begin{equation*}
    \Enc_{\bar{L}}^{\cC}(x_{\bar{L}})+b_{\bar{L}}+e_{X,\bar{L}} = \Enc_{\bar{L}}^{\cC}(x_{\bar{L}}')+b_{\bar{L}}'+e_{X,\bar{L}}'.
  \end{equation*}
  By~(\ref{eq:sdta}) and~(\ref{eq:sdtap}), we have $\delta^{\cC_{\bar{L}}}(\tilde{a}^{i-1})=e_{X,\bar{L}}+e_{X,\bar{L}}'$, so $\Enc_{\bar{L}}^{\cC}(x_{\bar{L}}+x_{\bar{L}}')\in B^i(\cC_{\bar{L}})$, and thus we must have $x_{\bar{L}}=x'_{\bar{L}}$ because $\Enc_{\bar{L}}^{\cC}$ maps every nonzero input to a nontrivial cohomology class. Thus
  \begin{equation}
    \label{eq:sdbebep}
    b_{\bar{L}}+e_{X,\bar{L}} = b_{\bar{L}}'+e_{X,\bar{L}}'.
  \end{equation}
  
  Therefore by analogous reasoning as used to derive (\ref{eq:sdneequiv}) in the proof of Lemma~\ref{lem:nesd}, the state~(\ref{eq:sdcoll}) can be expressed as
  \begin{align}
    \label{eq:sdequiv}
    \begin{split}
      \hspace{1em}&\hspace{-1em} \sum_{x_{\bar{L}}\in\cM_{\bar{L}}^{\cC,i},\; b_{\bar{L}}\in B^i(\cC_{\bar{L}})}Z^{e_{Z,\bar{L}}}\ket{\Enc_{\bar{L}}^{\cC}(x_{\bar{L}})+b_{\bar{L}}+e_{X,\bar{L}}}\bra{\Enc_{\bar{L}}^{\cC}(x_{\bar{L}})+b_{\bar{L}}+e_{X,\bar{L}}}Z^{e_{Z',\bar{L}}} \\
      &\otimes \sum_{x_L,x_L'\in\cM^{\cA,i-1}\otimes\bF_2^{L^{\cB,1}},\; b_L,b_L'\in B^{i-1}(\cA)\otimes\bF_2^{L^{\cB,1}}} \\
      &\hspace{5em}\rho_{x,x'} Z^{e_{Z,L}}X^{e_{X,L}}\ket{(\Enc^{\cA,i-1})^{L^{\cB,1}}(x_L)+b_L+(\delta^{\cC}(c_{\bar{L}}(b_{\bar{L}}))|_{A^{i-1}\times L^{\cB,1}})} \\
      &\hspace{12em} \bra{(\Enc^{\cA,i-1})^{L^{\cB,1}}(x_L')+b_L'+(\delta^{\cC}(c_{\bar{L}}(b_{\bar{L}}'))|_{A^{i-1}\times L^{\cB,1}})}X^{e_{X,L}'}Z^{e_{Z,L}'} \\
      &\otimes \nu_{x,x'},
    \end{split}
  \end{align}
  where above $b_{\bar{L}}'$ is defined by~(\ref{eq:sdbebep}), that is $b_{\bar{L}}'=b_{\bar{L}}+e_{X,\bar{L}}+e_{X,\bar{L}}'$.

  Then running line~\ref{li:sdret} of Algorithm~\ref{alg:switchdown} on the state in~(\ref{eq:sdequiv}) simply applies the correction $X^{\delta^{\cC}(c_{\bar{L}})|_{A^{i-1}\times L^{\cB,1}}}$ and incurs the fault $\cF_2$ at timestep $t=2$, so it follows by~(\ref{eq:sdsamec}) and~(\ref{eq:sdsamecp}), the final state is
  \begin{align}
    \label{eq:sdfinal}
    \begin{split}
      \hspace{1em}&\hspace{-1em} \sum_{x_{\bar{L}}\in\cM_{\bar{L}}^{\cC,i},\; b_{\bar{L}}\in B^i(\cC_{\bar{L}})}Z^{f_{Z,\bar{L}}}X^{f_{X,\bar{L}}}\ket{\Enc_{\bar{L}}^{\cC}(x_{\bar{L}})+b_{\bar{L}}}\bra{\Enc_{\bar{L}}^{\cC}(x_{\bar{L}})+b_{\bar{L}}}X^{e_{X,\bar{L}}+e_{X,\bar{L}}'+f_{X,\bar{L}}'}Z^{f_{Z,\bar{L}}} \\
      &\otimes \sum_{x_L,x_L'\in\cM^{\cA,i-1}\otimes\bF_2^{L^{\cB,1}},\; b_L,b_L'\in B^{i-1}(\cA)\otimes\bF_2^{L^{\cB,1}}} \\
      &\hspace{5em}\rho_{x,x'} Z^{f_{Z,L}}X^{f_{X,L}}\ket{(\Enc^{\cA,i-1})^{L^{\cB,1}}(x_L)+b_L+(\delta^{\cC}(\tilde{a}^{i-1})|_{A^{i-1}\times L^{\cB,1}})} \\
      &\hspace{12em} \bra{(\Enc^{\cA,i-1})^{L^{\cB,1}}(x_L')+b_L'+(\delta^{\cC}(\tilde{a}^{i-1})'|_{A^{i-1}\times L^{\cB,1}})}X^{f_{X,L}'}Z^{f_{Z,L}'} \\
      &\otimes \nu_{x,x'},
    \end{split}
  \end{align}
  where all $f_\alpha$ and $f_\alpha'$ are obtained by adding the errors from $\cF_2$ in to $e_\alpha$ and $e_\alpha'$, respectively. Therefore by definition $f_X=(f_{X,\bar{L}},f_{X,L})$ and $f_Z=(f_{Z,\bar{L}},f_{Z,L})$ are supported inside $S_{\cF}$.

  Algorithm~\ref{alg:switchdown} returns the trace over the first register (i.e.~qubits in $C_{\bar{L}}^i$) of the state in~(\ref{eq:sdfinal}). If $e_{X,\bar{L}}+e_{X,\bar{L}}'+f_{X,\bar{L}}+f_{X,\bar{L}}'\neq 0$, then this state is $0$, so assume that $e_{X,\bar{L}}+e_{X,\bar{L}}'+f_{X,\bar{L}}+f_{X,\bar{L}}'=0$.

  If $f_{Z,\bar{L}}+f_{Z,\bar{L}}'\notin Z_i(\cC_{\bar{L}})=B^i(\cC_{\bar{L}})^\perp$, then for every $x_{\bar{L}}\in\cM_{\bar{L}}^{\cC,i}$, half the terms in the sum over $b_{\bar{L}}\in B^i(\cC_{\bar{L}})$ receive a $+1$ phase from the Paulis $Z^{f_{Z,\bar{L}}},Z^{f_{Z,\bar{L}}'}$ and half receive a $-1$ phase, so the overall state vanishes. Thus assume that $f_{Z,\bar{L}}+f_{Z,\bar{L}}'\in Z_i(\cC_{\bar{L}})$. Recall that $\supp(f_{Z,\bar{L}}+f_{Z,\bar{L}}')\subseteq S_{\cF}$. Because $\gamma_{\mathrm{run}}\leq 1/3$, Claim~\ref{claim:sdSFavoid} implies that the subgraph of $G_{\mathrm{run}}$ induced by $S_{\cF}$ has no connected components with $\geq\eta_{\mathrm{run}}$ vertices. Hence every connected component of the subgraph of $G^{\cC_{\bar{L}}}_i\subseteq G_{\mathrm{run}}$ induced by $\supp(f_{Z,\bar{L}}+f_{Z,\bar{L}}')$ has $<\eta_{\mathrm{run}}$ vertices, and by assumption $\eta_{\mathrm{run}}\leq d_i(\cC_{\bar{L}})$. The restriction of $f_{Z,\bar{L}}+f_{Z,\bar{L}}'$ to every such connected component must lie in $Z_i(\cC_{\bar{L}})$, as if any two such connected components shared a parity-check in $C_{\bar{L},i-1}$, they would necessarily be connected in $G^{\cC_{\bar{L}}}_i$, a contradiction. But then by the definition of the $i$-systolic distance $d_i(\cC_{\bar{L}})$, it follows that the restriction of $f_{Z,\bar{L}}+f_{Z,\bar{L}}'$ is a boundary in $B_i(\cC_{\bar{L}})$, and thence $f_{Z,\bar{L}}+f_{Z,\bar{L}}'\in B_i(\cC_{\bar{L}})$. But then because every $\Enc_{\bar{L}}^{\cC}(x_{\bar{L}})+b_{\bar{L}}\in Z^i(\cC_{\bar{L}})=B_i(\cC_{\bar{L}})^\perp$, every term in the sum over $x_{\bar{L}},b_{\bar{L}}$ receives the same phase $(-1)^{(f_{Z,\bar{L}}+f_{Z,\bar{L}})\cdot(e_{X,\bar{L}}+e_{X,\bar{L}}'+f_{X,\bar{L}}+f_{X,\bar{L}}')}$. Hence the final output returned by Algorithm~\ref{alg:switchdown} is proportional to
  \begin{align*}
    \begin{split}
      \hspace{1em}&\hspace{-1em} \sum_{x_{\bar{L}}\in\cM_{\bar{L}}^{\cC,i}}\; \sum_{x_L,x_L'\in\cM^{\cA,i-1}\otimes\bF_2^{L^{\cB,1}},\; b_L,b_L'\in B^{i-1}(\cA)\otimes\bF_2^{L^{\cB,1}}} \\
      &\hspace{3em}\rho_{x,x'} Z^{f_{Z,L}}X^{f_{X,L}+(\delta^{\cC}(\tilde{a}^{i-1})|_{A^{i-1}\times L^{\cB,1}})}\ket{(\Enc^{\cA,i-1})^{L^{\cB,1}}(x_L)+b_L} \\
      &\hspace{15em} \bra{(\Enc^{\cA,i-1})^{L^{\cB,1}}(x_L')+b_L'}X^{f_{X,L}'+(\delta^{\cC}(\tilde{a}^{i-1})'|_{A^{i-1}\times L^{\cB,1}})}Z^{f_{Z,L}'} \\
      &\otimes \nu_{x,x'},
    \end{split}
  \end{align*}
  which differs from the state
  \begin{align*}
    \hspace{1em}&\hspace{-1em}\sum_{x_{\bar{L}}\in\cM_{\bar{L}}^{\cC,i}}\; ((\Enc^{\cA,i-1})^{L^{\cB,1}}\otimes\cI_{\cL})\left(\sum_{x_L,x_L'\in\cM^{\cA,i-1}\otimes\bF_2^{L^{\cB,1}}}\rho_{x,x'}\ket{x_L}\bra{x_L'}\otimes\nu_{x,x'}\right) \\
                &= ((\Enc^{\cA,i-1})^{L^{\cB,1}}\otimes\cI_{\cL})(\tr_{M^{\cC,i}\setminus (M^{\cA,i-1}\times L^{\cB,1})}(\rho))
  \end{align*}
  by a Pauli error supported inside
  \begin{equation*}
    S_{\cF}\cup S_{\mathrm{run}}\cup S_{\mathrm{err}}\cup S_{\mathrm{err}}' = (S_{\cF}\cup S_{\mathrm{run}}\cup S_{\mathrm{err}}) \cup S_{\cF}\cup S_{\mathrm{run}}\cup S_{\mathrm{err}'}.
  \end{equation*}
  The set above is $\cC_{\mathrm{out}}=\cE(G_{\mathrm{run}},\; \eta_{\mathrm{run}},\; 300w^7\gamma_{\mathrm{run}})$-avoiding by Claim~\ref{claim:sderravoid}, so we have shown that the output of Algorithm~\ref{alg:switchdown} is a Pauli $\cE_{\mathrm{out}}$-deviation of $((\Enc^{\cA,i-1})^{L^{\cB,1}}\otimes\cI_{\cL})(\tr_{M^{\cC,i}\setminus (M^{\cA,i-1}\times L^{\cB,1})}(\rho))$, as desired.
\end{proof}

\section{Upwards Code Switching Gadget via Teleportation}
\label{sec:switchup}
This section presents our gadget for switching from an $(r-1)$-dimensional code to an $r$-dimensional code. Specifically, we prepare logical bell pairs between the two codes, which we use to teleport a state from the lower-dimensional code to the higher-dimensional code.

We will use our downwards code switching gadget from Proposition~\ref{prop:switchdown} to prepare these logical bell pairs. In Remark~\ref{remark:sulowdim} in Appendix~\ref{app:suproof}, we explain why we need to use this gadget when $r=3$, and how we can avoid using it when $r\geq 4$.

We describe notation and state our upwards-switching result below.
We prove the result in Appendix~\ref{app:suproof}, as it is a fairly direct application of other gadgets we develop in this paper, some of which we only present in the sections below. Specifically, the logical bell pair preparation will require the state preparation in Proposition~\ref{prop:stateprep}, the downwards code switching in Proposition~\ref{prop:switchdown}, and the transvesal $CNOT$ in Lemma~\ref{lem:CNOTdiff}; the teleportation will then require the transersal $CNOT$ in Lemma~\ref{lem:CNOTsame} and the logical Pauli measurements in Lemma~\ref{lem:measure}. The application of these gadgets is perhaps unsurprising, given that a teleportation circuit can be performed using qubit state preparation, $CNOT$ gates, Pauli measurements, and unitary Pauli gates (see Figure~\ref{fig:teleport}).

Fix $r\in\bN$, and define $\cA^*$, $\cM^{\cA,*}$, $\Enc^{\cA}$, $\cB^*$, $\cM^{\cB,*}$, $L^{\cB,1}$, $\Enc^{\cB}$, $\cB_{\bar{L}}^*$, $\cM^{\cB,*}$, $\Enc_{\bar{L}}^{\cB}$, $\cC^*$, $\cM^{\cC,*}$, $\Enc^{\cC}$ as in Section~\ref{sec:csstate}.

Fix some $2\leq i\leq r-1$. Assume that the chain complex\footnote{\label{footnote:labinv} Recall the label inversion for chain complexes vs cochain complexes; if $\cC_*$ here were viewed as a cochain complex, it would have a $(m,0)$-small-set flip decoder at level $r-(i-1)$.} $\cC_*=\cA_*\otimes\cB_*$ has a $(m,0)$-small-set flip decoder at level $i-1$. 
Also assume that the cochain complex\footnote{A further clarification following Footnote~\ref{footnote:labinv}: here $\cD^*=\cA^*\otimes\cB_*$ is not equal to $\cC^*=\cA^*\otimes\cB^*$, as we have dualized $\cB$.} $\cD^*:=\cA^*\otimes\cB_*$ has a $(m,0)$-small-set flip decoder at level $i$, 
and that the chain complex $\cD_{\bar{L},*}:=\cA_*\otimes\cB_{\bar{L}}^*$ has a $(m,0)$-small-set flip decoder at level $i-1$. 

We also let $\Enc^{\cD}$, $\Enc_{\bar{L}}^{\cD}$ be encoding maps for $\cD$, $\cD_{\bar{L}}$ respectively that are defined analogously to $\Enc^{\cC}$, $\Enc_{\bar{L}}^{\cC}$ in Section~\ref{sec:csstate} (with $\cD$ replacing $\cC$).

Let $G_{\mathrm{run}}=(V_{\mathrm{run}},E_{\mathrm{run}})$ be some graph that contains\footnote{Note that by definition there is an isomorphism between the basis sets $C\cong D$ given by an appropriate permutation, which induces an isomorphism between the graphs $G^{\cC}_{[r]}\cong G^{\cD}_{[r]}$.} $G^{\cC}_{[r]}\cong G^{\cD}_{[r]}\subseteq G_{\mathrm{run}}$ as a subgraph. There exists some sufficiently small $\gamma_{\mathrm{in}}(w)>0$ and some sufficiently large $\zeta_{\mathrm{out}}(w)>0$ such that for every
\begin{align*}
  \eta_{\mathrm{in}} &\leq \min\{m,\; d_{i-1}(\cA),\; d^{i-1}(\cA),\; d^i(\cD_{\bar{L}})\} \\
  \gamma_{\mathrm{in}} &\leq \gamma_{\mathrm{in}}(w) \\
  \gamma_{\mathrm{out}} &\geq \zeta_{\mathrm{out}}(w)\cdot\gamma_{\mathrm{in}},
\end{align*}
then Proposition~\ref{prop:switchup} below holds.

Let $Q_{\mathrm{in}}$ be the $[[n_{\mathrm{in}},k_{\mathrm{in}}]]$ CSS code consisting of $|L^{\cB,1}|$ copies of the CSS code associated to level $i-1$ of $\cA^*$, and let $Q_{\mathrm{out}}$ be the $[[n_{\mathrm{out}},k_{\mathrm{out}}]]$ CSS code associated to level $i$ of $\cC^*$. Let
\begin{align*}
  \cE_{\mathrm{in}} &= \cE(G^{\cA}_{[r-1]},\; \eta_{\mathrm{in}},\; \gamma_{\mathrm{in}}) \\
  \cE_{\mathrm{run}} &= \cE(G_{\mathrm{run}},\; \eta_{\mathrm{in}},\; \gamma_{\mathrm{in}}) \\
  \cE_{\mathrm{out}} &= \cE(G_{\mathrm{run}},\; \eta_{\mathrm{in}},\; \gamma_{\mathrm{out}}),
\end{align*}
and define the decorated codes
\begin{align*}
  D_{\mathrm{in}} &= (Q_{\mathrm{in}},\; \Enc_{\mathrm{in}}=(\Enc^{\cA})^{\sqcup L^{\cB,1}},\; \cE_{\mathrm{in}}^{\sqcup L^{\cB,1}}) \\
  D_{\mathrm{out}} &= (Q_{\mathrm{out}},\; \Enc_{\mathrm{out}}=\Enc^{\cC},\; \cE_{\mathrm{out}}).
\end{align*}
In these decorated codes, the $n_{\mathrm{in}}$ physical qubits of $Q_{\mathrm{in}}$ are naturally associated with the set $A^{i-1}\times L^{\cB,1}$, while the $n_{\mathrm{out}}$ physical qubits of $Q_{\mathrm{out}}$ are naturally associated with the set $C^i$; both sets are naturally subsets of $C^i\subseteq V_{\mathrm{run}}$.

Similarly as described in Section~\ref{sec:csstate} (but with the input and output codes reversed), the logical qubits $M^{\cA,i-1}\times L^{\cB,1}$ are naturally a subset of the logical qubits $M^{\cC,i}$ of $Q_{\mathrm{out}}$. We therefore let $\bar{\cO}:\bC^{2^{k_{\mathrm{in}}}\times 2^{k_{\mathrm{in}}}} \rightarrow \bC^{2^{k_{\mathrm{out}}}\times 2^{k_{\mathrm{out}}}}$ be the channel that acts as the identity on all the input qubits, and pads the input with an additional $k_{\mathrm{out}}-k_{\mathrm{in}}$ qubits in the $\ket{0}$ state. That is, $\bar{\cO}(\rho)=\rho\otimes(\ket{0}\bra{0})^{\otimes k_{\mathrm{out}}-k_{\mathrm{in}}}$.

\begin{proposition}
  \label{prop:switchup}
  Define all variables as above in Section~\ref{sec:switchup}. Then there exists a Pauli fault-tolerant gadget $((\cQ,(\cE_{\mathrm{run}}^{\sqcup u})^{\sqcup T}),D_{\mathrm{in}},D_{\mathrm{out}})$ for $\bar{\cO}$, where $\cQ$ is an adaptive quantum circuit using quantum space $N=O(n_{\mathrm{out}}w)$ and time $T=O(nw)$, and $u=O(1)$.
\end{proposition}

\begin{remark}
  \label{remark:suu}
  As described above (and more formally in Appendix~\ref{app:suproof}), the gadget $\cQ$ is entirely composed of (a constant number of) gadgets previously presented in the paper. Therefore the physical qubits are naturally partitioned into a constant number of blocks, each of which is naturally a subset of $C\cong D\subseteq V_{\mathrm{run}}$; the constant $u$ is the number of such blocks used.
\end{remark}

\section{State Preparation Gadget}
\label{sec:stateprep}
In this section, we present our gadget for preparing logical $\ket{0}$ and $\ket{+}$ states of quantum codes given by $\geq 3$-dimensional tensor products of classical lossless-expander codes. Our gadget crucially relies on the small-set flip decoder in Proposition~\ref{prop:ssflip}.

\begin{algorithm}
  \caption{\label{alg:stateprep} State preparation gadget in Proposition~\ref{prop:stateprep}.}
  \SetKwInOut{Input}{Input}
  \SetKwInOut{Output}{Output}

  \SetKwFunction{FnStatePrep}{StatePrep}
  \SetKwProg{Fn}{Function}{:}{}

  \Input{None}
  \Output{$\Enc(\ket{+^k}\bra{+^k})$}

  \Fn{\FnStatePrep{$i,\cC^*$}}{
    Initialize an $n=|C^i|$-qubit code register $\rho\gets\ket{+^{C^i}}\bra{+^{C^i}}$, and an additional $|C^{i+1}|$-qubit ancilla register $\alpha\gets\ket{0^{C^{i+1}}}\bra{0^{C^{i+1}}}$ \\ \label{li:spinit}
    Run \FnSyndExt{$\rho;Z,\delta_i$} from Algorithm~\ref{alg:syndext}, and let $s\in\cC^{i+1}$ be the resulting syndrome \\ \label{li:spse}
    Run \FnSSFlipSyn{$\delta(s);i+1,\cC^*$} from Algorithm~\ref{alg:ssflip}, and let $a^{i+1}\in\cC^{i+1}$ be the output \\ \label{li:spdec}
    Run Gaussian elimination to find some $x\in\cC^i$ with $\delta(x)=s+a^{i+1}$ \\ \label{li:spge}
    Apply the Pauli $X^x$ to $\rho$ \\ \label{li:spcorr}
    \Return{$\rho$} \label{li:spret}
  }
\end{algorithm}

\begin{proposition}
  \label{prop:stateprep}
  For $r\in\bN$ and $1\leq i\leq r-2$, let $\cC^*$ be an $r$-dimensional cochain complex of locality $w$ with a $(m,0)$-small-set flip decoder at level $i+1$ (see Definition~\ref{def:ssflip}). Let $Q$ be the $[[n,k]]$ CSS code associated to level $i$ of $\cC^*$, with an associated CSS encoding map $\Enc$ (see Definition~\ref{def:cctodec}). Given some graph $G_{\mathrm{run}}=(V_{\mathrm{run}},E_{\mathrm{run}})$ that contains $G^{\cC}_{i,i+1,i+2}\subseteq G_{\mathrm{run}}$ as a subgraph and some
  \begin{align*}
    \eta_{\mathrm{run}} &\leq m \\
    \gamma_{\mathrm{run}} &\leq \frac{1}{50w^7(w+10)2^{w+10}},
  \end{align*}
  let
  \begin{align*}
    \cE_{\mathrm{run}} &= \cE(G_{\mathrm{run}},\; \eta_{\mathrm{run}},\; \gamma_{\mathrm{run}}) \\
    \cE_{\mathrm{out}} &= \cE(G_{\mathrm{run}},\; \eta_{\mathrm{run}},\; 100w^7(w+10)2^{w+10}\gamma_{\mathrm{run}}),
  \end{align*}
  and define the decorated code (see Definition~\ref{def:deccode}) $D_{\mathrm{out}} = (Q,\;\Enc,\;\cE_{\mathrm{out}})$.
  Let $\bar{\cO}:\bC\rightarrow\bC^{2^k\times 2^k}$ be the channel given by $\bar{\cO}(1)=\ket{+^k}\bra{+^k}$.

  Then there exists a Pauli fault-tolerant gadget $((\cQ,\cE_{\mathrm{run}}^{\sqcup T}),\emptyset,D_{\mathrm{out}})$ for $\bar{\cO}$, where $\cQ$ is an adaptive quantum circuit using quantum space $N=O(nw)$ and time $T=O(w)$.
\end{proposition}

The gadget $\cQ$ in Proposition~\ref{prop:stateprep}. is given in Algorithm~\ref{alg:stateprep}. It is well-known that the noisless execution prepares the desired logical $\ket{+^k}$ state; we include a brief proof below for completeness.

\begin{lemma}[Well known]
  \label{lem:nesp}
  In the absence of errors, \FnStatePrep{$i,\cC^*$} in Algorithm~\ref{alg:stateprep} returns $\Enc(\ket{+^k}\bra{+^k})$.
\end{lemma}
\begin{proof}
  By Corollary~\ref{cor:syndext}, following the syndrome-extraction measurements in line~\ref{li:spse}, the code register must equal $X^b\Enc(\ket{+^k}\bra{+^k})X^b$ for some $b\in\cC^i$ such that the measured syndrome is $s=\delta(b)$. Specifically, the overall state is given by the mixture
  \begin{equation*}
    \sum_{b\in\cC^i}X^b\Enc\left(\ket{+^k}\bra{+^k}\right)X^b\otimes\ket{\delta(b)}\bra{\delta(b)}.
  \end{equation*}
  This statement holds because the $X$ and $Z$ stabilizers of a CSS code always commute, and the initial state $\ket{+^n}$ lies in the $+1$ eigenstate of all $X$ stabilizers and logical operators, so the post-measurement state also lies in the $+1$ eigenstate of all $X$ stabilizers and logical operators, and hence must equal $\Enc(\ket{+^k}\bra{+^k})$ with some (superposition of) $X$ errors applied. The $Z$ measurements we performed must collapse this superposition down to a single Pauli error $X^b$, and by definition the syndrome is $s=\delta(b)$.

  Because $s=\delta(b)$, we have $\delta(s)=0$, so line~\ref{li:spdec} has a trivial decoding problem that gives $a^{i+1}=0$. Then line~\ref{li:spge} computes some $x\in\cC^i$ with $\delta(x)=s$, and line~\ref{li:spcorr} applies $X^x$ to compute the final output state $\rho=X^{b+x}\Enc(\ket{+^k}\bra{+^k})X^{b+x}$. By definition $\delta(b+x)=s+s=0$, so $X^{b+x}$ is a logical $X$-operator and hence preserves $\Enc\ket{+^k}$. Equivalently, $\Enc\ket{+^k}$ is by definition a uniform superposition over all $i$-cocycles in $Z^i(\cC)$, and this uniform superposition is preserved under adding in any given $i$-cocycle (such as $b+x$). Thus the output is $\rho=\Enc(\ket{+^k}\bra{+^k})$, as desired.
\end{proof}

\begin{remark}
  \label{remark:spexchange}
  While Proposition~\ref{prop:stateprep} is stated for the preparation of logical $\ket{+}$ states, the same result holds for preparation of logical $\ket{0}$ states, by simply exchanging the roles of the Pauli $X$ and $Z$ bases everywhere in the construction and analysis.
\end{remark}

To prove Proposition~\ref{prop:stateprep}, we analyze the noisy execution of Algorithm~\ref{alg:stateprep}. The proof is largely similar to the proof of Proposition~\ref{prop:switchdown} given in Section~\ref{sec:csanalysis} above, so we postpone the proof to Appendix~\ref{app:spproof}.

\section{Error Correction Gadget}
\label{sec:errcorr}
This section presents our gadget for performing error correction on codes with small-set syndrome-flip decoders.

\begin{algorithm}
  \caption{\label{alg:errcorr} Error correction gadget in Proposition~\ref{prop:errcorr}; all variables are defined as in the proposition. In particular, the output avoids a larger family $\cE_{\mathrm{out}}\supseteq\cE_{\mathrm{in}}$ of bad sets than the input, and hence is less corrupted. We call the $n=|C^i|$-qubit input the \emph{code register}, and the $|C^{i+1}|$ and $|C^{i-1}|$ qubit ancilla systems used by the syndrome extraction subroutines in line~\ref{li:ecseZ} and line~\ref{li:ecseX} respectively the $Z$ and $X$ \emph{syndrome registers}.}
  \SetKwInOut{Input}{Input}
  \SetKwInOut{Output}{Output}

  \SetKwFunction{FnErrCorr}{ErrCorr}
  \SetKwProg{Fn}{Function}{:}{}

  \Input{$\sigma\in\bC^{\cC^i\times\cC^i}$ that is a $\cE_{\mathrm{in}}$-deviation of $\Enc(\rho)$ for $\rho\in\bC^{2^k\times 2^k}$}
  \Output{$\sigma'\in\bC^{\cC^i\times\cC^i}$ that is a $\cE_{\mathrm{out}}$-deviation of $\Enc(\rho)$}

  \Fn{\FnErrCorr{$\sigma;i,\cC^*$}}{
    Run \FnSyndExt{$\sigma;Z,\delta^{\cC}_i$} from Algorithm~\ref{alg:syndext}, and let $s_Z\in\cC^{i+1}$ be the resulting syndrome \\ \label{li:ecseZ}
    Run \FnSSFlipSyn{$s_Z;i,\cC^*$} from Algorithm~\ref{alg:ssflip}, and let $a_Z^i\in\cC^i$ be the output \\ \label{li:ecdecZ}
    Apply the Pauli $X^{a_Z^i}$ to $\sigma$ \\ \label{li:eccorrZ}
    Run \FnSyndExt{$\sigma;X,\partial^{\cC}_i$} from Algorithm~\ref{alg:syndext}, and let $s_X\in\cC_{i-1}$ be the resulting syndrome \\ \label{li:ecseX}
    Run \FnSSFlipSyn{$s_X;i,\cC_*$} from Algorithm~\ref{alg:ssflip}, and let $a_{X,i}\in\cC_i$ be the output \\ \label{li:ecdecX}
    Apply the Pauli $X^{a_{X,i}}$ to $\sigma$ \\ \label{li:eccorrX}
    \Return{$\sigma$} \label{li:ecret}
  }
\end{algorithm}

\begin{proposition}
  \label{prop:errcorr}
  For $r\in\bN$ and $1\leq i\leq r-1$, let $\cC^*$ be an $r$-dimensional cochain complex of locality $w$ with a $(m,\gamma_{\mathrm{dec}})$-small-set flip decoder at level $i$. Let $Q$ be the $[[n,k]]$ CSS code associated to level $i$ of $\cC^*$, with an associated CSS encoding map $\Enc$ (see Definition~\ref{def:cctodec}). Given some graph $G_{\mathrm{run}}=(V_{\mathrm{run}},E_{\mathrm{run}})$ that contains $G^{\cC}_{i-1,i,i+1}\subseteq G_{\mathrm{run}}$ as a subgraph and some
  \begin{align}
    \label{eq:ecbadparams}
    \begin{split}
      \eta_{\mathrm{in}} &\leq m \\
      \gamma_{\mathrm{in}} &\leq \frac{1}{50w^7T2^{T+1}} \\
      \eta_{\mathrm{run}} &\leq \eta_{\mathrm{in}} \\
      \gamma_{\mathrm{run}} &\leq \gamma_{\mathrm{in}},
    \end{split}
  \end{align}
  let
  \begin{align*}
      \gamma_{\mathrm{out}} &= \frac{(w+1)T2^{T+6}\gamma_{\mathrm{run}}}{\gamma_{\mathrm{dec}}} \\
    \cE_{\mathrm{in}} &= \cE(G_{\mathrm{run}},\; \eta_{\mathrm{in}},\; \gamma_{\mathrm{in}}) \\
    \cE_{\mathrm{run}} &= \cE(G_{\mathrm{run}},\; \eta_{\mathrm{run}},\; \gamma_{\mathrm{run}}) \\
    \cE_{\mathrm{out}} &= \cE(G_{\mathrm{run}},\; \eta_{\mathrm{run}},\; \gamma_{\mathrm{out}}).
  \end{align*}
  Define the decorated codes (see Definition~\ref{def:deccode}) $D_{\mathrm{in}}=(Q,\;\Enc,\;\cE_{\mathrm{in}})$ and $D_{\mathrm{out}}=(Q,\;\Enc,\;\cE_{\mathrm{out}})$. Let $\bar{\cO}=\cI_k:\bC^{2^k\times 2^k}\rightarrow\bC^{2^k\times 2^k}$ be the identity channel.

  Then there exists a Pauli fault-tolerant gadget $((\cQ,\cE_{\mathrm{run}}^{\sqcup T}),D_{\mathrm{in}},D_{\mathrm{out}})$ for $\bar{\cO}$, where $\cQ$ is an adaptive quantum circuit using quantum space $N=O(nw)$ and time $T=O(w)$.
\end{proposition}

The proof of Proposition~\ref{prop:errcorr} has a similar structure as the proof of Proposition~\ref{prop:switchdown} in Section~\ref{sec:csanalysis}, as well as the proof of Proposition~\ref{prop:stateprep} in Appendix~\ref{app:spproof}.

An important difference is that here in Proposition~\ref{prop:errcorr}, we must obtain a decorated output code that can have the specified family of bad sets $\cE_{\mathrm{out}}=\cE(G_{\mathrm{run}},\; \eta_{\mathrm{run}},\; \gamma_{\mathrm{out}})$ satisfying $\gamma_{\mathrm{out}}<\gamma_{\mathrm{in}}$. That is, here we must actually reduce the size of the error, rather than simply controlling its growth as in other gadgets. We accomplish this error reduction by arguing that if the output error were too large, then the small-set flip decoder in Algorithm~\ref{alg:errcorr} could not have terminated, which gives a contradiction. Related techniques are for instance used in \cite{fawzi_efficient_2018}. For readability, we have postponed the full proof to Appendix~\ref{app:ecproof}.

\section{Basic Gadgets}
\label{sec:basicgadgets}
In this section, we present basic gadgets that are (in some similar form) well-known or folklore in the literature. However, we provide formal result statements and brief proofs in our fault-tolerance framework so that we can integrate these gadgets with our other gadgets.

\subsection{$CNOT$ Gadgets}
We begin with a transversal $CNOT$ gadget between two identical codeblocks.

\begin{lemma}
  \label{lem:CNOTsame}
  For $r\in\bN$ and $1\leq i\leq r-1$, let $\cC^*$ be an $r$-dimensional cochain complex. Let $Q$ be the $[[n,k]]$ CSS code associated to level $i$ of $\cC^*$, with an associated CSS encoding map $\Enc$ (see Definition~\ref{def:cctodec}).

  Given some graph $G_{\mathrm{run}}=(V_{\mathrm{run}},E_{\mathrm{run}})$ with $C^i\subseteq V_{\mathrm{run}}$ and some $\eta_{\mathrm{in}},\gamma_{\mathrm{in}},\gamma_{\mathrm{run}}>0$, let
  \begin{align*}
    \gamma_{\mathrm{out}} &= 2\gamma_{\mathrm{in}}+\gamma_{\mathrm{run}} \\
    \cE_{\mathrm{in}} &= \cE(G_{\mathrm{run}},\; \eta_{\mathrm{in}},\; \gamma_{\mathrm{in}}) \\
    \cE_{\mathrm{run}} &= \cE(G_{\mathrm{run}},\; \eta_{\mathrm{in}},\; \gamma_{\mathrm{run}}) \\
    \cE_{\mathrm{in}} &= \cE(G_{\mathrm{run}},\; \eta_{\mathrm{in}},\; \gamma_{\mathrm{out}}).
  \end{align*}
  Define the decorated codes (see Definition~\ref{def:deccode}) $D_{\mathrm{in}}=(Q,\;\Enc,\;\cE_{\mathrm{in}})^{\sqcup 2}$ and $D_{\mathrm{out}}=(Q,\;\Enc,\;\cE_{\mathrm{out}})^{\sqcup 2}$. Let $\bar{\cO}:\bC^{2^{[k]\sqcup[k]}\times 2^{[k]\sqcup[k]}}\rightarrow\bC^{2^{[k]\sqcup[k]}\times 2^{[k]\sqcup[k]}}$ be the channel acting on two blocks of $k$ qubits that applies transversal $CNOT$ between the blocks, meaning that $CNOT_{(1,x),(2,x)}$ is applied for every $x\in[k]$, where our $2k$ logical qubits are labeled by the set $[2]\times[k]=[k]\sqcup[k]$. Equivalently, $\bar{\cO}(\rho)=CNOT^{\otimes k}\;\rho\; CNOT^{\otimes k}$.

  Then there exists a Pauli fault-tolerant gadget $((\cQ,(\cE_{\mathrm{run}}^{\sqcup 2})^{\sqcup T}),D_{\mathrm{in}},D_{\mathrm{out}})$ for $\bar{\cO}$, where $\cQ$ is an adaptive quantum circuit using quantum space $N=2n$ and time $T=1$.
\end{lemma}
\begin{proof}
  The desired adaptive quantum circuit $\cQ$ is simply the depth-$1$ circuit that applies transversal physical $CNOT$, meaning that physical $CNOT_{(1,c),(2,c)}$ is applied for every $c\in C^i$, where our $2n$ physical qubits are labeled by the set $[2]\times C^i=C^i\sqcup C^i$.

  In the absence of noise, this physical circuit implements the desired logical $CNOT$ gates, as for every $x,x'\in\bF_2^k$, then
  \begin{align*}
    CNOT^{\otimes n}\ket{\Enc(x)}\ket{\Enc(x')}
    &= \sum_{y\in\Enc(x),y'\in\Enc(x')}\ket{y}\ket{y+y'} \\
    &= \sum_{y\in\Enc(x),y''\in\Enc(x+x')}\ket{y}\ket{y''} \\
    &= \ket{\Enc(x)}\ket{\Enc(x+x')},
  \end{align*}
  where the second equality above holds because $\Enc:\bF_2^k\xrightarrow{\sim}H^i(\cC)=Z^i(\cC)/B^i(\cC)$ is a linear map.

  Now consider running the circuit $\cQ[\cF]$ with noise from a $\cE_{\mathrm{run}}^{\sqcup 2}$-avoiding fault $\cF$ and a $\cE_{\mathrm{in}}^{\sqcup 2}$-avoiding Pauli error on the input. Within every connected component of $G_{\mathrm{run}}\sqcup G_{\mathrm{run}}$ with $\geq\eta_{\mathrm{in}}$ vertices, at most $\gamma_{\mathrm{in}}$-fraction of the qubits receive input errors from the first codeblock, at most $\gamma_{\mathrm{in}}$-fraction of the qubits receive input errors from the second codeblock, and at most $\gamma_{\mathrm{run}}$-fraction of the qubits receive errors from the fault $\cF$. Note that the connected component must be contained within one of the two copies of $G_{\mathrm{run}}$, yet the physical $CNOT$ gates can propagate errors between qubits at the same positions in the two blocks, and hence we can have $\gamma_{\mathrm{in}}$-fraction corruptions from each block. Thus the total fraction of qubits in the connected component with Pauli errors is $\leq 2\gamma_{\mathrm{in}}+\gamma_{\mathrm{run}}$, as desired.
\end{proof}

We now present a transversal $CNOT$ between two codes of different dimensions. Below, recall from Definition~\ref{def:deccode} that we use $\sqcup$ to denote the combination of disjoint code blocks.

\begin{lemma}
  \label{lem:CNOTdiff}
  For $r\in\bN$ define $\cA^*$, $\cM^{\cA,*}$, $\Enc^{\cA}$, $\cB^*$, $\cM^{\cB,*}$, $L^{\cB,1}$, $\Enc^{\cB}$ $\cC^*=\cA^*\otimes\cB^*$, $\cM^{\cC,*}=\cM^{\cA,*}\otimes\cM^{\cB,*}$, $\Enc^{\cC}=\Enc^{\cA}\otimes\Enc^{\cB}$ as in Section~\ref{sec:csstate} (though in this lemma we make no assumptions regarding small-set flip decodability). Let $Q_{\mathrm{con}}$ be the $[[n_{\mathrm{con}},k_{\mathrm{con}}]]$ CSS code consisting of $|L^{\cB,1}|$ copies (labeled by the set $L^{\cB,1}$) of the CSS code associated to level $i-1$ of $\cA^*$, and let $Q_{\mathrm{tar}}$ be the $[[n_{\mathrm{tar}},k_{\mathrm{tar}}]]$ CSS code associated to level $i$ of $\cC^*$. The physical qubits of $Q_{\mathrm{tar}}$ are naturally labeled by the set $C^i$, while the physical qubits of $Q_{\mathrm{con}}$ are labeled by the set $A^{i-1}\times L^{\cB,1}\subseteq A^{i-1}\times B^1\subseteq C^i$.

  Given some graph $G_{\mathrm{run}}=(V_{\mathrm{run}},E_{\mathrm{run}})$ with $C^i\subseteq V_{\mathrm{run}}$ and some $\eta_{\mathrm{in}},\gamma_{\mathrm{in}},\gamma_{\mathrm{run}}>0$, let
  \begin{align*}
    \gamma_{\mathrm{out}} &= 2\gamma_{\mathrm{in}}+\gamma_{\mathrm{run}} \\
    \cE_{\mathrm{in}} &= \cE(G_{\mathrm{run}},\; \eta_{\mathrm{in}},\; \gamma_{\mathrm{in}}) \\
    \cE_{\mathrm{run}} &= \cE(G_{\mathrm{run}},\; \eta_{\mathrm{in}},\; \gamma_{\mathrm{run}}) \\
    \cE_{\mathrm{out}} &= \cE(G_{\mathrm{run}},\; \eta_{\mathrm{in}},\; \gamma_{\mathrm{out}}).
  \end{align*}
  Let $Q=Q_{\mathrm{con}}\sqcup Q_{\mathrm{tar}}$ with encoding map $\Enc=(\Enc^{\cA})^{L^{\cB,1}}\sqcup\Enc^{\cC}$. Define the decorated codes (see Definition~\ref{def:deccode}) $D_{\mathrm{in}}=(Q,\;\Enc,\;\cE_{\mathrm{in}}^{\sqcup 2})$ and $D_{\mathrm{out}}=(Q,\;\Enc,\;\cE_{\mathrm{out}}^{\sqcup 2})$, where we associate each set of physical qubits $A^{i-1}\times L^{\cB,1}\subseteq V_{\mathrm{run}}$ and $C^i\subseteq V_{\mathrm{run}}$ with a separate copy of $G_{\mathrm{run}}$.

  Let $\bar{\cO}:\bC^{2^{[k_{\mathrm{con}}]\sqcup[k_{\mathrm{tar}}]}\times 2^{[k_{\mathrm{con}}]\sqcup[k_{\mathrm{tar}}]}}\rightarrow\bC^{2^{[k_{\mathrm{con}}]\sqcup[k_{\mathrm{tar}}]}\times 2^{[k_{\mathrm{con}}]\sqcup[k_{\mathrm{tar}}]}}$ be the channel acting on $k_{\mathrm{con}}+k_{\mathrm{tar}}$ logical qubits labeled by $(M^{\cA,{i-1}}\times L^{\cB,1})\sqcup M^{\cC,i}$ that applies applies $CNOT_{x,x}$ for every $x\in M^{\cA,{i-1}}\times L^{\cB,1}\subseteq M^{\cC,i}$, with control in the first (size-$k_{\mathrm{con}}$) block and target in the second (size-$k_{\mathrm{tar}}$) block. Equivalently, $\bar{\cO}(\rho)=CNOT^{\otimes k_{\mathrm{con}}}\;\rho\; CNOT^{\otimes k_{\mathrm{con}}}$.

  Then there exists a Pauli fault-tolerant gadget $((\cQ,(\cE_{\mathrm{run}}^{\sqcup 2})^{\sqcup T}),D_{\mathrm{in}},D_{\mathrm{out}})$ for $\bar{\cO}$, where $\cQ$ is an adaptive quantum circuit using quantum space $N=n_{\mathrm{con}}+n_{\mathrm{tar}}$ and time $T=1$.
\end{lemma}
\begin{proof}
  The desired adaptive quantum circuit $\cQ$ is simply the depth-$1$ circuit that applies transversal physical $CNOT$ between the entire first (control) codeblock and the subset of the second (target) codeblock whose qubit labels are shared with the first codeblock. That is we apply physical $CNOT_{x,x}$ for every $x\in A^{i-1}\times L^{\cB,1}\subseteq C^i$, where the control is in the first (size-$n_{\mathrm{con}}$) block and the target is in the second (size-$n_{\mathrm{tar}}$) block.

  In the absence of noise, this physical circuit implements the desired logical $CNOT$ gates, as for every $x\in\cM^{\cA,i-1}\otimes\bF_2^{L^{\cB,1}}$, $x'\in\cM^{\cB,i}$, then
  \begin{align*}
    CNOT^{\otimes k_{\mathrm{con}}}\ket{(\Enc^{\cA})^{L^{\cB,1}}(x)}\ket{\Enc^{\cC}(x')}
    &= \sum_{y\in(\Enc^{\cA})^{L^{\cB,1}}(x),y'\in\Enc^{\cC}(x')}\ket{y}\ket{y+y'} \\
    &= \sum_{y\in(\Enc^{\cA})^{L^{\cB,1}}(x),y''\in\Enc^{\cC}(x+x')}\ket{y}\ket{y''} \\
    &= \ket{(\Enc^{\cA})^{L^{\cB,1}}(x)}\ket{\Enc^{\cC}(x+x')}.
  \end{align*}
  Here we view $x\in\cM^{\cA,i-1}\otimes\bF_2^{L^{\cB,1}}$ as also belonging to $\cM^{\cB,i}\supseteq\cM^{\cA,i-1}\otimes\bF_2^{L^{\cB,1}}$ (by simply padding $x$ with $0$s in the additional components), and similarly we view $y'\in\cA^{i-1}\otimes\bF_2^{L^{\cB,1}}$ as also belonging to $\cC^i$. The second equality above then holds because by definition $(\Enc^{\cA})^{L^{\cB,1}}(x)\subseteq\Enc^{\cC}(x)$.

  The analysis of $\cQ$ in the presence of noise, which proves the fault-tolerance of our gadget, is essentially identical to that in the proof of Lemma~\ref{lem:CNOTsame}, so we omit it to avoid redundancy.
\end{proof}

\subsection{Hadamard Gadget}
We now present a transversal Hadamard gadget on a codeblock. Note that this gadget has different (specifically, dual) input and output codes. We will later combine this gadget with our code switching gadgets to obtain a constant-overhead gadget for logical Hadamard that has the same input and out codes (see Proposition~\ref{prop:hadsame}).

\begin{lemma}
  \label{lem:hadamard}
  For $r\in\bN$ and $1\leq i\leq r-1$, let $\cC^*$ be an $r$-dimensional cochain complex. Let $Q_{\mathrm{in}}$ be the $[[n,k]]$ CSS code associated to level $i$ of the cochain complex $\cC^*$, with an associated CSS encoding map $\Enc_{\mathrm{in}}$ (see Definition~\ref{def:cctodec}). Let $Q_{\mathrm{out}}$ be the $[[n,k]]$ CSS code associated to level $i$ of the chain complex $\cC_*$, with associated CSS encoding map $\Enc_{\mathrm{out}}$. Assume that for every $x,x'\in\bF_2^k$ and every $y\in\Enc_{\mathrm{in}}(x),\;y'\in\Enc_{\mathrm{out}}(x')$, we have
  \begin{equation}
    \label{eq:dualenc}
    x\cdot x' = y\cdot y'
  \end{equation}

  Given some graph $G_{\mathrm{run}}=(V_{\mathrm{run}},E_{\mathrm{run}})$ with $C^i\subseteq V_{\mathrm{run}}$ and some $\eta_{\mathrm{in}},\gamma_{\mathrm{in}},\gamma_{\mathrm{run}}>0$, let
  \begin{align*}
    \gamma_{\mathrm{out}} &= \gamma_{\mathrm{in}}+\gamma_{\mathrm{run}} \\
    \cE_{\mathrm{in}} &= \cE(G_{\mathrm{run}},\; \eta_{\mathrm{in}},\; \gamma_{\mathrm{in}}) \\
    \cE_{\mathrm{run}} &= \cE(G_{\mathrm{run}},\; \eta_{\mathrm{in}},\; \gamma_{\mathrm{run}}) \\
    \cE_{\mathrm{in}} &= \cE(G_{\mathrm{run}},\; \eta_{\mathrm{in}},\; \gamma_{\mathrm{out}}).
  \end{align*}
  Define the decorated codes (see Definition~\ref{def:deccode}) $D_{\mathrm{in}}=(Q_{\mathrm{in}},\;\Enc_{\mathrm{in}},\;\cE_{\mathrm{in}})$ and $D_{\mathrm{out}}=(Q_{\mathrm{out}},\;\Enc_{\mathrm{out}},\;\cE_{\mathrm{out}})$. Let $\bar{\cO}:\bC^{2^k\times 2^k}\rightarrow\bC^{2^k\times 2^k}$ be the channel acting on $k$ qubits that applies a Hadamard gate to every qubit. Equivalently, $\bar{\cO}(\rho)=H^{\otimes k}\;\rho\; H^{\otimes k}$.

  Then there exists a Pauli fault-tolerant gadget $((\cQ,\cE_{\mathrm{run}}^{\sqcup T}),D_{\mathrm{in}},D_{\mathrm{out}})$ for $\bar{\cO}$, where $\cQ$ is an adaptive quantum circuit using quantum space $N=n$ and time $T=1$.
\end{lemma}
\begin{proof}
  The desired adaptive quantum circuit $\cQ$ is simply the depth-$1$ circuit that applies transversal physical Hadamard, that is, applies $H^{\otimes n}$.

  In the absence of noise, this physical circuit implements the desired logical $H$ gates, as for every $x\in\bF_2^k$, then
  \begin{align*}
    H^{\otimes n}\Enc_{\mathrm{in}}\ket{x}
    &= \sum_{y\in\Enc_{\mathrm{in}}(x)}H^{\otimes n}\ket{y} \\
    &\propto \sum_{y\in\Enc_{\mathrm{in}}(x)}\sum_{y'\in\bF_2^n}(-1)^{y\cdot y'}\ket{y'} \\
    &= \sum_{y\in\Enc_{\mathrm{in}}(x),\; y'\in Z_i(\cC)}(-1)^{y\cdot y'}\ket{y'},
  \end{align*}
  while
  \begin{align*}
    \Enc_{\mathrm{out}}H^{\otimes k}\ket{x}
    &\propto \Enc_{\mathrm{out}}\sum_{x'\in\bF_2^k}(-1)^{x\cdot x'}\ket{x'} \\
    &\propto \sum_{x'\in\bF_2^k}(-1)^{x\cdot x'}\sum_{y'\in\Enc_{\mathrm{out}}(x')}\ket{y'} \\
    &\propto \sum_{y\in\Enc_{\mathrm{in}}(x)}\sum_{y'\in Z_i(\cC)}(-1)^{y\cdot y'}\ket{y'},
  \end{align*}
  where the final $\propto$ above holds by~(\ref{eq:dualenc}). The two equations above have the same RHS, so we must have $H^{\otimes n}\Enc_{\mathrm{in}}=\Enc_{\mathrm{out}}H^{\otimes k}$. Thus $\cQ=H^{\otimes n}$ implements the desired logical $H^{\otimes k}$ in the absence of noise.

  Now consider running the circuit $\cQ[\cF]$ with noise from a $\cE_{\mathrm{run}}$-avoiding fault $\cF$ and a $\cE_{\mathrm{in}}$-avoiding Pauli error on the input. Within every connected component of $G_{\mathrm{run}}$ with $\geq\eta_{\mathrm{in}}$ vertices, at most $\gamma_{\mathrm{in}}$-fraction of the qubits receive input errors, and at most $\gamma_{\mathrm{run}}$-fraction of the qubits receive fault errors. Thus the total fraction of qubits in the connected component with Pauli errors is $\leq\gamma_{\mathrm{in}}+\gamma_{\mathrm{run}}$, as desired.
\end{proof}

\begin{remark}
  \label{remark:dualenc}
  Note that the condition~(\ref{eq:dualenc}) in Lemma~\ref{lem:hadamard} is by definition satisfied for $\Enc_{\mathrm{in}}=\Enc^i$, $\Enc_{\mathrm{out}}=\Enc_i$ for $\Enc^i$, $\Enc_i$ defined as in Lemma~\ref{lem:enc1dim}, and therefore also for the product maps $\Enc^i$, $\Enc_i$ defined as in Definition~\ref{def:encrdim}.
\end{remark}

\subsection{Logical Pauli $X$ and $Z$ Measurement Gadgets}
\label{sec:measure}
In this section, we show that measuring out an entire code block in the Pauli $Z$ basis induces logical $Z$ measurements on all of the encoded logical qubits. The same analysis implies an analogous result for Pauli $X$.

\begin{algorithm}
  \caption{\label{alg:measure} Logical Pauli $Z$ measurement gadget in Lemma~\ref{lem:measure}; all variables are defined as in the lemma. Here we write $\delta=\delta^{\cC}$.}
  \SetKwInOut{Input}{Input}
  \SetKwInOut{Output}{Output}

  \SetKwFunction{FnMeasure}{Measure}
  \SetKwProg{Fn}{Function}{:}{}

  \Input{$\sigma\in\bC^{\cC^i\times\cC^i}$ that is a $\cE_{\mathrm{in}}$-deviation of $\Enc(\rho)$ for $\rho\in\bC^{2^k\times 2^k}$}
  \Output{Classical string $x\in\bF_2^k$ giving outcome of Pauli $Z$ measurements on $\rho$}

  \Fn{\FnMeasure{$\sigma;i,\cC^*$}}{
    Measure Pauli $Z$ on all $n$ qubits, and let $z\in\cC^i$ be the measurement outcome \\ \label{li:mem}
    Run \FnSSFlipSyn{$\delta(z);i,\cC^*$} from Algorithm~\ref{alg:ssflip}, and let $a^i\in\cC^i$ be the output \\ \label{li:medec}
    \Return{$\Enc^{-1}(z+a^i)$} \label{li:meret}
  }
\end{algorithm}

\begin{lemma}
  \label{lem:measure}
  For $r\in\bN$ and $1\leq i\leq r-1$, let $\cC^*$ be an $r$-dimensional cochain complex of locality $w$ with a $(m,0)$-small-set flip decoder (see Definition~\ref{def:ssflip}) at level $i$. Let $Q$ be the $[[n,k]]$ CSS code associated to level $i$ of $\cC^*$, with an associated CSS encoding map $\Enc$ (see Definition~\ref{def:cctodec}). Given some graph $G_{\mathrm{in}}=(V_{\mathrm{in}},E_{\mathrm{in}})$ that contains $G^{\cC}_{i-1,i,i+1}\subseteq G_{\mathrm{in}}$ as a subgraph and some
  \begin{align*}
    \eta_{\mathrm{in}} &\leq \min\{m,\; d_i(\cC)\} \\
    \gamma_{\mathrm{in}} &\leq \frac{1}{50w^7},
  \end{align*}
  let
  \begin{align*}
    \cE_{\mathrm{in}} &= \cE(G_{\mathrm{in}},\; \eta_{\mathrm{in}},\; \gamma_{\mathrm{in}}),
  \end{align*}
  and define the decorated code (see Definition~\ref{def:deccode}) $D_{\mathrm{in}}=(Q,\;\Enc,\;\cE_{\mathrm{in}})$.

  Let $\bar{\cO}:\bC^{2^k\times 2^k}\rightarrow\bC^{2^k\times 2^k}$ be the channel that takes as input $k_{\mathrm{in}}=k$ logical qubits (and $M_{\mathrm{in}}=0$ classical bits) and outputs $M_{\mathrm{out}}=k$ classical bits (and $k_{\mathrm{out}}=0$ qubits), which is given by measuring the $k$ input qubits in the Pauli $Z$ basis, and outputting the measurement outcomes (see Definition~\ref{def:gadget}).

  Then there exists a Pauli fault-tolerant gadget $((\cQ,\emptyset),D_{\mathrm{in}},\emptyset)$ for $\bar{\cO}$, where $\cQ$ is an adaptive quantum circuit using quantum space $N=n$ and time $T=2$.
\end{lemma}
\begin{proof}
  The proof will be similar to that of Proposition~\ref{prop:switchdown}, but even simpler as here there is no product involved. We will therefore omit some details to avoid redundancy.

  The desired adaptive quantum circuit $\cQ$ is given in Algorithm~\ref{alg:measure}. By definition line~\ref{li:mem} and line~\ref{li:medec} each use $1$ timestep, so the entire circuit uses $T=2$ timesteps. Furthermore, the circuit by definition uses $N=|C^i|=n$ physical qubits.

  Our goal is to show that for every $\ell\in\bN$, every $\rho\in\bC^{2^{k+\ell}\times 2^{k+\ell}}$, and every Pauli $\cE_{\mathrm{in}}$-deviation\footnote{See Footnote~\ref{footnote:sdabuse}.} $\sigma=E\Enc(\rho)E'$ (for Paulis $E,E'$ such that $\supp(E)\cup\supp(E')\subseteq C^i$ is $\cE_{\mathrm{in}}$-avoiding), the resulting output $\cQ(\sigma)\in\bF_2^k$ is the result of measuring the first ($k$-qubit) register of $\rho$. Note that by Definition~\ref{def:fault}, a fault $\cF$ will only act on the post-measurement qubits in Algorithm~\ref{alg:measure}, which are discarded (i.e.~traced out) anyways. Hence the output of Algorithm~\ref{alg:measure} remains the same under arbitrary faults.

  Write $E=Z^{e_Z}X^{e_X}$, $E'=X^{e_X'}Z^{e_Z'}$, and expand
  \begin{align*}
    \rho
    &= \sum_{x,x'\in\bF_2^k}\rho_{x,x'}\ket{x}\bar{x'}\otimes\nu_{x,x'}
  \end{align*}
  for some $\nu_{x,x'}\in\bC^{2^\ell\times 2^\ell}$. Then the input state just prior to the measurements in line~\ref{li:mem} of Algorithm~\ref{alg:measure} is
  \begin{align*}
    E\Enc(\rho)E'
    &= \sum_{x,x',y,y'}\rho_{x,x'}Z^{e_Z}\ket{y+e_X}\bra{y'+e_X'}Z^{e_Z'}\otimes\nu_{x,x'}
  \end{align*}
  where the sum above is over all $x,x'\in\bF_2^k$, $y\in\Enc(x)$, and $y'\in\Enc(x')$. Performing the measurements in line~\ref{li:mem} then collapses the state to
  \begin{equation}
    \label{eq:mepm}
    \sum_{z\in C^i}\sum_{x,x',y,y'}\rho_{x,x'}Z^{e_Z}\ket{z}\bra{z}Z^{e_Z'}\otimes\nu_{x,x'},
  \end{equation}
  where the inner sum above is over all $x,x'\in\bF_2^k$, $y\in\Enc(x)$, $y'\in\Enc(x')$ for which
  \begin{equation}
    \label{eq:mez}
    y+e_X = z = y'+e_X'.
  \end{equation}

  Let $S_{\mathrm{in}}=\supp(E)\cup\supp(E')$, so that by definition $S_{\mathrm{in}}$ is $\cE_{\mathrm{in}}$-avoiding.

  Let $S_{\mathrm{syn}}\subseteq C^i\sqcup C^{i+1}$ be the footprint (see Definition~\ref{def:footprint}) of the call to \FnSSFlipSyn{$\delta(z);i,\cC^*$} in line~\ref{li:medec} of Algorithm~\ref{alg:measure}.

  \begin{claim}
    \label{claim:mesynavoid}
    The set $S_{\mathrm{in}}\cup S_{\mathrm{syn}}$ is $\cE(G_{\mathrm{in}},\; \eta_{\mathrm{in}},\; 5w^3\gamma_{\mathrm{in}})$-avoiding.
  \end{claim}

  The proof of Claim~\ref{claim:mesynavoid} is nearly identical to that of Claim~\ref{claim:sdsynavoid} (just with the subscript ``in'' replacing the subscript ``run,'' and the coefficient $5w^3$ replacing $15w^3$ as by definition here $S_{\mathrm{in}}$ is $\cE_{\mathrm{in}}$-avoiding), so we omit it to avoid redundancy.

  Furthermore, similarly as described following Claim~\ref{claim:sdsynavoid} in the proof of Proposition~\ref{prop:switchdown}, because $\gamma_{\mathrm{in}}\leq 1/5w^3$, the subgraph of $G_{\mathrm{in}}$ induced by $S_{\mathrm{in}}\cup S_{\mathrm{syn}}$ has no connected components containing $\geq\eta_{\mathrm{in}}$ vertices. Therefore because $\eta_{\mathrm{in}}\leq m$, within every such connected component, the restricted execution \FnSSFlipSyn{$\delta(z|_{V\cap C^i});i,\cC^*$} outputs $a^i|_{V\cap C^i}$, so that $\delta(z+a^i|_{V\cap C^i})=\delta(e_X+a^i|_{V\cap C^i})=0$. We analogously have $\delta(e_X'+a^i|_{V\cap C^i})=0$.

  We now define $S_{\mathrm{err}}$ and $\tilde{a}^{i-1},\tilde{a}^i$ (resp.~$S_{\mathrm{err}}'$ and $(\tilde{a}^{i-1})',(\tilde{a}^i)'$) to be the footprint and output of \FnSSFlipErr{$e_X+a^i;i,\cC^*$} (resp.~\FnSSFlipErr{$e_X'+a^i;i,\cC^*$}), respectively. We prove the following claim for the unprimed variables $e_X,\dots$, but it analogously applies to the primed variables $e_X',\dots$.

  \begin{claim}
    \label{claim:meerravoid}
    The set $S_{\mathrm{in}}\cup S_{\mathrm{syn}}\cup S_{\mathrm{err}}$ is $\cE(G_{\mathrm{in}},\; \eta_{\mathrm{in}},\; 50w^7\gamma_{\mathrm{in}})$-avoiding.
  \end{claim}

  The proof of Claim~\ref{claim:meerravoid} is nearly identical to that of Claim~\ref{claim:sderravoid}, so we omit it to avoid redundancy.

  Furthermore, similarly as described following Claim~\ref{claim:sderravoid} in the proof of Proposition~\ref{prop:switchdown}, we must have $\tilde{a}^i=0$. Also, because $\gamma_{\mathrm{in}}\leq 1/50w^7$, we conclude that the subgraph of $G_{\mathrm{in}}$ induced by $S_{\mathrm{in}}\cup S_{\mathrm{syn}}\cup S_{\mathrm{err}}$ has no connected components containing $\geq\eta_{\mathrm{in}}$ vertices. Then because $\eta_{\mathrm{in}}\leq m$, by the assumption that $\cC^*$ has an $m$-small-set error-flip decoder at level $i$, we conclude that within every connected component $V$ we have $\delta(\tilde{a}^{i-1}|_{V\cap C^{i-1}})=e_X+a^i|_{V\cap C^i}$, and thus
  \begin{equation}
    \label{eq:meta}
    \delta(\tilde{a}^{i-1}) = e_X+a^i.
  \end{equation}
  As the analysis above also applies to the primed variables, we also have
  \begin{equation*}
    \delta((\tilde{a}^{i-1})') = e_X'+a^i.
  \end{equation*}

  The above equations combined with~(\ref{eq:mez}) imply that
  \begin{equation*}
    y+y' = e_X+e_X' = \delta(\tilde{a}^{i-1}+(\tilde{a}^{i-1})') \in B^i(\cC).
  \end{equation*}
  Thus we must have $x=x'$, so if the post-measurement state in~(\ref{eq:mepm}) is nonzero, then we must have $e_X+e_X'\in B^i(\cC)$, and the state can be equivalently written as
  \begin{equation}
    \label{eq:mefin}
    \sum_{x\in\bF_2^k}\; \sum_{y\in\Enc(x)}\rho_{x,x'}Z^{e_Z}\ket{y+e_X}\bra{y+e_X}Z^{e_Z'}\otimes\nu_{x,x'}.
  \end{equation}

  Now similarly as in the proof of Proposition~\ref{prop:switchdown}, if $e_Z+e_Z'\notin Z_i(\cC)=B^i(\cC)^\perp$, then the phases induced by the Paulis $Z^{e_Z},Z^{e_Z'}$ in the inner sum in~(\ref{eq:mefin}) cancel, and the expression becomes $0$. Thus assume $e_Z+e_Z'\in Z_i(\cC)$. Now because $\supp(e_Z+e_Z')\subseteq S_{\mathrm{in}}$ is $\cE_{\mathrm{in}}$-avoiding, the subgraph of $G^{\cC}_i\subseteq G_{\mathrm{in}}$ induced by this set contains no connected components with $\geq\eta_{\mathrm{in}}$ vertices. Therefore by the assumption that $\eta_{\mathrm{in}}\leq d_i(\cC)$, we cannot have $e_Z+e_Z'\in Z_i(\cC)$ belonging to a nontrivial homology class, that is, we msut have $e_Z+e_Z'\in B_i(\cC)$. Thus all terms in the sum in~(\ref{eq:mefin}) receive the same phase $(-1)^{(e_Z+e_Z')\cdot e_X}$ from the Paulis $Z^{e_Z},Z^{e_Z'}$, so we can remove these Paulis while simply inducing a global phase on the overall expression. Tracing out the code register of the resulting state yields
  \begin{equation*}
    \sum_{x\in\bF_2^k}\rho_{x,x}\nu_{x,x} = \tr_{[k]}(\rho),
  \end{equation*}
  where for a given $x\in\bF_2^k$, by~(\ref{eq:meta}), Algorithm~\ref{alg:measure} will return
  \begin{equation*}
    \Enc^{-1}(z+a^i) = \Enc^{-1}(y+e_X+a^i) = \Enc^{-1}(y+B^i(\cC)) = x,
  \end{equation*}
  as desired.
\end{proof}




\section{Applications}
\label{sec:apply}
In this section, we combine our gadgets from the sections above to create new gadgets for constant-overhead fault-tolerant logical operations.


In Section~\ref{sec:applycode} below, we describe the lossless-expander-based codes with which we instantiate our gadgets. The subsequent sections apply these codes to construct our desired fault-tolerant gadgets in Theorem~\ref{thm:main}. Recall that Proposition~\ref{prop:stateprep} implies item~\ref{it:mainprep} of Theorem~\ref{thm:main}. Section~\ref{sec:appperm} and Section~\ref{sec:hadsame} below prove item~\ref{it:mainperm} and item~\ref{it:mainHa} respectively in Theorem~\ref{thm:main}. Section~\ref{sec:apptarget} then proves item~\ref{it:mainHt} and item~\ref{it:mainCNOTt} in Theorem~\ref{thm:main}. Section~\ref{sec:app2dprep} highlights a surprising consequence of item~\ref{it:mainprep} in Theorem~\ref{thm:main} (that follows from Proposition~\ref{prop:stateprep} and Proposition~\ref{prop:switchdown}), namely, that we can prepare 2-dimensional product code states in bulk with constant space-time overhead.


\begin{remark}
  As Theorem~\ref{thm:main} is stated with fault-tolerance defined in the sense of a threshold under locally stochastic noise, it follows from the results listed above combined with Proposition~\ref{prop:errcorr}, Lemma~\ref{lem:redPauli}, and Lemma~\ref{lem:threshold}.
\end{remark}

\subsection{Code Instantiation}
\label{sec:applycode}
In this section, we describe the codes arising from tensor products of 1-dimensional cochain complexes from lossless expanders that we will use to instantiate our gadgets throughout the remainder of Section~\ref{sec:apply}.

We first fix some $r\in\bN$ and $\epsilon>0$, and we define $\epsilon=\epsilon(r,1/4)$ to be the value defined in Proposition~\ref{prop:ssflip}. For this value of $\epsilon$, we then define $\mu,R,\Delta_L,\Delta_R$ as in Corollary~\ref{cor:losslessfamily}, and we fix some $(\Delta_L,\Delta_R)$-biregular $(\mu,\epsilon)$-lossless expander $G=(V_L\sqcup V_R,E)$ with specified set\footnote{The set $L^{G,1}$ here is called $S_i$ in Corollary~\ref{cor:losslessfamily}.} $L^{G,1}\subseteq V_R$ from the family in Corollary~\ref{cor:losslessfamily}.

Following the notation in Lemma~\ref{lem:enc1dim}, let $\cC^{G,*}$ denote the 1-dimensional cochain complex associated to $G$. We fix some information sets $M^{G,0}\subseteq V_L=C^{G,0}$, $M^{G,1}\subseteq V_R=C^{G,1}$ for $Z^0(\cC^G)$, $Z_1(\cC^G)$ respectively such that $L^{G,1}\subseteq M^{G,1}$. We then define $\cM^{G,*}$ and $\Enc^G:\cM^{G,*}\rightarrow\cC^{G,*}$ as in Lemma~\ref{lem:enc1dim}.

In this section, we consider codes associated to cochain complexes of the form $\cC^*={\cC^{(1)}}^*\otimes\cdots\otimes{\cC^{(r)}}^*$ with encoding map $\Enc^{\cC}=\Enc^{(1)}\otimes\cdots\otimes\Enc^{(r)}$ as defined in Definition~\ref{def:encrdim}, where each cochain complex ${\cC^{(h)}}^*$ equals either $\cC^{G,*}$ or its dual complex $\cC^G_*$. By Proposition~\ref{prop:ssflip}, this cochain complex $\cC^*$ (resp.~the chain complex $\cC_*$) has a $(m,\gamma_{\mathrm{dec}})$-small-set flip decoder at every level $0\leq i\leq r-1$ (resp.~$1\leq i\leq r$) for $m=\mu|V_L|/(r\Delta_L^{r+1}+1)$ and $\gamma_{\mathrm{dec}}=1/10$.

By Corollary~\ref{cor:losslessfamily} and Proposition~\ref{prop:ssflip}, we similarly have a $(m,\gamma_{\mathrm{dec}})$-small-set flip decoder for the cochain complex ${\cC^{(1)}}^*\otimes\cdots\otimes{\cC^{(r-1)}}^*\otimes\cB_{\bar{L}}^*$ (and its associated chain complex), where $\cB_{\bar{L}}^*$ is defined as in Section~\ref{sec:csstate} for $\cB={\cC^{(r)}}^*$ and $L^{\cB,1}=L^{G,1}$. Note that this statement holds when replacing any of the $r$ factors ${\cC^{(h)}}^*$ with $\cB_{\bar{L}}^*$.

Because the classical codes $Z^i(\cC^G),Z_i(\cC^G)$ have distance $\geq\mu|V_L|$ by Lemma~\ref{lem:losslesstoun}, the quantum code at every level $1\leq i\leq r-1$ of $\cC^*$ must have distance $\min\{d^i(\cC),d_i(\cC)\}\geq\mu|V_L|$ (see e.g.~\cite[Lemma~3.36]{golowich_quantum_2025-1}, which in turn is based on arguments in \cite{tillich_quantum_2014,bravyi_homological_2014}).

Here we think of $|V_L|=\Theta(|V_R|)$ as growing while $r$ and the resulting $r\mu,R,\Delta_L,\Delta_R$ remain constant. Hence the code associated to some level $1\leq i\leq r-1$ of $\cC^*$ has length
\begin{equation*}
  n = \Theta_r(|V_L|^{r})
\end{equation*}
and distance
\begin{equation*}
  d\geq\mu|V_L|=\Omega_r(n^{1/r}).
\end{equation*}
We will always assume that $i$ of the values $h\in[r]$ have ${\cC^{(h)}}^*=\cC^{G,*}$ and $r-i$ of these values have ${\cC^{(h)}}^*=\cC^G_*$, so that the code at level $i$ has dimension
\begin{equation*}
  k\geq(R|V_R|)^{r}=\Theta_r(n)
\end{equation*}
by Corollary~\ref{cor:losslessfamily} and Proposition~\ref{prop:kunneth}. Also by definition, the locality $w$ of $\cC^*$ does not depend on $|V_L|=\Theta(|V_R|)$, so that
\begin{equation*}
  w \leq O_r(1).
\end{equation*}

\subsection{Logical Qubit Permutations}
\label{sec:appperm}
In this section, we apply our code switching gadgets to perform highly flexible and parallel logical permutations with constant space-time overhead. We begin with a gadget that arbitrarily permutes the ``slabs,'' or axis-parallel codimension-1 hyperplanes, within a set (of arbitrary size $\kappa\in\bN$) of $r$-dimensional hypercubes $(L^{G,1})^r$ of logical qubits.

\begin{proposition}
  \label{prop:permslab}
  For arbitrary fixed $r\geq 3$, define all variables as in Section~\ref{sec:applycode}, and let $Q$ be the $[[n,k]]$ code associated to some level $1\leq i\leq r-1$ of $\cC^*={\cC^{(1)}}^*\otimes\cdots\otimes{\cC^{(r)}}^*$ with encoding map $\Enc^{\cC}$ as in Section~\ref{sec:applycode}.

  There exist sufficiently small $\gamma_{\mathrm{in}}(r)>0$ and $\zeta_{\mathrm{run}}(r)>0$ such that for every
  \begin{align*}
    \eta_{\mathrm{in}} &\leq m \\
    \gamma_{\mathrm{in}} &\leq \gamma_{\mathrm{in}}(r) \\
    \gamma_{\mathrm{run}} &\leq \zeta_{\mathrm{run}}(r)\cdot\gamma_{\mathrm{in}},
  \end{align*}
  the following holds.

  Let $G_{\mathrm{run}}=(V_{\mathrm{run}},E_{\mathrm{run}}):=G^{\cC}_{[r]}$ and
  \begin{align*}
    \cE_{\mathrm{in}} = \cE_{\mathrm{out}} &= \cE(G_{\mathrm{run}},\; \eta_{\mathrm{in}},\; \gamma_{\mathrm{in}}) \\
    \cE_{\mathrm{run}} &= \cE(G_{\mathrm{run}},\; \eta_{\mathrm{in}},\; \gamma_{\mathrm{run}}),
  \end{align*}
  and define the decorated code
  \begin{equation*}
    D = (Q,\; \Enc^{\cC},\; \cE_{\mathrm{in}}).
  \end{equation*}
  Here the $n=|C^i|$ physical qubits of $Q$ are naturally associated with the set $C^i\subseteq V_{\mathrm{run}}$.

  For every $\kappa\in\bN$, every $h\in[r]$, and every permutation\footnote{Recall that by definition $(L^{G,1})^{\sqcup\kappa}=[\kappa]\times L^{G,1}$.} $\pi:(L^{G,1})^{\sqcup\kappa}\rightarrow (L^{G,1})^{\sqcup\kappa}$, there exists a Pauli fault-tolerant gadget $((\cQ,((\cE_{\mathrm{run}}^{\sqcup\kappa})^{\sqcup u})^{\sqcup T}),D^{\sqcup\kappa},D^{\sqcup\kappa})$ for $\bar{\cO}$, where:
  \begin{itemize}
  \item $\cQ$ is an adaptive quantum circuit using quantum space $N=O_r(\kappa\cdot n)$ and time $T=O_r(1)$, and $u=O(1)$.
  \item $\bar{\cO}:\bC^{(\cM^{\cC,i})^{\kappa}\times(\cM^{\cC,i})^{\kappa}}\rightarrow\bC^{(\cM^{\cC,i})^{\kappa}\times(\cM^{\cC,i})^{\kappa}}$ is a channel that acts on $\kappa\cdot k$ qubits labeled by the set $(M^{\cC,i})^{\sqcup\kappa}$, that permutes the qubits in $((L^{G,1})^r)^{\sqcup\kappa}\subseteq(M^{\cC,i})^{\sqcup\kappa}$ as follows: under the isomorphism $((L^{G,1})^r)^{\sqcup\kappa}\cong(L^{G,1})^{h-1}\times(L^{G,1})^{\sqcup\kappa}\times(L^{G,1})^{r-h}$, then $\bar{\cO}$ permutes the qubits according to $I^{\otimes h-1}\otimes\pi\otimes I^{\otimes r-h}$ (so that the qubit at position $(b_1,\dots,b_h,\dots,b_r)\in(L^{G,1})^{h-1}\times(L^{G,1})^{\sqcup\kappa}\times(L^{G,1})^{r-h}$ is moved to $(b_1,\dots,\pi(b_h),\dots,b_r)$); the qubits in $(M^{\cC,i})^{\sqcup\kappa}\setminus((L^{G,1})^r)^{\sqcup\kappa}$ may be moved or reset arbitrarily.
  \end{itemize}
\end{proposition}

Analogously as described in Remark~\ref{remark:suu}, the gadget $\cQ$ in Proposition~\ref{prop:permslab} invokes a constant number of gadgets previously presented in the paper. As a result, the physical qubits are naturally partitioned into some constant number $u=O(1)$ of blocks, each of which is naturally a subset of $C^{\sqcup\kappa}=V_{\mathrm{run}}^{\sqcup\kappa}$.

\begin{proof}[Proof of Proposition~\ref{prop:permslab}]
  We first present the desired circuit $\cQ$ assuming that ${\cC^{(h)}}^*=\cC^{G,*}$ and $i\geq 2$; we then show how to relax these assumptions to also allow for ${\cC^{(h)}}^*=\cC^G_*$ or $i=1$.

  By the assumptions above, we may apply Proposition~\ref{prop:switchdown} with $\cA^*=\bigotimes_{h'\in[r]\setminus\{h\}}{\cC^{(h')}}^*$ and $\cB^*={\cC^{(h)}}^*$ to map each of the $\kappa$ input $Q$-codeblocks down to a set of $|L^{G,1}|$ codeblocks of the code $Q'$ associated to level $i-1$ of $\cA^*$ (so we obtain a total of $\kappa\cdot|L^{G,1}|$ codeblocks of $Q'$).

  To apply the logical permutation $I^{\otimes h-1}\otimes\pi\otimes I^{\otimes r-h}$ to this state, we first prepare an ancilla block of qubits, also labeled by the set $(A^{i-1}\times L^{G,1})^{\sqcup\kappa}$, and initialized arbitrarily. We swap our $\kappa\cdot|L^{G,1}|$ codeblocks of $Q'$ into this ancilla block, and then swap the qubits back into the original block. However, when performing this second round of swaps, each $Q'$-block (of which there are $\kappa\cdot|L^{G,1}|$) with label $b\in(L^{G,1})^{\sqcup\kappa}$ is moved into the location with label $\pi(b)\in(L^{G,1})^{\sqcup\kappa}$. Each physical swap here is implemented with three $CNOT$ gates. We subsequently apply Proposition~\ref{prop:switchup} to each of the $\kappa$ sets of $|L^{G,1}|$ codeblocks of $Q'$ to return to $\kappa$ codeblocks of $Q$. Finally, we apply the error correction gadget in Proposition~\ref{prop:errcorr} to each resulting $Q$-codeblock.

  If instead ${\cC^{(h)}}^*=\cC^{G}_*$ and $i\leq r-2$, an analogous circuit as described above works, with simply the roles of the $X$ and $Z$ bases swapped. The only remaining cases are when ${\cC^{(h)}}^*=\cC^{G,*}$ and $i=1$, or when ${\cC^{(h)}}^*=\cC^G_*$ and $i=r-1$; we consider the former, as the latter is analogous.

  Specifically, in this case, we cannot switch down with $\cA^*$ and $\cB^*$ as described above, as then the qubits would be in the code associated to level $0$ of $\cA^*$, which as poor distance and small-set flip decodability in the $X$-basis. Instead, for each input $Q$-codeblock, we first apply Proposition~\ref{prop:switchdown} followed by Proposition~\ref{prop:switchup} to change one of the other factors $h'\neq h\in[r]$ from ${\cC^{(h')}}^*=\cC^G_*$ to $\cC^{G,*}$, thereby bringing our logical qubits into the code associated to level $2$ of our $r$-dimensional complex. We can then apply Proposition~\ref{prop:switchdown}, the swap gates (see above), and Proposition~\ref{prop:switchup} as described above to induce the desired logical permutation, where these applications switch the $h$th factor. Finally, we switch back to the original code $Q$ with another more application of Proposition~\ref{prop:switchdown} followed by Proposition~\ref{prop:switchup}, where these final applications switch back the $h'$th factor to its original status. We then finish by applying the error correction gadget in Proposition~\ref{prop:errcorr} to each resulting $Q$-codeblock.

  We will now show that the results referenced above imply that this resulting circuit $\cQ$ yields a Pauli fault-tolerant gadget $((\cQ,((\cE_{\mathrm{run}}^{\sqcup\kappa})^{\sqcup u})^{\sqcup T}),D^{\sqcup\kappa},D^{\sqcup\kappa})$ for $\bar{\cO}$ using quantum space $N=O_r(\kappa\cdot n)$ and time $T=O_r(1)$, as desired.

  The two layers of swap gates described above each by definition use space $\leq 2\kappa n$ and time $3$, and cannot propagate errors between different physical qubits in $(A^{i-1}\times L^{G,1})^{\sqcup\kappa}$. Thus if the Pauli error on each $A^{i-1}$-block (of which there are $\kappa\cdot|L^{G,1}|$) within these qubits is $\cE(G^{\cA}_{[r-1]},\;\eta_{\mathrm{in}},\;\gamma)$-avoiding just prior to performing the swaps, then the error is $\cE(G^{\cA}_{[r-1]},\;\eta_{\mathrm{in}},\;\gamma+12\gamma_{\mathrm{run}})$-avoiding following the swaps. The factor of $12$ here arises because the swaps are comprised of $6$ layers of $CNOT$ gates, and in each of these $6$ timesteps there are two $(A^{i-1}\times L^{G,1})^{\sqcup\kappa}$-blocks of physical qubits that could experience an error from the fault.

  Combining the above reasoning with Proposition~\ref{prop:switchdown}, Proposition~\ref{prop:switchup}, and Proposition~\ref{prop:errcorr} then yields the desired result. First, it follows immediately that $\cQ$ uses quantum space $N=O_r(\kappa\cdot n)$ and time $T=O_r(1)$. Furthermore, by the fault-tolerance guarantees in these propositions and of the swap gates described above, every gadget invoked within $\cQ$ prior to the error correction gadget will take as input a set of codeblocks, each of which is a Pauli $\cE(G,\; \eta_{\mathrm{in}},\; \gamma)$-deviation of the desired logical state at that timestep for some $\gamma=O_r(\gamma_{\mathrm{in}})$. Here $G=G^{\cC}_{[r]}$ or $G=G^{\cA}_{[r-1]}$ for a $C^i$- or $A^{i-1}$-block of physical qubits, respectively. By the results listed above, each such gadget must then output a Pauli $\cE(G,\; \eta_{\mathrm{in}},\; \gamma')$-deviation of the desired logical state at the timestep following that gadget, for some $\gamma'=O_r(\gamma)$.

  Thus the error correction gadget receives as input a set of codeblocks, each of which is a Pauli $\cE(G_{\mathrm{run}},\; \eta_{\mathrm{in}},\; O_r(\gamma_{\mathrm{in}}))$-deviation of the desired output state. As long as $\gamma_{\mathrm{in}}(r)>0$ and $\zeta_{\mathrm{run}}(r)>0$ are sufficiently small constants depending only on $r$, then applying the assumption that $\gamma_{\mathrm{run}}\leq\zeta_{\mathrm{run}}(r)\cdot\gamma_{\mathrm{in}}$ with Proposition~\ref{prop:errcorr} implies that the error correction gadget outputs a $\cE_{\mathrm{in}}=\cE(G_{\mathrm{run}},\; \eta_{\mathrm{in}},\; \gamma_{\mathrm{in}})$-deviation of the desired output codeblock.
\end{proof}

For intuition, it is helpful to consider the $\kappa=1$ case in Proposition~\ref{prop:permslab}. In this case, the logical qubits of $Q$ are labeled by a single $r$-dimensional hypercube $(L^{G,1})^r$, and Proposition~\ref{prop:permslab} provides a constant-overhead gadget for arbirarily permuting the slabs (i.e.~axis-parallel codimension-1 hyperplanes) in any given direction $h\in[r]$.

The corollary below provides an example application of Proposition~\ref{prop:permslab}. Specifically, we show that by applying a cyclic permutation in each direction $h\in[r]$, we can induce a global cyclic permutation on a $\Theta(k)$-sized set of logical qubits.

\begin{corollary}
  \label{cor:logcyc}
  Define $r,Q,n,k,i,\cC^*,G_{\mathrm{run}},\cE_{\mathrm{in}}=\cE_{\mathrm{out}},\cE_{\mathrm{run}},D$ as in Proposition~\ref{prop:permslab}. Choose some relatively prime positive integers $\ell_1,\dots,\ell_r\leq|L^{G,1}|$. Assume that we fix some ordering of the set $L^{G,1}\cong[|L^{G,1}|]$. For every $h\in[r]$, given some $s_h\in[\ell_h]$, let $\pi_h^{s_h}:L^{G,1}\rightarrow L^{G,1}$ be the permutation that cyclically rotates the first $\ell_h$ elements forwards by $s_h$ positions, that is,
  \begin{equation*}
    \pi_h^{s_h}(j) = \begin{cases}
      j+s_h\pmod{\ell_h},&j\in[\ell_h]\\
      j,&j\in[|L^{G,1}|]\setminus[\ell_h].
    \end{cases}
  \end{equation*}

  Then there exists a Pauli fault-tolerant gadget $((\cQ,(\cE_{\mathrm{run}}^{\sqcup u})^{\sqcup T}),D,D)$ for $\bar{\cO}$, where:
  \begin{itemize}
  \item $\cQ$ is an adaptive quantum circuit using quantum space $N=O_r(n)$ and time $T=O_r(1)$, and $u=O(1)$.
  \item $\bar{\cO}:\bC^{\cM^{\cC,i}\times\cM^{\cC,i}}\rightarrow\bC^{\cM^{\cC,i}\times\cM^{\cC,i}}$ is a channel that acts on $k$ qubits labeled by the set $M^{\cC,i}$, that permutes the qubits in $(L^{G,1})^r\subseteq M^{\cC,i}$ according to $\pi_1^{s_1}\otimes\cdots\otimes\pi_r^{s_r}$ (so that the qubit at position $(b_1,\dots,b_r)\in(L^{G,1})^r$ is moved to $(\pi_1^{s_1}(b_1),\dots,\pi_r^{s_r}(b_r))$); the qubits in $M^{\cC,i}\setminus(L^{G,1})^r$ may be moved or reset arbitrarily.
  \end{itemize}
\end{corollary}
\begin{proof}
  The desired circuit $\cQ$ simply applies Proposition~\ref{prop:permslab} $r$ times, with the permutation $\pi_h^{s_h}$ for every $h\in[r]$. The analysis follows directly from Proposition~\ref{prop:permslab}.
\end{proof}

The action of the logical permutation $\pi_1^{s_1}\otimes\cdots\otimes\pi_r^{s_r}$ on the logical qubits in $[\ell_1]\times\cdots\times[\ell_r]\subseteq(L^{G,1})^r$ is by definition equivalent to the action of addition by $(s_1,\dots,s_r)$ on the abelian group $\bZ_{\ell_1}\times\cdots\times\bZ_{\ell_r}$. Because we chose $\ell_1,\dots,\ell_r$ to be relatively prime, this group is isomorphic to a the cyclic group $\bZ_{\ell_1\cdots\ell_r}$. Thus as we can choose relatively prime $\ell_1,\dots,\ell_r=\Theta(|L^{G,1}|)=\Theta(|V_L|)$, we have obtained a constant-overhead fault-tolerant gadget for performing arbitrary cyclic shifts on a linear number of logical qubits (in a constant-rate code $Q$).

An alternative cyclic permutation gadget was given by \cite[Section~A.3]{xu_fast_2024}; however, this gadget relied on codes with a strong cyclic symmetry, for which it is unclear to what extend good asymptotic families can be found. In contrast, our gadget in Corollary~\ref{cor:logcyc} makes no symmetry assumptions on the code.

\subsection{Parallel Logical Hadamard}
\label{sec:hadsame}
The proposition below obtains a gadget to perform Hadamard gates in parallel on $\Theta(k)$ logical qubits, with the same input and output codes. For this purpose, we combine Lemma~\ref{lem:hadamard}, which performs transversal Hadamard but with different input and output codes, with our code switching gadgets. Our gadget below uses constant space-time overhead, which significantly improves upon the $\Omega(\sqrt{k})$ time overhead in the Hadamard gadget of \cite[Section~A.5]{xu_fast_2024}. In contrast to such prior works, our gadget also makes no symmetry assumptions on the underlying codes.

\begin{proposition}
  \label{prop:hadsame}
  Define $r,Q,n,k,i,\cC^*,G_{\mathrm{run}},\cE_{\mathrm{in}}=\cE_{\mathrm{out}},\cE_{\mathrm{run}},D$ as in Proposition~\ref{prop:permslab}. Then there exists a Pauli fault-tolerant gadget $((\cQ,(\cE_{\mathrm{run}}^{\sqcup u})^{\sqcup T}),D,D)$ for $\bar{\cO}$, where:
  \begin{itemize}
  \item $\cQ$ is an adaptive quantum circuit using quantum space $N=O_r(n)$ and time $T=O_r(1)$, and $u=O(1)$.
  \item $\bar{\cO}:\bC^{\cM^{\cC,i}\times\cM^{\cC,i}}\rightarrow\bC^{\cM^{\cC,i}\times\cM^{\cC,i}}$ is a channel that acts on $k$ qubits labeled by the set $M^{\cC,i}$, that performs a logical Hadamard gate on every qubit in $(L^{G,1})^r\subseteq M^{\cC,i}$; the qubits in $M^{\cC,i}\setminus(L^{G,1})^r$ may be moved or reset arbitrarily. 
  \end{itemize}
\end{proposition}
\begin{proof}
  The desired gadget $\cQ$ first applies transversal Hadamard to the input codeblock using Lemma~\ref{lem:hadamard}. By Lemma~\ref{lem:hadamard} and Remark~\ref{remark:dualenc}, the input logical state with $H^{\otimes k}$ applied is now encoded in the code associated to level $i$ of the \emph{chain} complex $\cC_*=\cC^{(1)}_*\otimes\cdots\otimes\cC^{(r)}_*$ with encoding map $\Enc^{\cC}_*$. By definition, this complex $\cC_*$ is given by our original cochain complex $\cC^*={\cC^{(1)}}^*\otimes\cdots\otimes{\cC^{(1)}}^*$ with every factor of $\cC^{G,*}$ replaced by $\cC^G_*$ and vice versa. To return to the original code $Q$, we therefore loop through the directions $h\in[r]$ one at a time, and for each $h$ we apply Proposition~\ref{prop:switchdown} to switch down to $|L^{G,1}|$ copies of a code associated to $\bigotimes_{h'\in[r]\setminus\{h\}}{\cC_0^{(h')}}^*$ (where ${\cC_0^{(h')}}^*$ denotes the current value of the factor in direction $h'$), and then we apply Proposition~\ref{prop:switchup} to switch up to a code associated to $\bigotimes_{h'\in[r]}{\cC_1^{(h')}}^*$, where ${\cC_1^{(h')}}^*={\cC^{(h')}}^*$ if $h'=h$ and ${\cC_1^{(h')}}^*={\cC_0^{(h')}}^*$ if $h'\neq h$. After all of these code switchings, we apply the error correction gadget in Proposition~\ref{prop:errcorr}.

  Similarly as in the proof of Proposition~\ref{prop:permslab}, the above procedure may fail if at some point (say at step $h\in[r]$) we end up with in the code at the bottom or top level of the cochain complex. However, this issue may be resolved exactly as in Proposition~\ref{prop:permslab}: we may temporarily switch one of the other factors $h'\neq h$ from ${\cC_0^{(h')}}^*$ to ${\cC_0^{(h')}}_*$, perform the desired switching on the $h$th factor, and then switch the $h'$th factor back to ${\cC_0^{(h')}}^*$.

  Analogously as in the proof of Proposition~\ref{prop:permslab}, the results listed above imply that $((\cQ,(\cE_{\mathrm{run}}^{\sqcup u})^{\sqcup T}),D,D)$ is a Pauli fault-tolerant gadget for $\bar{\cO}$, as desired.
\end{proof}

\subsection{Targeting of Individual Logical Qubits}
\label{sec:apptarget}

In this section, for the code $Q$ associated to some level $1\leq i\leq r-1$ of $\cC^*={\cC^{(1)}}^*\otimes\cdots{\cC^{(r)}}^*$, we present a gadget for swapping an arbitrary pair of logical qubits across a pair of $Q$-codeblocks, while leaving the remaining logical qubits unchanged. This gadget requires just constant space-time overhead, and is based on repeated applications of our logical permutation gadget in Proposition~\ref{prop:permslab}.

By using this gadget to extract a desired logical qubit (or pair of logical qubits) into an ancilla codeblock, and then performing transversal (e.g.~$CNOT$ or Hadamard) gates on the ancilla block, we can thereby perform constant-overhead targeted logical (e.g.~$CNOT$ or Hadamard) gates.

\begin{proposition}
  \label{prop:target}
  Define $r,Q,n,k,i,\cC^*,G_{\mathrm{run}},\cE_{\mathrm{in}}=\cE_{\mathrm{out}},\cE_{\mathrm{run}},D$ as in Proposition~\ref{prop:permslab}. Then for every $b,b'\in((L^{G,1})^r)^{\sqcup 2}$, there exists a Pauli fault-tolerant gadget $((\cQ,(\cE_{\mathrm{run}}^{\sqcup u})^{\sqcup T}),D,D)$ for $\bar{\cO}$, where:
  \begin{itemize}
  \item $\cQ$ is an adaptive quantum circuit using quantum space $N=O_r(n)$ and time $T=O_r(1)$, and $u=O_r(1)$.
  \item $\bar{\cO}:\bC^{(\cM^{\cC,i})^2\times(\cM^{\cC,i})^2}\rightarrow\bC^{(\cM^{\cC,i})^2\times(\cM^{\cC,i})^2}$ is a channel that acts on $2k$ qubits labeled by the set $(M^{\cC,i})^{\sqcup 2}$ (i.e.~two blocks of $M^{\cC,i}$), that swaps the two qubits $b,b'\in\times((L^{G,1})^r)^{\sqcup 2}$, and leaves all other qubits in $((L^{G,1})^r)^{\sqcup 2}$ unchanged; the qubits in $(M^{\cC,i})^{\sqcup 2}\setminus((L^{G,1})^r)^{\sqcup 2}$ may be moved or reset arbitrarily. 
  \end{itemize}
\end{proposition}
\begin{proof}
  We first describe a subroutine (i.e.~subgadget) that can swap the logical qubit $b$ out into some fixed position $\tilde{b}\in(L^{G,1})^r$ of a fresh ancilla codeblock; we can then run the same subroutine to extract $b'$ into another fresh ancilla codeblock, and then run the subroutine in reverse to insert each of these two logical qubits back into the other one's original location.

  The desired subroutine was illustrated in Figure~\ref{fig:target} in Section~\ref{sec:intro}. It simply prepares $r$ fresh ancilla $Q$-codeblocks using Proposition~\ref{prop:stateprep}, and then performs $2r-1$ applications of Proposition~\ref{prop:permslab} to swap slabs between these different codeblocks. Call the $Q$-codeblock containing $b=(b_1,\dots,b_h)\in(L^{G,1})^r$ codeblock $0$, and label the ancilla $Q$-codeblocks with $1,\dots,r$. The first $r$ swaps loop through $h=1,\dots,r$, and apply Proposition~\ref{prop:permslab} to swap the direction-$h$ slab at position $b_h$ in codeblock $h-1$ with the direction-$h$ slab at position $\tilde{b}_h$ in codeblock $h$. The final $r-1$ swaps then loop back through $h=r-1,\dots,1$, and apply Proposition~\ref{prop:permslab} to again swap the direction-$h$ slab at position $b_h$ in codeblock $h-1$ with the direction-$h$ slab at position $\tilde{b}_h$ in codeblock $h$.

  By definition, following the swaps above, position $b\in(L^{G,1})^r$ of codeblock $0$ will contain the logical qubit that was originally at position $\tilde{b}\in(L^{G,1})^r$ of codeblock $r$; all logical qubits in positions $(L^{G,1})^r\setminus\{b\}$ in codeblock $0$ remain unchanged. Meanwhile, position $\tilde{b}$ of codeblock $r$ will contain the logical qubit that was originally at position $b$ of codeblock $0$.

  Thus we have a procedure to extract a single logical qubit from a codeblock. We can apply this procedure to both logical qubits $b,b'$ (each with their own $r$ ancilla codeblocks), and then swap the resulting codeblocks containing the extracted qubits. Subsequently, running the same extraction procedure but with codeblock $r$ set to contain the extracted logical qubit (and fresh ancilla code states only for codeblocks $1,\dots,r-1$) has the effect of inserting the extracted qubits back into the original codeblocks. Hence we have given a procedure to insert logical qubit $b$ into the position previously held by $b'$, and vice versa. At the end of this entire procedure, we run the error correction gadget in Proposition~\ref{prop:errcorr}.

  The algorithm described above yields the desired gadget $\cQ$. This gadget by definition runs $O_r(1)$ applications of the gadgets in Proposition~\ref{prop:stateprep}, Proposition~\ref{prop:permslab}, and Proposition~\ref{prop:errcorr}. Each such gadget uses quantum space $O_r(n)$ and time $O_r(1)$, so our entire quantum space and time usage is $O_r(n)$ and $O_r(1)$, respectively. These propositions also directly imply that $((\cQ,(\cE_{\mathrm{run}}^{\sqcup u})^{\sqcup T}),D,D)$ is a Pauli fault-tolerant gadget for $\bar{\cO}$, by analogous reasoning as used in the proof of Proposition~\ref{prop:permslab}.
\end{proof}

As an immediate corollary to Proposition~\ref{prop:permslab}, we obtain gadgets that perform a targeted $CNOT$ or Hadamard gate using constant space-time overhead.

\begin{corollary}
  \label{cor:target}
  Define $r,Q,n,k,i,\cC^*,G_{\mathrm{run}},\cE_{\mathrm{in}}=\cE_{\mathrm{out}},\cE_{\mathrm{run}},D$ as in Proposition~\ref{prop:permslab}. Then the following hold:
  \begin{enumerate}
  \item (Targeted inter-block $CNOT$) Given arbitrary $b,b'\in((L^{G,1})^r)^{\sqcup 2}$, there exists a Pauli fault-tolerant gadget $((\cQ,(\cE_{\mathrm{run}}^{\sqcup u})^{\sqcup T}),D^{\sqcup 2},D^{\sqcup 2})$ for $\bar{\cO}$, where $\bar{\cO}:\bC^{(\cM^{\cC,i})^2\times(\cM^{\cC,i})^2}\rightarrow\bC^{(\cM^{\cC,i})^2\times(\cM^{\cC,i})^2}$ is a channel that acts on $2k$ qubits labeled by the set $(M^{\cC,i})^{\sqcup 2}$ (i.e.~two blocks of $M^{\cC,i}$), that performs a logical $CNOT$ gate on the two qubits $b,b'\in((L^{G,1})^r)^{\sqcup 2}$, and leaves all other qubits inside $((L^{G,1})^r)^{\sqcup 2}$ unchanged; the qubits in $(M^{\cC,i})^{\sqcup 2}\setminus((L^{G,1})^r)^{\sqcup 2}$ may be moved or reset arbitrarily. 
  \item (Targeted intra-block $CNOT$) Given arbitrary $b,b'\in(L^{G,1})^r$, there exists a Pauli fault-tolerant gadget $((\cQ,(\cE_{\mathrm{run}}^{\sqcup u})^{\sqcup T}),D,D)$ for $\bar{\cO}$, where $\bar{\cO}:\bC^{\cM^{\cC,i}\times\cM^{\cC,i}}\rightarrow\bC^{\cM^{\cC,i}\times\cM^{\cC,i}}$ is a channel that acts on $k$ qubits labeled by the set $M^{\cC,i}$, that performs a logical $CNOT$ gate on the two qubits $b,b'\in(L^{G,1})^r$, and leaves all other qubits inside $(L^{G,1})^r$ unchanged; the qubits in $(M^{\cC,i})^{\sqcup 2}\setminus((L^{G,1})^r)^{\sqcup 2}$ may be moved or reset arbitrarily. 
  \item\label{it:cthad} (Targeted Hadamard) Given arbitrary $b\in(L^{G,1})^r$, there exists a Pauli fault-tolerant gadget $((\cQ,(\cE_{\mathrm{run}}^{\sqcup u})^{\sqcup T}),D,D)$ for $\bar{\cO}$, where $\bar{\cO}:\bC^{\cM^{\cC,i}\times\cM^{\cC,i}}\rightarrow\bC^{\cM^{\cC,i}\times\cM^{\cC,i}}$ is a channel that acts on $k$ qubits labeled by the set $M^{\cC,i}$, that performs a logical Hadamard gate on the qubit $b\in(L^{G,1})^r$, and leaves all other qubits inside $(L^{G,1})^r$ unchanged; the qubits in $(M^{\cC,i})^{\sqcup 2}\setminus((L^{G,1})^r)^{\sqcup 2}$ may be moved or reset arbitrarily. 
  \end{enumerate}
  In each case above, $\cQ$ is an adaptive quantum circuit using quantum space $N=O_r(n)$ and time $T=O_r(1)$, and $u=O(1)$.
\end{corollary}
\begin{proof}
  We consider the three statements separately:
  \begin{enumerate}
  \item\label{it:targetCNOTdiff} The desired gadget $\cQ$ performs the following:
    \begin{itemize}
    \item Prepare two ancilla $Q$-codeblocks using Proposition~\ref{prop:stateprep}.
    \item Swap the two logical qubits $b,b'$ from the input codeblocks into some fixed position $\tilde{b}\in(L^{G,1})^r$ of these two respective ancilla codeblocks using Proposition~\ref{prop:target}.
    \item Apply transversal $CNOT$ between the two ancilla codeblocks using Lemma~\ref{lem:CNOTsame}.
    \item Swap the two logical qubits back into their original positions in the input codeblocks using Proposition~\ref{prop:target}.
    \item Run the error correction gadget in Proposition~\ref{prop:errcorr}.
    \end{itemize}
    The fault-tolerance analysis follows directly from the results referenced above.
  \item The proof is analogous to that for the targeted inter-block $CNOT$ in item~\ref{it:targetCNOTdiff} above, except that in the intra-block case, all four applications of Proposition~\ref{prop:target} are applied to the unique input codeblock (along with the appropriate ancilla codeblock).
  \item The desired gadget $\cQ$ performs the following:
    \begin{itemize}
    \item Prepare an ancilla $Q$-codeblock using Proposition~\ref{prop:stateprep}.
    \item Swap the logical qubit $b$ from the input codeblock into some position $b\in(L^{G,1})^r$ of this ancilla codeblock using Proposition~\ref{prop:target}.
    \item Apply logical $H^{\otimes(L^{G,1})^r}$ to this ancilla codeblock using Proposition~\ref{prop:hadsame}.
    \item Swap the logical qubit $b$ back into its original position in the input codeblock using Proposition~\ref{prop:target}.
    \item Run the error correction gadget in Proposition~\ref{prop:errcorr}.
    \end{itemize}
    The fault-tolerance analysis follows directly from the results referenced above.
  \end{enumerate}
\end{proof}

Our gadget for targeted logical swap in Proposition~\ref{prop:target} is related to the selective inter-block teleportation gadget in \cite[Section~A.2]{xu_fast_2024}. Specifically, \cite{xu_fast_2024} show how to swap the logical qubits at a given position across two codeblocks using a constant number of ``logical cycles'' (which corresponds to a constant-depth circuit in our setting where we have constant-depth state preparation). They also extend this result to swap more such pairs of qubits, though with a higher time overhead.

Using a single logical swap, we are able to obtain the constant-overhead targeted $CNOT$ gate in Corollary~\ref{cor:target} because Proposition~\ref{prop:target} allows us to swap two logical qubits at \emph{arbitrary} (possibly different) positions within the codeblock; in contrast, the techniques in \cite[Section~A.2]{xu_fast_2024} most naturally apply to swapping logical qubits at the same position in different codeblocks, and hence seem to require additional properties such as cyclic code symmetries to obtain targeted $CNOT$ gates.

For demonstrative purposes, we have stated Proposition~\ref{prop:target} and Corollary~\ref{cor:target} for the case of performing a single logical swap, $CNOT$, or Hadamard gate. However, our gadgets described in Section~\ref{sec:appperm} naturally perform highly parallel and flexible logical permutations, which can be used to efficiently perform many such gates in parallel. Such techniques will are particularly well suited for performing the same logical gate on many logical qubits within a given slab.

\subsection{Constant-Overhead State Preparation of 2-Dimensional Codes in Bulk}
\label{sec:app2dprep}
In this section, we simply remark that combining our state preparation and downwards code switching gadgets yield a gadget for preparing a stack of 2-dimensional product codes with constant-space-time overhead.

Specifically, set $r=3$, and consider applying Proposition~\ref{prop:stateprep} to construct either a $\ket{+^k}$ logical state (if $i=1$) or a $\ket{0^k}$ logical state (if $i=2$) in the $[[n,k]]$ code $Q$ associated to level $i$ of $\cC^*={\cC^{(1)}}^*\otimes{\cC^{(2)}}^*\otimes{\cC^{(3)}}^*$ (as defined in Section~\ref{sec:applycode}). Then applying Proposition~\ref{prop:switchdown} to this state results in a stack of $|L^{G,1}|=\Theta(\sqrt{n'})$ copies of the $[[n',k']]$ code $Q'$ associated to level $1$ of $\cC^{G,*}\otimes\cC^G_*$, in logical state either $\ket{+^{k'}}$ or $\ket{0^{k'}}$. The propositions listed above imply the fault-tolerance of this procedure. Indeed, we implicitly use this procedure in our proof of Proposition~\ref{prop:switchup} in Appendix~\ref{app:suproof}.

The code $Q'$ is a 2-dimensional product code (often called a hypergraph product code \cite{tillich_quantum_2014}). It is uncommon to have state preparation gadgets for such 2-dimensional codes with constant space-time overhead. Indeed, we are unaware of any such gadgets for a single such code state. It is therefore perhaps interesting that we are able to circumvent this challenge by simultaneously preparing $\Theta(k')$ states of $Q'$ with constant overhead.

\section{Acknowledgments}
We thank Ting-Chun Lin for helpful discussions.

\bibliographystyle{alpha}
\bibliography{library1,library2}

\appendix

\section{Fault-Tolerance Proof for State Preparation}
\label{app:spproof}
In this section, we provide the proof of Proposition~\ref{prop:stateprep}. The general structure if this proof is similar to that of Proposition~\ref{prop:switchdown} given in Section~\ref{sec:csanalysis}.

\begin{proof}[Proof of Proposition~\ref{prop:stateprep}]
  The desired circuit $\cQ$ is given in Algorithm~\ref{alg:stateprep}. By definition line~\ref{li:spinit} uses $1$ timestep, line~\ref{li:spse} uses $w+4$ timesteps (see Corollary~\ref{cor:syndext}), line~\ref{li:spdec} and line~\ref{li:spge} together use $0$ timesteps (as they simply perform a classical computation), and line~\ref{li:spcorr} uses $1$ timestep. Thus the entire circuit uses $T=w+6=O(w)$ timesteps. Furthermore, line~\ref{li:spse} uses $|C^i|+|C^{i+1}|\leq n+nw=n(1+w)=O(nw)$ physical qubits, and every other line uses a subset of these qubits, so the entire space usage is $N=O(nw)$.

  Also note that the vertex set $V_{\mathrm{run}}\supseteq V(G^{\cC}_{i,i+1,i+2})=C^i\sqcup C^{i+1}\sqcup C^{i+2}$ contains the set $C^i\sqcup C^{i+1}$ of physical qubits, so $\cE_{\mathrm{run}}$ is a well-defined family of bad sets for $\cQ$.
  
  Our goal is to show that for every $\cE_{\mathrm{run}}$-avoiding Pauli fault $\cF$, the resulting output $\cQ[\cF](1)$ is a Pauli $\cE_{\mathrm{out}}$-deviation of $\Enc(\ket{+^k}\bra{+^k})$. The reference system in Definition~\ref{def:gadget} has no effect on the proof, so we ignore it here for simplicity (see Remark~\ref{remark:refsys}).

  For this purpose, we track all errors that arise at any point during the execution of the circuit $\cQ$ using the graph $G_{\mathrm{run}}$. By definition Algorithm~\ref{alg:stateprep} consists entirely of Clifford gates, and we have restricted attention to a Pauli fault $\cF$. Therefore all of the unitary (i.e.~non-measurement) Clifford gates simply propagate Pauli errors to (possibly) different Pauli errors.



  Let $S_{\mathrm{\cF}}\subseteq C^i\sqcup C^{i+1}$ denote the set of all physical qubits that lie in (any timestep of) the forward-looking lightcone in $\cQ$ (given by Algorithm~\ref{alg:stateprep}) of any qubit in $\bigcup_{t\in[T]}\supp(\cF_t)$, where we take the lightcones starting from the start of the circuit (i.e.~timestep $1$). Note that this foward-looking lightcone only counts quantum gates in line~\ref{li:spinit}, line~\ref{li:spse}, and line~\ref{li:spcorr}, and does not count classical gates from line~\ref{li:spdec} or line~\ref{li:spge}.
  Let $S_{\mathrm{syn}}\subseteq C^{i+1}\sqcup C^{i+2}$ denote the footprint (see Definition~\ref{def:footprint}) of the call to \FnSSFlipSyn{$\delta(s);i+1,\cC^*$} in line~\ref{li:spdec}.

  Furthermore, because errors remain Pauli when propagating through our Clifford circuit as described above, just prior to performing the Pauli-$Z$ measurements in the syndrome extraction subroutine in line~\ref{li:spse}, our corrupted state equals
  \begin{equation*}
    E_0\left(\sum_{y\in\cC^i}\ket{y}\otimes\ket{\delta(y)}\right)\left(\sum_{y\in\cC^i}\bra{y}\otimes\bra{\delta(y)}\right)E_0'
  \end{equation*}
  which is the true state from the execution described in Lemma~\ref{lem:nesp} with some Pauli error $E_0,E_0'$ applied. This error $E_0,E_0'$ is precisely determined by propagating the Pauli errors from $\cF$ through the Clifford gates in the circuit, and $\supp(E_0),\supp(E_0')\subseteq S_{\cF}$. Write $E_0=Z^{(e_Z,f_Z)}X^{(e_X,f_X)}$ and $E_0'=X^{(e_X',f_X')}Z^{(e_Z',f_Z')}$ where $e_P,e_P'\in\cC^i$ and $f_P,f_P'\in\cC^{i+1}$ for $P\in\{X,Z\}$. Then we can rewrite the state above as
  \begin{equation*}
    \sum_{s,s'\in\cC^{i+1}}\sum_{y,y'}Z^{e_Z}X^{e_X}\ket{y}\bra{y'}X^{e_X'}Z^{e_Z'} \otimes Z^{f_Z}\ket{s}\bra{s'}Z^{f_Z'},
  \end{equation*}
  where the second sum is over all $y,y'\in\cC^i$ for which $\delta(y)+f_X=s$ and $\delta(y')+f_X'=s'$.
  Then after performing the Pauli-$Z$ measurements on the ancilla register in the call to \FnSyndExt{$\rho;Z,\delta_i$}, the state collapses to only those terms in the sum with $s=s'$, which gives
  \begin{equation}
    \label{eq:postmeasure}
    \sum_{s\in\cC^{i+1}}\sum_{y,y'\in\delta^{-1}(s)}Z^{e_Z}X^{e_X}\ket{y+g(f_X)}\bra{y'+g'(f_X')}X^{e_X'}Z^{e_Z'} \otimes Z^{f_Z}\ket{s}\bra{s}Z^{f_Z'},
  \end{equation}
  where $g=g(f_X),\;g'=g'(f_X')\in\cC^i$ are any fixed elements satisfying $\delta(g)=f_X,\;\delta(g')=f_X'$. Algorithm~\ref{alg:stateprep} then uses the syndrome $s$ to compute an appropriate $x\in\cC^i$, and returns the state above with the correction $X^x$ applied (with possibly some additional Pauli errors supported inside $S_{\cF}$ from the fault during this final Pauli correction).

  To complete the proof of the proposition, we show Lemma~\ref{lem:oneterm} below, which analyzes the output of Algorithm~\ref{alg:stateprep} resulting from the term associated to each possible syndrome in~(\ref{eq:postmeasure}). Specifically, because $T\leq w+6$ as shown below, Lemma~\ref{lem:oneterm} implies that Algorithm~\ref{alg:stateprep} outputs a Pauli $\cE_{\mathrm{out}}$-deviation of $\Enc(\ket{+^k}\bra{+^k})$, as desired.

  \begin{lemma}
    \label{lem:oneterm}
    There exist Pauli errors $E,E'$ whose supports are $\cE(G_{\mathrm{run}},\; \eta_{\mathrm{run}},\; 100w^7T2^T\gamma_{\mathrm{run}})$-avoiding such that when lines~\ref{li:spdec}--\ref{li:spret} of Algorithm~\ref{alg:stateprep} are applied to each term
    \begin{equation}
      \label{eq:postmeasureterm}
      \sum_{y,y'\in\delta^{-1}(s)} Z^{e_Z}X^{e_X}\ket{y+g}\bra{y'+g'}X^{e_X'}Z^{e_Z'} \otimes \ket{s}\bra{s}
    \end{equation}
    in the sum in~(\ref{eq:postmeasure}),\footnote{Note that $Z^{f_Z},Z^{f_Z'}$ applied to $\ket{s}\bra{s}$ simply induce a global phase on each such term.} the output is of the form $E\Enc(\ket{+^k}\bra{+^k})E'$.
  \end{lemma}
  \begin{proof}
    Fix some $\bar{y}\in\cC^i$ with $\delta(\bar{y})=s$, and let $b=\bar{y}+g$, $f=f_X$, $b'=\bar{y}+g'$, $f'=f_X'$. Then by definition, the measurement syndrome is
    \begin{equation*}
      s = \delta(b)+f = \delta(b')+f'
    \end{equation*}
    where $\supp(f),\supp(f')\subseteq S_{\cF}\cap C^{i+1}$.
    Letting $a^{i+1}\in\cC^{i+1}$ be the variable defined in line~\ref{li:spdec} of Algorithm~\ref{alg:stateprep}, we then define $S_{\mathrm{err}}$ and $\tilde{a}^i,\tilde{a}^{i+1}$ (resp.~$S_{\mathrm{err}}'$ and $(\tilde{a}^i)',(\tilde{a}^{i+1})'$) to be the footprint and output of \FnSSFlipErr{$f+a^{i+1};i+1,\cC^*$} (resp.~\FnSSFlipErr{$f'+a^{i+1};i+1,\cC^*$}), respectively. Note that by definition $f+a^{i+1}=s+a^{i+1}+\delta(b)$ (resp.~$f'+a^{i+1}=s+a^{i+1}+\delta(b')$). We state and prove the following claims for the unprimed variables $b,f,\dots$, but they analogously apply to the primed variables $b',f',\dots$.

    \begin{claim}
      \label{claim:SFavoid}
      The set $S_{\cF}$ is $\cE(G_{\mathrm{run}},\; \eta_{\mathrm{run}},\; T2^T\gamma_{\mathrm{run}})$-avoiding.
    \end{claim}
    \begin{proof}
      Assume for a contradiction that there exists some $V\subseteq V_{\mathrm{run}}$ of size $|V|\geq\eta_{\mathrm{run}}$ that induces a connected subgraph of $G_{\mathrm{run}}$ with $|V\cap S_{\cF}|/|V|\geq T2^T\gamma_{\mathrm{run}}$. We can repeatedly add to $V$ any point in $S_{\cF}\setminus V$ that lies in the neighborhood of $V$ until there are no more such points. The resulting set $V$ can only have larger size $|V|\geq\eta_{\mathrm{run}}$ and $S_{\cF}$-density $|V\cap S_{\cF}|/|V|\geq 2^T\gamma_{\mathrm{run}}$, and contains all or none of every connected component of the subgraph of $G_{\mathrm{run}}$ induced by $S_{\cF}$.

      By definition, the circuit $\cQ$ has $T$ timesteps and consists of 1 and 2-qubit gates, where every 2-qubit gate acts on a pair of qubits connected by an edge in $G_{\mathrm{run}}$. Therefore every point in $\bigcup_{t\in[T]}\supp(\cF_t)$ has a forward-lightcone hitting $<2^T$ qubits in $C^i\sqcup C^{i+1}$, which induce a conected subgraph of $G_{\mathrm{run}}$. Thus
      \begin{align*}
        |V\cap S_{\cF}|
        &< \left|V\cap\bigcup_{t\in[T]}\supp(\cF_t)\right|\cdot 2^T \leq T2^T\gamma_{\mathrm{run}}|V|,
      \end{align*}
      which contradicts the assumption that $|V\cap S_{\cF}|/|V|\geq T2^T\gamma_{\mathrm{run}}$, as desired. The second inequality above holds by the assumption that $\cF$ is $\cE_{\mathrm{run}}^{\sqcup T}$-avoiding, so that each $\cF_t$ is $\cE_{\mathrm{run}}$-avoiding.
    \end{proof}

    \begin{claim}
      \label{claim:synavoid}
      The set $S_{\cF}\cup S_{\mathrm{syn}}$ is $\cE(G_{\mathrm{run}},\; \eta_{\mathrm{run}},\; 5w^3T2^T\gamma_{\mathrm{run}})$-avoiding.
    \end{claim}
    \begin{proof}
      Assume for a contradiction that there exists some $V\subseteq V_{\mathrm{run}}$ of size $|V|\geq\eta_{\mathrm{run}}$ that induces a connected subgraph of $G_{\mathrm{run}}$ with $|V\cap(S_{\cF}\cup S_{\mathrm{syn}})|\geq 5w^3T2^T\gamma_{\mathrm{run}}$. As in the proof of Claim~\ref{claim:SFavoid}, we may repeatedly add points in $(S_{\cF}\cup S_{\mathrm{syn}})\setminus V$ that lie in the neighborhood of $V$, in order to assume that $V$ contains all or none of every connected component of the subgraph of $G_{\mathrm{run}}$ induced by $S_{\cF}\cup S_{\mathrm{syn}}$.

      Recall as described above that the call \FnSSFlipSyn{$\delta(s);i+1,\cC^*$} in line~\ref{li:spdec} of Algorithm~\ref{alg:stateprep} has $s=\delta(b)+f$ with $\supp(f)\subseteq S_{\cF}\cap C^{i+1}$. Then by Lemma~\ref{lem:footprint}, we have
      \begin{align*}
        |V\cap S_{\mathrm{syn}}|
        &\leq |V\cap\supp(f)|/\gamma_{\mathrm{syn}} \\
        &\leq |V\cap S_{\cF}| \cdot 4w^3 \\
        &< T2^T\gamma_{\mathrm{run}}|V| \cdot 4w^3,
      \end{align*}
      where the third inequality above holds by Claim~\ref{claim:SFavoid}. We then again apply Claim~\ref{claim:SFavoid} to conclude that
      \begin{equation*}
        |V\cap (S_{\cF}\cup S_{\mathrm{syn}})| \leq |V\cap S_{\cF}|+|V\cap S_{\mathrm{syn}}| < (4w^3+1)T2^T\gamma_{\mathrm{run}}|V|,
      \end{equation*}
      which contradicts the assumption that $|V\cap(S_{\cF}\cup S_{\mathrm{syn}})|\geq 5w^3T2^T\gamma_{\mathrm{run}}$, as desired.
    \end{proof}

    By definition, each update $c^{i+1}$ that gets added into $a^{i+1}$ during the a given step in the execution of \FnSSFlipSyn{$\delta(s);i+1,\cC^*$} has $\supp(c^{i+1})$ contained within the neighborhood in $G^{\cC}_{i,i+1,i+2}\subseteq G_{\mathrm{run}}$ of the footprint from the previous step. As a consequence, if we restrict $s$ to a given connected component $V$ of the subgraph of $G_{\mathrm{run}}$ induced by $S_{\cF}\cup S_{\mathrm{syn}}$ (and zero out all values outside of $V$), the output of \FnSSFlipSyn{$\delta(s|_{V\cap C^{i+1}});i+1,\cC^*$} must equal the restriction $a^{i+1}|_{V\cap C^{i+1}}$ of the output $a^{i+1}$ of \FnSSFlipSyn{$\delta(s);i+1,\cC^*$} (see Lemma~\ref{lem:sslocal}).

    Recall that by assumption $\eta_{\mathrm{run}}\leq m$. Also by definition, $\gamma_{\mathrm{run}}\leq 1/5w^3T2^T$, so it follows by Claim~\ref{claim:synavoid} that the subgraph of $G_{\mathrm{run}}$ induced by $S_{\cF}\cup S_{\mathrm{syn}}$ has no connected components containing $\geq\eta_{\mathrm{run}}$ vertices. Therefore within every connected component $V$ of this induced subgraph, by the assumption that $\cC^*$ has a $(m,0)$-small-set syndrome-flip decoder at level $i+1$, the while loop in the restricted execution \FnSSFlipSyn{$\delta(s|_{V\cap C^{i+1}});i+1,\cC^*$} will not terminate until the computed output $a^{i+1}|_{V\cap C^{i+1}}$ satisfies $\delta(s+a^{i+1}|_{V\cap C^{i+1}})=0$. Furthermore, by definition $s=\delta(b)+f$ with $\supp(f)\subseteq S_{\cF}$ and $\supp(a^{i+1})\subseteq S_{\mathrm{syn}}$, and hence $s+a^{i+1}+\delta(b)=f+a^{i+1}$ has $\supp(f+a^{i+1})\subseteq S_{\cF}\cup S_{\mathrm{syn}}$.

    \begin{claim}
      \label{claim:erravoid}
      The set $S_{\cF}\cup S_{\mathrm{syn}}\cup S_{\mathrm{err}}$ is $\cE(G_{\mathrm{run}},\; \eta_{\mathrm{run}},\; 50w^7T2^T\gamma_{\mathrm{run}})$-avoiding.
    \end{claim}
    \begin{proof}
      The proof is similar to that of Claim~\ref{claim:synavoid}. Assume for a contradiction that there exists some $V\subseteq V_{\mathrm{run}}$ of size $|V|\geq\eta_{\mathrm{run}}$ that induces a connected subgraph of $G_{\mathrm{run}}$ with $|V\cap(S_{\cF}\cup S_{\mathrm{syn}}\cup S_{\mathrm{err}})|\geq 50w^7T2^T\gamma_{\mathrm{run}}$. As in the proof of Claim~\ref{claim:synavoid}, we may repeatedly add points in $(S_{\cF}\cup S_{\mathrm{syn}}\cup S_{\mathrm{err}})\setminus V$ that lie in the neighborhood of $V$, in order to assume that $V$ contains all or none of every connected component of the subgraph of $G_{\mathrm{run}}$ induced by $S_{\cF}\cup S_{\mathrm{syn}}\cup S_{\mathrm{err}}$.

      By Lemma~\ref{lem:footprint}, we have
      \begin{align*}
        |V\cap S_{\mathrm{err}}|
        &\leq |V\cap\supp(f+a^{i+1})|/\gamma_{\mathrm{err}} \\
        &\leq |V\cap(S_{\cF}\cup S_{\mathrm{syn}})| \cdot 8w^4 \\
        &< 5w^3T2^T\gamma_{\mathrm{run}}|V| \cdot 8w^4,
      \end{align*}
      where the third inequality above holds by Claim~\ref{claim:synavoid}. We then again apply Claim~\ref{claim:synavoid} to conclude that
      \begin{equation*}
        |V\cap(S_{\cF}\cup S_{\mathrm{syn}}\cup S_{\mathrm{err}})| \leq |V\cap(S_{\cF}\cup S_{\mathrm{syn}})| + |V\cap S_{\mathrm{err}}| < (8w^4+1)5w^3T2^T\gamma_{\mathrm{run}}|V|,
      \end{equation*}
      which contradicts the assumption that $|V\cap(S_{\cF}\cup S_{\mathrm{syn}}\cup S_{\mathrm{err}})|\geq 50w^7T2^T\gamma_{\mathrm{run}}$, as desired.
    \end{proof}

    Recall that we use $\tilde{a}^i,\tilde{a}^{i+1}$ to denote the variables that are updated throughout the execution of and then returned by \FnSSFlipErr{$f+a^{i+1};i+1,\cC^*$} (corresponding to $a^{i-1},a^i$ respectively in Algorithm~\ref{alg:ssflip}). As shown above, we have
    \begin{equation*}
      \delta(f+a^{i+1}) = \delta(s+a^{i+1}+\delta(b)) = \delta(s+a^{i+1}) = 0,
    \end{equation*}
    and hence the entire execution of \FnSSFlipErr{$f+a^{i+1};i+1,\cC^*$} will have $\tilde{a}^{i+1}=0$; that is, every iteration of the while loop that does not fail must add some $c^i$ in to $\tilde{a}^i$. By definition, every such $c^i$ must have $\supp(c^i)$ contained within the neighborhood in $G^{\cC}_{i,i+1,i+2}\subseteq G_{\mathrm{run}}$ of the footprint from the previous step. As a consequence, if we restrict $f+a^{i+1}$ to a connected component $V$ if $S_{\cF}\cup S_{\mathrm{syn}}\cup S_{\mathrm{err}}$ (and zero out all values outside of $V$), the output of \FnSSFlipErr{$f+a^{i+1}|_{V\cap C^{i+1}};i+1,\cC^*$} must equal the restriction $\tilde{a}^i|_{V\cap C^i},\tilde{a}^{i+1}|_{V\cap C^{i+1}}$ of the output $\tilde{a}^i,\tilde{a}^{i+1}$ of \FnSSFlipErr{$f+a^{i+1};i+1,\cC^*$} (see Lemma~\ref{lem:sslocal}).

    Recall that by assumption $\eta_{\mathrm{run}}\leq m$. Also by definition, $\gamma_{\mathrm{run}}\leq 1/50w^7T2^T$, so it follows by Claim~\ref{claim:erravoid} that the subgraph of $G_{\mathrm{run}}$ induced by $S_{\cF}\cup S_{\mathrm{syn}}\cup S_{\mathrm{err}}$ has no connected components containing $\geq\eta_{\mathrm{run}}$ vertices. Therefore within every connected component $V$ of this induced subgraph, by the assumption that $\cC^*$ has a $m$-small-set error-flip decoder at level $i+1$, the while loop in the restricted execution \FnSSFlipErr{$f+a^{i+1}|_{V\cap C^{i+1}};i+1,\cC^*$} will not terminate until the computed output $\tilde{a}^i|_{V\cap C^i}$ satisfies $\delta(\tilde{a}^i|_{V\cap C^i})=f+a^{i+1}|_{V\cap C^{i+1}}$; recall from above that we will have $\tilde{a}^{i+1}|_{V\cap C^{i+1}}=0$. It also follows that the output will always be such a valid $\tilde{a}^i|_{V\cap C^i}$, and never FAIL. Therefore we conclude that
    \begin{equation*}
      \delta(\tilde{a}^i) = f+a^{i+1},
    \end{equation*}
    and by definition $\supp(\tilde{a}^i)\subseteq S_{\mathrm{err}}$.

    It follows that $\delta(b)=s+f=s+a^{i+1}+\delta(\tilde{a}^i)$, so $b+\tilde{a}^i$ is a valid choice for $x\in\cC^i$ in line~\ref{li:spge} of Algorithm~\ref{alg:stateprep}.
    Recall that the entire analysis above also applies to the primed variables $b',f',(\tilde{a}^i)'$, so that we also have $\delta(b')=s+f'=s+a^{i+1}+\delta((\tilde{a}^i)')$, and $b'+(\tilde{a}^i)'$ is also a valid choice for $x\in\cC^i$.

    Therefore assuming Algorithm~\ref{alg:stateprep} is in the state~(\ref{eq:postmeasureterm}) in prior to line~\ref{li:spdec}, it will compute $x=b+\tilde{a}^i+z=b'+(\tilde{a}^i)'+z'$ for some $i$-cocycles $z,z'\in Z^i(\cC)$, and then output the first register (i.e.~qubits $C^i$) of
    \begin{align*}
      \hspace{1em}&\hspace{-1em}\tilde{E}\left(\sum_{y,y'\in\delta^{-1}(s)} \ket{y+g+x}\bra{y'+g'+x} \otimes \ket{s}\bra{s}\right)\tilde{E'} \\
                  &= \tilde{E}\left(\sum_{y,y'\in\delta^{-1}(s)} \ket{(y+\bar{y}+z)+\tilde{a}^i}\bra{(y'+\bar{y}+z')+(\tilde{a}^i)'} \otimes \ket{s}\bra{s}\right)\tilde{E'} \\
                  &= \tilde{E}\left(\sum_{y,y'\in Z^i(\cC)}X^{\tilde{a}^i}\ket{y}\bra{y'}X^{(\tilde{a}^i)'}\otimes\ket{s}\bra{s}\right)\tilde{E}' \\
                  &= \tilde{E}\left(X^{\tilde{a}^i}\Enc\left(\ket{+^k}\bra{+^k}\right)X^{(\tilde{a}^i)'}\otimes\ket{s}\bra{s}\right)\tilde{E}',
    \end{align*}
    where $\tilde{E},\tilde{E}'$ are simply the Pauli errors arising from propagating the errors from $\cF$ through the unitary Clifford gates. Hence these errors do not depend on the measurement outcome $s$, and satisfy $\supp(\tilde{E}),\supp(\tilde{E}')\subseteq S_{\cF}$.
    Also, by definition $\tilde{a}^i,(\tilde{a}^i)'$ only depend on $f,f'$ and $a^{i+1}$, the latter of which only depends on $\delta(s)=\delta(f)=\delta(f')$. By definition $\supp(\tilde{a}^i)\subseteq S_{\mathrm{err}}$ and $\supp((\tilde{a}^i)')\subseteq S_{\mathrm{err}}'$, so we conclude that the output above differs from $\Enc(\ket{+^k}\bra{+^k})$ by a Pauli error that only depends on $\cF$, and is supported inside
    \begin{equation*}
      S_{\cF}\cup S_{\mathrm{syn}}\cup S_{\mathrm{err}}\cup S_{\mathrm{err}}' = (S_{\cF}\cup S_{\mathrm{syn}}\cup S_{\mathrm{err}}) \cup (S_{\cF}\cup S_{\mathrm{syn}}\cup S_{\mathrm{err}}'),
    \end{equation*}
    which by Claim~\ref{claim:erravoid} is $\cE(G_{\mathrm{run}},\; \eta_{\mathrm{run}},\; 100w^7T2^T\gamma_{\mathrm{run}})$-avoiding, as desired.
  \end{proof}

\end{proof}

\section{Fault-Tolerance Proof for Error Correction}
\label{app:ecproof}
In this section, we present the proof of Proposition~\ref{prop:errcorr}, which is somewhat similar in structure to the proofs of Propostion~\ref{prop:switchdown} and Proposition~\ref{prop:stateprep}.

\begin{proof}[Proof of Proposition~\ref{prop:errcorr}]
  The desired gadget $\cQ$ is given in Algorithm~\ref{alg:errcorr}. By definition line~\ref{li:ecseZ} and line~\ref{li:ecseX} each use $w+4$ timesteps (see Corollary~\ref{cor:syndext}), line~\ref{li:ecdecZ} and line~\ref{li:ecdecX} each use $0$ timesteps (as they perform classical computations), and line~\ref{li:eccorrZ} and line~\ref{li:eccorrX} each use $1$ timestep, for a total of $T=2(w+5)=2w+10=O(w)$ timesteps. Furthermore, in addition to the $n=|C^i|$ physical qubits, line~\ref{li:ecseZ} uses $|C^{i+1}|$ ancilla qubits and line~\ref{li:ecseX} uses $|C^{i-1}$ ancilla qubits, for a total space usage of $N=|C^{i-1}|+|C^i|+|C^{i+1}|\leq nw+n+nw\leq 3nw=O(nw)$ physical qubits.

  Also note that the vertex set $V_{\mathrm{run}}\supseteq V(G^{\cC}_{i,i+1,i+2})=C^i\sqcup C^{i+1}\sqcup C^{i+2}$ contains the set $C^i\sqcup C^{i+1}$ of physical qubits, so $\cE_{\mathrm{run}}$ is a well-defined family of bad sets for $\cQ$.

  Our goal is to show that for every $\ell\in\bN$, every $\rho\in\bC^{2^{k+\ell}\times 2^{k+\ell}}$, every $\cE_{\mathrm{in}}$-deviation $\sigma=E_0\Enc(\rho)E_0'\in\bC^{2^{n+\ell}\times 2^{n+\ell}}$ of\footnote{Similarly as described in Footnote~\ref{footnote:sdabuse}, we abuse notation by letting $\Enc$ denote the isomorphism from $\bF_2^k\xrightarrow{\sim}H^i(\cC)$, the associated encoding isometry, and the associated encoding channel. We also sometimes write $\Enc$ as a shorthand for $\Enc\otimes\cI_\ell$. The meaning will always be made clear from the argument.} $\Enc(\rho)$ (for Paulis $E_0,E_0'$ such that $\supp(E_0)\cup\supp(E_0')\subseteq C^i$ is $\cE_{\mathrm{in}}$-avoiding), and every $\cE_{\mathrm{run}}^{\sqcup T}$-avoiding Pauli fault $\cF$, the resulting output $\cQ[\cF](\sigma)$ is a Pauli $\cE_{\mathrm{out}}$-deviation of $\Enc(\rho)$.

  Let $S_{\cF}\subseteq C^{i-1}\sqcup C^i\sqcup C^{i+1}$ denote the set of all physical qubits that lie in (any timestep of) the forward-looking lightcone in $\cQ$ (given by Algorithm~\ref{alg:errcorr}) of any qubit in $\bigcup_{t\in[T]}\supp(\cF_t)$, where we take each lightcone starting from timestep $1$ for data qubits in $C^i$, or the timestep the qubit was initialized for ancilla qubits in $C^{i-1}$ and $C^{i+1}$. Note that this foward-looking lightcone only counts quantum gates in line~\ref{li:ecseZ}, line~\ref{li:eccorrZ}, line~\ref{li:ecseX}, and line~\ref{li:eccorrX}, and does not count classical gates from line~\ref{li:ecdecZ} and line~\ref{li:ecdecX}.

  Similarly, let $S_{\mathrm{in}}\subseteq C^{i-1}\sqcup C^i\sqcup C^{i+1}$ denote the set of all physical qubits that lie in (any timestep of) the forward-looking lightcone in $\cQ$ of any point in $\supp(E_0)\cup\supp(E_0')$.

  The following claim follows from the same reasoning used to show Claim~\ref{claim:SFavoid} in Section~\ref{sec:stateprep}, so we omit the proof to avoid redundancy.
  
  \begin{claim}
    \label{claim:ecSFinavoid}
    The set $S_{\cF}$ is $\cE(G_{\mathrm{run}},\; \eta_{\mathrm{run}},\; T2^T\gamma_{\mathrm{run}})$-avoiding, and the set $S_{\mathrm{in}}$ is $\cE(G_{\mathrm{run}},\; \eta_{\mathrm{in}},\; 2^T\gamma_{\mathrm{in}})$-avoiding.
  \end{claim}

  Lemma~\ref{lem:ecz} below analyzes the $Z$-correction in lines~\ref{li:ecseZ}--\ref{li:eccorrZ} of Algorithm~\ref{alg:errcorr}; the analysis of the $X$-correction in lines~\ref{li:ecseX}--\ref{li:eccorrX} is then exactly analogous.
  
  \begin{lemma}
    \label{lem:ecz}
    Following the execution of line~\ref{li:eccorrZ} in Algorithm~\ref{alg:errcorr}, the joint state of the code register $C^i$ and the $\ell$-qubit ancilla system is proportional to $E_2\Enc(\rho)E_2'$ for a Pauli error $E_2=Z^{e_{2,Z}}X^{e_{2,X}},\; E_2'=X^{e_{2,X}'}Z^{e_2,Z'}$ such that
    \begin{equation*}
      \supp(e_{2,X})\cup\supp(e_{2,X}')\subseteq C^i \hspace{1em}\text{ is }\hspace{1em} \cE\left(G_{\mathrm{run}},\; \eta_{\mathrm{run}},\; \frac{(w+1)T2^{T+4}\gamma_{\mathrm{run}}}{\gamma_{\mathrm{dec}}}\right)\text{-avoiding},
    \end{equation*}
    and
    \begin{equation*}
      \supp(e_{2,Z})\cup\supp(e_{2,Z})\subseteq (S_{\mathrm{in}}\cup S_{\cF})\cap C^i.
    \end{equation*}
  \end{lemma}
  \begin{proof}
    By the assumption that all on the input and in the fault $\cF$ are Pauli errors, and all gates that we apply are Clifford. Therefore throughout the execution from the beginning of Algorithm~\ref{alg:errcorr} through the moment just prior to the Pauli $Z$ measurements in the call to \FnSSFlipSyn{$s_Z;i,\cC^*$} in line~\ref{li:ecseZ}, the state differs from the noisless execution (with no input or fault errors) by a Pauli error. In this noiseless execution, all measurement outcomes are $0$. Therefore in the noisy execution, the measurement outcome $s_Z\in\cC^{i+1}$ is also deterministic, and the post-measurement state is either $0$ (if one of the syndrome qubits has an $X$ error from one side but not the other) or else is
    \begin{equation}
      \label{eq:ecpostmeas}
      Z^{e_{1,Z}}X^{e_{1,X}}\Enc(\rho)X^{e_{1,X}'}Z^{e_{1,Z}'}
    \end{equation}
    for some $e_{1,X},e_{1,Z},e_{1,X}',e_{1,Z}'\in\cC^i$ supported inside $S_{\mathrm{in}}\cup S_{\cF}$, such that\footnote{Throughout this proof, we write $\delta=\delta^{\cC}$.}
    \begin{equation}
      \label{eq:sZequiv}
      \delta(e_{1,X})+f_X = s_Z = \delta(e_{1,X}')+f_X'
    \end{equation}
    for some $f_X,f_X'\in\cC^{i+1}$ supported inside $S_{\cF}$. If the state collapses to $0$, then the output is $0$ and the lemma holds, so assume that the post-measurement state is of the form in~(\ref{eq:ecpostmeas}).
    
    Let $S_{\mathrm{syn}}\subseteq C^i\sqcup C^{i+1}$ (resp.~$S_{\mathrm{syn}}'\subseteq C^i\sqcup C^{i+1}$) denote the footprint (see Definition~\ref{def:footprint}) of the call to \FnSSFlipSyn{$s_Z=\delta(e_{1,X})+f_X;\; i,\; \cC^*$} (resp.~\FnSSFlipSyn{$s_Z=\delta(e_{1,X}')+f_X';\; i,\; \cC^*$}) in line~\ref{li:ecdecZ} of Algorithm~\ref{alg:errcorr}. Note that we could have $S_{\mathrm{syn}}\neq S_{\mathrm{syn}}'$ even though by~(\ref{eq:sZequiv}) holds because by Definition~\ref{def:footprint}, the footprint can depend on the choice decomposition of $s_Z$ into a syndrome plus a syndrome error. The claims we prove below for the unprimed variables $e_{1,X},f_X,\dots$ apply analogously to the primed variables $e_{1,X}',f_X',\dots$.

    \begin{claim}
      \label{claim:ecsynavoid}
      The set $S_{\mathrm{in}}\cup S_{\cF}\cup S_{\mathrm{syn}}$ is $\cE(G_{\mathrm{run}},\; \eta_{\mathrm{in}},\; 5w^3T2^{T+1}\gamma_{\mathrm{in}})$-avoiding.
    \end{claim}
    \begin{proof}
      The proof is similar to that of Claim~\ref{claim:synavoid}. Assume for a contradiction that there exists some $V\subseteq V_{\mathrm{run}}$ of size $|V|\geq\eta_{\mathrm{in}}$ that induces a connected subgraph of $G_{\mathrm{run}}$ with $|V\cap(S_{\mathrm{in}}\cup S_{\cF}\cup S_{\mathrm{syn}})|\geq 5w^3T2^{T+1}\gamma_{\mathrm{in}}$. As in the proof of Claim~\ref{claim:synavoid}, we may repeatedly add points in $(S_{\mathrm{in}}\cup S_{\cF}\cup S_{\mathrm{syn}})\setminus V$ that lie in the neighborhood of $V$, in order to assume that $V$ contains all or none of every connected component of the subgraph of $G_{\mathrm{run}}$ induced by $S_{\mathrm{in}}\cup S_{\cF}\cup S_{\mathrm{syn}}$.

      Recall as described above that the call \FnSSFlipSyn{$s_Z;i,\cC^*$} in line~\ref{li:ecdecZ} of Algorithm~\ref{alg:errcorr} has $s_Z=\delta(e_{1,X})+f_X$ with $\supp(e_{1,X})\subseteq S_{\mathrm{in}}\cup S_{\cF}\cap C^i$ and $\supp(f_X)\subseteq S_{\mathrm{in}}\cup S_{\cF}\cap C^{i+1}$. Then by Lemma~\ref{lem:footprint}, we have
      \begin{align*}
        |V\cap S_{\mathrm{syn}}|
        &\leq |V\cap(\supp(e_{1,X})\cup\supp(f))|/\gamma_{\mathrm{syn}} \\
        &\leq |V\cap (S_{\mathrm{in}}\cup S_{\cF})| \cdot 4w^3 \\
        &< (2^T\gamma_{\mathrm{in}}+T2^T\gamma_{\mathrm{run}})|V| \cdot 4w^3 \\
        &\leq T2^{T+1}\gamma_{\mathrm{in}}|V|\cdot 4w^3.
      \end{align*}
      where the third inequality above holds by Claim~\ref{claim:ecSFinavoid} and because $\eta_{\mathrm{run}}\leq\eta_{\mathrm{in}}$, and the fourth inequality holds because $\gamma_{\mathrm{run}}\leq\gamma_{\mathrm{in}}$ (see~(\ref{eq:ecbadparams}). We then again apply Claim~\ref{claim:ecSFinavoid} to conclude that
      \begin{equation*}
        |V\cap (S_{\mathrm{in}}\cup S_{\cF}\cup S_{\mathrm{syn}})| \leq |V\cap (S_{\mathrm{in}}\cup S_{\cF})|+|V\cap S_{\mathrm{syn}}| < (4w^3+1)T2^{T+1}\gamma_{\mathrm{in}}|V|,
      \end{equation*}
      which contradicts the assumption that $|V\cap(S_{\mathrm{in}}\cup S_{\cF}\cup S_{\mathrm{syn}})|\geq 5w^3T2^{T+1}\gamma_{\mathrm{in}}$, as desired.
    \end{proof}
    
    Similarly as in the proof of Proposition~\ref{prop:stateprep} and Proposition~\ref{prop:switchdown}, each update $c^i$ computed in \FnSSFlipSyn{$s_Z;i,\cC^*$} is supported within the neighborhood in $G^{\cC}_{i-1,i,i+1}\subseteq G_{\mathrm{run}}$ of the footprint from the previous step. As a consequence, if we restrict $s_Z$ to a given connected component $V$ of the subgraph of $G_{\mathrm{run}}$ induced by $S_{\mathrm{in}}\cup S_{\cF}\cup S_{\mathrm{syn}}$ (and zero out all values outside of $V$), the output of \FnSSFlipSyn{$\delta(s_Z|_{V\cap C^i});i,\cC^*$} must equal the restriction $a_Z^i|_{V\cap C^i}$ of the output $a_Z^i$ of \FnSSFlipSyn{$s_Z;i+1,\cC^*$} (see Lemma~\ref{lem:sslocal}).

    Recall that by assumption $\eta_{\mathrm{in}}\leq m$. Also by definition, $\gamma_{\mathrm{in}}\leq 1/5w^3T2^{T+1}$, so it follows by Claim~\ref{claim:ecsynavoid} that the subgraph of $G_{\mathrm{run}}$ induced by $S_{\mathrm{in}}\cup S_{\cF}\cup S_{\mathrm{syn}}$ has no connected components containing $\geq\eta_{\mathrm{in}}$ vertices.
    
    
    Furthermore, by definition $\supp(e_{1,X})\subseteq S_{\mathrm{in}}\cup S_{\cF}$ and $\supp(a_Z^i)\subseteq S_{\mathrm{syn}}$, so $\supp(e_{1,X}+a_Z^i)\subseteq S_{\mathrm{in}}\cup S_{\cF}\cup S_{\mathrm{syn}}$.

    Let $S_{\mathrm{err}}$ and $\tilde{a}_Z^{i-1},\tilde{a}_Z^i$ (resp.~$S_{\mathrm{err}}'$ and $(\tilde{a}_Z^{i-1})',(\tilde{a}_Z^i)'$) denote the footprint and output respectively of \FnSSFlipErr{$e_{1,X}+a_Z^i\;i,\cC^*$} (resp.~\FnSSFlipErr{$e_{1,X}'+a_Z^i\;i,\cC^*$}).

    \begin{claim}
      \label{claim:ecerravoid}
      The set $S_{\mathrm{in}}\cup S_{\cF}\cup S_{\mathrm{syn}}\cup S_{\mathrm{err}}$ is $\cE(G_{\mathrm{run}},\; \eta_{\mathrm{in}},\; 50w^7T2^{T+1}\gamma_{\mathrm{in}})$-avoiding.
    \end{claim}
    \begin{proof}
      The proof is similar to that of Claim~\ref{claim:erravoid}. Assume for a contradiction that there exists some $V\subseteq V_{\mathrm{run}}$ of size $|V|\geq\eta_{\mathrm{in}}$ that induces a connected subgraph of $G_{\mathrm{run}}$ with $|V\cap(S_{\mathrm{in}}\cup S_{\cF}\cup S_{\mathrm{syn}}\cup S_{\mathrm{err}})|\geq 50w^7T2^{T+1}\gamma_{\mathrm{in}}$. As in the proof of Claim~\ref{claim:ecsynavoid}, we may repeatedly add points in $(S_{\mathrm{in}}\cup S_{\cF}\cup S_{\mathrm{syn}}\cup S_{\mathrm{err}})\setminus V$ that lie in the neighborhood of $V$, in order to assume that $V$ contains all or none of every connected component of the subgraph of $G_{\mathrm{run}}$ induced by $S_{\mathrm{in}}\cup S_{\cF}\cup S_{\mathrm{syn}}\cup S_{\mathrm{err}}$.
      
      By Lemma~\ref{lem:footprint}, we have
      \begin{align*}
        |V\cap S_{\mathrm{err}}|
        &\leq |V\cap\supp(e_{1,X}+a_Z^i)|/\gamma_{\mathrm{err}} \\
        &\leq |V\cap(S_{\mathrm{in}}\cup S_{\cF}\cup S_{\mathrm{syn}})| \cdot 8w^4 \\
        &< 5w^3T2^{T+1}\gamma_{\mathrm{in}}|V| \cdot 8w^4,
      \end{align*}
      where the third inequality above holds by Claim~\ref{claim:ecsynavoid}. We then again apply Claim~\ref{claim:ecsynavoid} to conclude that
      \begin{equation*}
        |V\cap(S_{\mathrm{in}}\cup S_{\cF}\cup S_{\mathrm{syn}}\cup S_{\mathrm{err}})| \leq |V\cap(S_{\mathrm{in}}\cup S_{\cF}\cup S_{\mathrm{syn}})| + |V\cap S_{\mathrm{err}}| < (8w^4+1)5w^3T2^{T+1}\gamma_{\mathrm{in}}|V|,
      \end{equation*}
      which contradicts the assumption that $|V\cap(S_{\mathrm{in}}\cup S_{\cF}\cup S_{\mathrm{syn}}\cup S_{\mathrm{err}})|\geq 50w^7T2^{T+1}\gamma_{\mathrm{in}}$, as desired.
    \end{proof}

    Similarly as in the proof of Proposition~\ref{prop:stateprep} and Proposition~\ref{prop:switchdown}, each update $c^{i-1}$ and $c^i$ computed in \FnSSFlipErr{$e_{1,X}+a_Z^i\;i,\cC^*$} is supported within the neighborhood in $G^{\cC}_{i-1,i,i+1}\subseteq G_{\mathrm{run}}$ of the footprint from the previous step. As a consequence, if we restrict $e_{1,X}+a_Z^i$ to a given connected component $V$ of the subgraph of $G_{\mathrm{un}}$ induced by $S_{\mathrm{in}}\cup S_{\cF}\cup S_{\mathrm{syn}}\cup S_{\mathrm{err}}$ (and zero out all values outside of $V$), the output of \FnSSFlipErr{$e_{1,X}+a_Z^i|_{V\cap C^i}\;i,\cC^*$} must equal the restriction $\tilde{a}_Z^{i-1}|_{V\cap C^{i-1}},\; \tilde{a}_Z^i|_{V\cap C^i}$ of the output $\tilde{a}_Z^{i-1},\tilde{a}_Z^i$ of \FnSSFlipErr{$e_{1,X}+a_Z^i\;i,\cC^*$} (see Lemma~\ref{lem:sslocal}).

    Recall that by assumption $\eta_{\mathrm{in}}\leq m$. Also by assumption~(\ref{eq:ecbadparams}), $\gamma_{\mathrm{in}}\leq 1/50w^7T2^{T+1}$, so it follows by Claim~\ref{claim:ecerravoid} that the subgraph of $G_{\mathrm{run}}$ induced by $S_{\mathrm{in}}\cup S_{\cF}\cup S_{\mathrm{syn}}\cup S_{\mathrm{err}}$ has no connected components containing $\geq\eta_{\mathrm{in}}$ vertices. Therefore within every connected component $V$ of this induced subgraph, by the assumption that $\cC^*$ has a $m$-small-set error-flip decoder at level $i$, the while loop in the restricted execution \FnSSFlipErr{$e_{1,X}+a_Z^i|_{V\cap C^i}\;i,\cC^*$} will not terminate until the computed output $\tilde{a}_Z^{i-1}|_{V\cap C^{i-1}},\; \tilde{a}_Z^i|_{V\cap C^i}$ satisfies
    \begin{equation*}
      \delta(\tilde{a}_Z^{i-1}|_{V\cap C^{i-1}})+(\tilde{a}_Z^i|_{V\cap C^i})=e_{1,X}+a_Z^i|_{V\cap C^i}.
    \end{equation*}
    It also follows that the output will never be FAIL, and that
    \begin{equation}
      \label{eq:taea}
      \delta(\tilde{a}_Z^{i-1})+\tilde{a}_Z^i = e_{1,X}+a_Z^i,
    \end{equation}
    where by definition $\supp(\tilde{a}_Z^{i-1}),\delta(\supp(\tilde{a}_Z^{i-1})),\supp(\tilde{a}_Z^i)\subseteq S_{\mathrm{err}}$.

    Furthermore, because every update to $\tilde{a}_Z^i|_{V\cap C^i}$ in the execution of \FnSSFlipErr{$e_{1,X}+a_Z^i|_{V\cap C^i}\;i,\cC^*$} increases $|\tilde{a}_Z^i|_{V\cap C^i}|$ by at most $|c^i|\leq w$ while decreasing $|\delta(e_{1,X}+a_Z^i+\tilde{a}_Z^i|_{V\cap C^i})|$ by at least $1$, we must have
    \begin{equation}
      \label{eq:tadta}
      |\tilde{a}_Z^i|_{V\cap C^i}| \leq w|\delta(e_{1,X}+a_Z^i|_{V\cap C^i})| = w|\delta(\tilde{a}_Z^i|_{V\cap C^i})|.
    \end{equation}

    \begin{claim}
      \label{claim:tSavoid}
      The set $\tilde{S}:=\supp(\tilde{a}_Z^i)\cup\supp(\delta(\tilde{a}_Z^i))$ is $\cE(G_{\mathrm{run}},\; \eta_{\mathrm{run}},\; \tilde{\gamma})$-avoiding for
      \begin{equation}
        \label{eq:tgamma}
        \tilde{\gamma} = \frac{(w+1)T2^{T+2}\gamma_{\mathrm{run}}}{\gamma_{\mathrm{dec}}}.
      \end{equation}
    \end{claim}
    \begin{proof}
      Assume for a contradiction that there exists some $V\subseteq V_{\mathrm{run}}$ of size $|V|\geq\eta_{\mathrm{run}}$ that induces a connected subgraph of $G_{\mathrm{run}}$ with $|V\cap\tilde{S}|/|V|\geq \tilde{\gamma}$. We may repeatedly add points in $(S_{\cF}\cup\tilde{S})\setminus V$ that lie in the neighborhood of $V$, in order to assume that $V$ contains all or none of every connected component of the subgraph of $G_{\mathrm{run}}$ induced by $S_{\cF}\cup\tilde{S}$. As adding points in $S_{\cF}$ cannot decrease the fraction of points in $V$ that lie in $S_{\cF}\cap\tilde{S}$, We also then must have
      \begin{equation*}
        |V\cap(S_{\cF}\cup\tilde{S})| \geq \tilde{\gamma}|V|.
      \end{equation*}
      However, by Claim~\ref{claim:ecSFinavoid}
      \begin{equation}
        \label{eq:VcSF}
        |V\cap S_{\cF}| \leq T2^T\gamma_{\mathrm{run}}|V|,
      \end{equation}
      so we must have
      \begin{equation}
        \label{eq:VctS}
        |V\cap\tilde{S}| \geq (\tilde{\gamma}-T2^T\gamma_{\mathrm{run}})|V|.
      \end{equation}

      Now because $\tilde{S}\subseteq S_{\mathrm{err}}$ and $\gamma_{\mathrm{in}}\leq 1/50w^7T2^{T+1}$, Claim~\ref{claim:ecerravoid} implies that every connected component $V'\subseteq V_{\mathrm{run}}$ of the subgraph of $G_{\mathrm{run}}$ induced by $S_{\cF}\cup\tilde{S}\subseteq S_{\cF}\cup S_{\mathrm{err}}$ contains at most $|V'|\leq\eta_{\mathrm{in}}\leq m$ vertices. Therefore for every such $V'$, we must have
      \begin{equation}
        \label{eq:tabound}
        |\delta(\tilde{a}_Z^i|_{V'\cap C^i})| \leq \frac{|f_X|_{V'\cap C^{i+1}}|}{\gamma_{\mathrm{dec}}}.
      \end{equation}
      Specifically, if~(\ref{eq:tabound}) did not hold, then by the assumption that $\cC^*$ has a $(m,\gamma_{\mathrm{dec}})$-small-set syndrome-flip decoder at level $i$ (see Definition~\ref{def:ssflip}), there would exist some $c^0\in C^0$ and $c^i\in\cC^i$ with $c^0\preceq c^i$ such that
      \begin{align*}
        |\delta(c^i)+(\delta(\tilde{a}_Z^i)+f_X)|_{V'\cap C^{i+1}}|
        &< |\delta(\tilde{a}_Z^i)+f_X|_{V'\cap C^{i+1}}|.
      \end{align*}
      In order for the above inequality to hold, we must have $\supp(\delta(c^i))\cap\supp(\delta(\tilde{a}_Z^i)+f_X|_{V'\cap C^{i+1}})\neq\emptyset$, and hence $c^0$ lies above some element of $\supp(\delta(\tilde{a}_Z^i)+f_X|_{V'\cap C^{i+1}})$, so every element of $\supp(c^i)$ and $\supp(\delta(c^i))$ lie inside the neighborhood in $G_{\mathrm{run}}$ of $\supp(\delta(\tilde{a}_Z^i)+f_X|_{V'\cap C^{i+1}})\subseteq V'$. But by definition no other connected components of the subgraph of $G_{\mathrm{run}}$ induced by $S_{\cF}\cup\tilde{S}$ intersect this neighborhood, and $\supp(\tilde{a}_Z^i),\supp(\delta(\tilde{a}_Z^i))\subseteq\tilde{S}$ and $\supp(f_X)\subseteq S_{\cF}$, so it follows that 
      \begin{align*}
        |\delta(c^i)+\delta(\tilde{a}_Z^i)+f_X|
        &< |\delta(\tilde{a}_Z^i)+f_X|.
      \end{align*}
      By~(\ref{eq:taea}) we have $\delta(\tilde{a}_Z^i)=\delta(e_{1,X}+a_Z^i)$, so the above inequality is equivalent to
      \begin{align*}
        |\delta(c^i)+\delta(e_{1,X}+a_Z^i)+f_X|
        &< |\delta(e_{1,X}+a_Z^i)+f_X|.
      \end{align*}
      But this inequality contradicts the assumption that \FnSSFlipSyn{$s_Z=\delta(e_{1,X})+f_X;i+1,\cC^*$} terminated with output $a_Z^i$, as $c^i$ is a valid update to add in to $a_Z^i$ that further reduces the sydrome weight. Thus~(\ref{eq:tabound}) holds for every connected component $V'\subseteq V_{\mathrm{run}}$ of the subgraph of $G_{\mathrm{run}}$ induced by $S_{\cF}\cup\tilde{S}$.

      By construction $V$ contains all or none of every such connected component $V'$, and $\tilde{a}_Z^i,\delta(\tilde{a}_Z^i),f_X$ are supported inside $S_{\cF}\cup\tilde{S}$, so it follows that
      \begin{equation*}
        |\delta(\tilde{a}_Z^i|_{V\cap C^i})| \leq \frac{|f_X|_{V\cap C^{i+1}}|}{\gamma_{\mathrm{dec}}}.
      \end{equation*}
      However, because $\supp(f_X)\subseteq S_{\cF}$, by~(\ref{eq:VcSF}) we have
      \begin{equation*}
        |f_X|_{V\cap C^{i+1}}| \leq |V\cap S_{\cF}| \leq T2^T\gamma_{\mathrm{run}}|V|.
      \end{equation*}
      Meanwhile, because by~(\ref{eq:tadta})
      \begin{equation*}
        |V\cap\tilde{S}| = |\tilde{a}_Z^i|_{V\cap C^i}|+|\delta(\tilde{a}_Z^i|_{V\cap C^i})| \leq (w+1)|\delta(\tilde{a}_Z^i|_{V\cap C^i})|,
      \end{equation*}
      by~(\ref{eq:VctS}) we have
      \begin{equation*}
        |\delta(\tilde{a}_Z^i|_{V\cap C^i})| \geq \frac{|V\cap\tilde{S}|}{w+1} \geq \frac{\tilde{\gamma}-T2^T\gamma_{\mathrm{run}}}{w+1}|V|.
      \end{equation*}
      Combining the above inequalities, we conclude that
      \begin{equation*}
        \frac{\tilde{\gamma}-T2^T\gamma_{\mathrm{run}}}{w+1} \leq \frac{T2^T\gamma_{\mathrm{run}}}{\gamma_{\mathrm{dec}}},
      \end{equation*}
      which contradicts the definition of $\tilde{\gamma}$ in~(\ref{eq:tgamma}), where here we use the fact that $\gamma_{\mathrm{dec}}\leq 1$ by Remark~\ref{remark:gdl1}. Thus the assumption that there exists some $V\subseteq V_{\mathrm{run}}$ of size $|V|\geq\eta_{\mathrm{run}}$ that induces a connected subgraph of $G_{\mathrm{run}}$ with $|V\cap\tilde{S}|/|V|\geq \tilde{\gamma}$ was false, as desired.
    \end{proof}

    Recall that the entire analysis above was stated for the unprimed variables $e_{1,X},f_X,\dots$ but similarly applies to the primed variables $e_{1,X}',f_X',\dots$, so that for instance~(\ref{eq:taea}) becomes
    \begin{equation}
      \label{eq:taeap}
      \delta((\tilde{a}_Z^{i-1})')+(\tilde{a}_Z^i)' = e_{1,X}'+a_Z^i.
    \end{equation}
    Then by~(\ref{eq:taea}) and~(\ref{eq:taeap}), applying the correction $X^{a_Z^i}$ in line~\ref{li:eccorrZ} of Algorithm~\ref{alg:errcorr} to the state in~(\ref{eq:ecpostmeas}) yields the corrected state 
    \begin{equation*}
      Z^{e_{1,Z}}X^{e_{1,X}+a_Z^i}\Enc(\rho)X^{e_{1,X}'+a_Z^i}Z^{e_{1,Z}'} = Z^{e_{1,Z}}X^{\tilde{a}_Z^i}\Enc(\rho)X^{(\tilde{a}_Z^i)'}Z^{e_{1,Z}'},
    \end{equation*}
    where we use the fact that $X^{\delta(\tilde{a}_Z^{i-1})},X^{\delta((\tilde{a}_Z^{i-1})')}$ are by definition stabilizers of our code $Q$ and hence preserve the code state $\Enc(\rho)$. Recall here that $\supp(e_{1,Z}),\supp(e_{1,Z'})\subseteq(S_{\mathrm{in}}\cup S_{\cF})\cap C^i$, while $\supp(\tilde{a}_Z^i),\supp(\tilde{a}_Z^i)'$ are both $\cE(G_{\mathrm{run}},\;\eta_{\mathrm{run}},\;\tilde{\gamma})$-avoiding by Claim~\ref{claim:tSavoid}.

    An additional Pauli error from $\cF$, which is supported inside $S_{\cF}$, will also occur at the timestep of line~\ref{li:eccorrZ}. Therefore the final state after line~\ref{li:eccorrZ} will be
    \begin{equation*}
      Z^{e_{2,Z}}X^{e_{2,X}}\Enc(\rho)X^{e_{2,X}'}Z^{e_{2,Z}'},
    \end{equation*}
    where $\supp(e_{2,Z})\cup\supp(e_{2,Z}')\subseteq(S_{\mathrm{in}}\cup S_{\cF})\cap C^i$, and $\supp(e_{2,X})\cup\supp(e_{2,X}')\subseteq C^i$ is $\cE(G_{\mathrm{run}},\;\eta_{\mathrm{run}},\;\gamma)$-avoiding for
    \begin{equation*}
      \gamma = 2\tilde{\gamma}+T2^T\gamma_{\mathrm{run}} \leq \frac{(w+1)T2^{T+4}\gamma_{\mathrm{run}}}{\gamma_{\mathrm{dec}}},
    \end{equation*}
    as desired, where we have also applied Claim~\ref{claim:ecSFinavoid}
  \end{proof}

  Applying exactly analogous reasoning as used to prove Lemma~\ref{lem:ecz} regarding the $Z$-correction in lines~\ref{li:ecseZ}--\ref{li:eccorrZ} of Algorithm~\ref{alg:errcorr}, but instead to the $X$-correction in lines~\ref{li:ecseX}--\ref{li:eccorrX}, yields the desired result. Specifically, because the only additional $X$ errors after line~\ref{li:eccorrX} arise from $\cF$, we conclude that the final output of Algorithm~\ref{alg:errcorr} is proportional to $E_3\Enc(\rho)E_3'$ for a Pauli error $E_3=Z^{e_{3,Z}}X^{e_{3,X}},\; E_3'=X^{e_{3,X}'}Z^{e_3,Z'}$ such that $\supp(e_{3,X})\cup\supp(e_{3,X}')\subseteq C^i$ and $\supp(e_{3,Z})\cup\supp(e_{3,Z}')\subseteq C^i$ are both
  \begin{equation*}
    \cE\left(G_{\mathrm{run}},\; \eta_{\mathrm{run}},\; \frac{(w+1)T2^{T+5}\gamma_{\mathrm{run}}}{\gamma_{\mathrm{dec}}}\right)\text{-avoiding},
  \end{equation*}
  so the overall error support $\supp(E_3)\cup\supp(E_3')$ is 
  \begin{equation*}
    \cE\left(G_{\mathrm{run}},\; \eta_{\mathrm{run}},\; \frac{(w+1)T2^{T+6}\gamma_{\mathrm{run}}}{\gamma_{\mathrm{dec}}}\right)\text{-avoiding},
  \end{equation*}
  as desired.
\end{proof}

\section{Fault-Tolerance Proof for Upwards Code Switching}
\label{app:suproof}
In this section, we present the proof of Proposition~\ref{prop:switchup}. As described in Section~\ref{sec:switchup}, this proof simply combines gadgets presented previously in the paper.

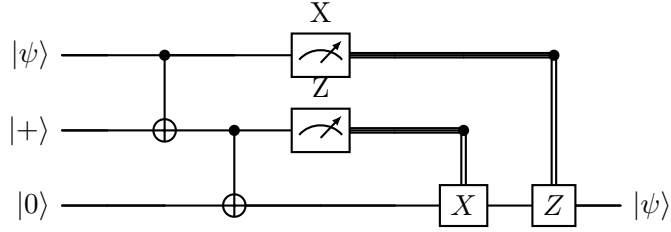
\begin{figure}
  \centering
  \begin{quantikz}[row sep={1cm,between origins}, column sep=0.6cm]
    \lstick{$\ket{\psi}$} & \qw      & \ctrl{1} & \qw       & \meter{$X$} & \cw & \cw & \cwbend{2} \\
    \lstick{$\ket{+}$}    & \qw      & \targ{}  & \ctrl{1}  & \meter{$Z$} & \cw & \cwbend{1} \\
    \lstick{$\ket{0}$}    & \qw      & \qw      & \targ{}   & \qw         & \qw & \gate{X}     & \gate{Z} & \qw \rstick{$\ket{\psi}$}
  \end{quantikz}
  \caption{\label{fig:teleport} Standard circuit for teleporting a qubit $\ket{\psi}$.}
\end{figure}


\begin{proof}[Proof of Proposition~\ref{prop:switchup}]
  We construct the gadget $\cQ$ by implementing the standard teleportation gadget in Figure~\ref{fig:teleport} using the gadgets previously presented in this paper. Specifically, Figure~\ref{fig:teleport} shows the circuit for teleporting a single qubit. We implement this circuit on $k_{\mathrm{in}}$ logical qubits simultaneously, where the first wire corresponds to the input code state in $Q_{\mathrm{in}}$, the second wire corresponds to the state $\Enc_{\mathrm{in}}(\ket{+}\bra{+}^{\otimes k_{\mathrm{in}}})$ the we prepare, and the third wire corresponds to the state $\Enc_{\mathrm{out}}(\ket{0}\bra{0}^{\otimes k_{\mathrm{out}}})$ that we also prepare. The extra $k_{\mathrm{out}}-k_{\mathrm{in}}$ logical $\ket{0}$ qubits in this third codeblock are precisely the logical $\ket{0}$ qubits that $\bar{\cO}$ pads onto the input.

  Specifically, we prepare $\Enc_{\mathrm{out}}(\ket{0}\bra{0}^{\otimes k_{\mathrm{out}}})$ directly using the gadget in Proposition~\ref{prop:stateprep}, though where we dualize to the chain complex $\cC_*$ and exchange the roles of the $X$ and $Z$ bases (see Remark~\ref{remark:spexchange}). Here we rely on the $(m,0)$-small-set flip decodability of $\cC_*$ at level $i-1$.

  We are unable to similarly prepare $\Enc_{\mathrm{in}}(\ket{+}\bra{+}^{\otimes k_{\mathrm{in}}})$ by directly applying Proposition~\ref{prop:stateprep}, as we do not assume the necessary small-set flip decoders for this lower-dimensional code (see Remark~\ref{remark:sulowdim}). Instead, letting $k_{\cD}=\dim(H^{i-1}(\cD))$, we first apply Proposition~\ref{prop:stateprep} to prepare $\Enc^{\cD,i-1}(\ket{0}\bra{0}^{\otimes k_{\cD}}$. Here we use the $(m,0)$-small-set flip decodability of $\cD^*$ at level $i$. We then apply Proposition~\ref{prop:switchdown} to map this $r$-dimensional code state state down to the $(r-1)$-dimensional code state $\Enc_{\mathrm{in}}(\ket{0}\bra{0}^{\otimes k_{\mathrm{in}}})=(\Enc^{\cA})^{\sqcup L^{\cB,1}}(\ket{0}\bra{0}^{\otimes k_{\mathrm{in}}})$. Here we use the $(m,0)$-small-set flip decodability of $\cD_{\bar{L},*}$ at level $i-1$. Also note that in this application of Proposition~\ref{prop:switchdown}, we have dualized to a chain complex, and exchanged the roles of the $X$ and $Z$ bases (see Remark~\ref{remark:sdexchange}).

  Once we have prepared these logical ancilla states, we then perform the two logical $CNOT$ gates in Figure~\ref{fig:teleport} using Lemma~\ref{lem:CNOTsame} and Lemma~\ref{lem:CNOTdiff} respectively. We then perform the logical Pauli measurements using Lemma~\ref{lem:measure}, and the Pauli corrections by simply applying the appropriate logical Pauli operators (in a depth-$1$ Pauli circuit; the fault-tolerance analysis of this step is straightforward).

  The desired result then follows directly from the results listed above proving fault-tolerance of our gadgets; Proposition~\ref{prop:parcomp} implies that the fault-tolerance holds even when each gadget in the construction only acts on a subset of the physical qubits in the overall circuit.
\end{proof}

\begin{remark}
  \label{remark:sulowdim}
  When $r=3$, then $(r-1)=2$-dimensional codes typically will not support the single-shot state preparation gadget given in Section~\ref{sec:stateprep}. Hence as described in the proof of Proposition~\ref{prop:switchup} above, we prepare the logical $\ket{+}$ state in Figure~\ref{fig:teleport} using the downwards code switching gadget from Section~\ref{sec:switchdown}.

  If we consider codes of one dimension higher, and only consider $r\geq 4$ and $2\leq i\leq r-2$, then we are able to avoid using the downwards code switching gadget, and instead directly prepare the logical $\ket{+}$ states using Proposition~\ref{prop:stateprep}. In fact, this approach also allows us to perform downwards code switching without Proposition~\ref{prop:switchdown} (using a similar construction as in Proposition~\ref{prop:switchup} with state preparation via Proposition~\ref{prop:stateprep}). This approach furthermore allows us to instantiate these gadgets with codes of rate arbitrarily close to $1$, as opposed to just some small constant rate, because we would no longer need small-set flip decodability of the restricted complex $\cD_{\bar{L},*}$ (recall that $R>0$ is just a small constant in Corollary~\ref{cor:losslessfamily}).

  To avoid redundancy in our presentation, we primarily state results that apply to arbitrary $r\geq 3$; the extensions described above for the $r\geq 4$ case are then straightforward adaptations of our arguments.
\end{remark}

\section{Proposed Method for Universal Computation via Transversal Non-Clifford Gates}
\label{app:protocol}

In the main text above, we focused on constant-overhead gadgets for various addressable and parallel Clifford gates. To achieve universal quantum computation, it suffices to add a gadget for performing a non-Clifford gate, such as $CCZ$ or $T$. Standard approaches include magic state distillation (see e.g.~\cite{bravyi_magic-state_2012}), or switching into a code that natively supports transversal (i.e.~constant-depth fault-tolerant) non-Clifford gates. As this latter approach seems well-suited for our techniques, in this appendix we discuss the possibility of applying our code switching to codes with transversal non-Clifford gates. Such a protocol could be viewed as an extension of the color/surface-code-based protocol of \cite{bombin_dimensional_2016} to qLDPC codes with better parameters. As explained below, here we simply propose a high-level approach without proving fault-tolerance. The independent and concurrent work of \cite{tan2025single} provides a more rigorous treatment for a specific family of codes based on those of \cite{zhu_topological_2025}.

We emphasize that while this section describes one proposal for universal fault-tolerant computation, our gadgets in the main text can already be applied to reduce the space-time overhead of specific components of existing schemes for universal fault-tolerant computation.

QLDPC codes with transversal non-Clifford (in particular, $CCZ$) gates often arise from 3-dimensional tensor products of classical LDPC codes (e.g.~\cite{bombin_topological_2007,golowich_quantum_2025-1,zhu_topological_2025})\footnote{The codes of \cite{golowich_quantum_2025-1} have polylogarithmic rather than constant locality (i.e.~check weight), so they may be considered ``nearly LDPC.''}. If our code-switching-based gadgets for Clifford gates could be applied to these codes, then the transversal $CCZ$ gate would provide a constant-overhead gadget for $CCZ$ gates, yielding a complete scheme for universal fault-tolerant quantum computation. In particular, the constructions of \cite{golowich_quantum_2025-1} and \cite{zhu_topological_2025} are able to perform many (specifically, $\Theta(n^{1-\epsilon})$ and $\Theta(n^{1/3})$ respectively; see Section~\ref{sec:transnc} below) logical $CCZ$ gates in a codeblock using a constant-depth physical circuit. Hence a code-switching-based scheme using these codes may yield improvements over distillation-based schemes, which often only perform a single logical non-Clifford gate in a codeblock in any given gadget (see e.g.~\cite{nguyen_quantum_2025}).

However, the construction of qLDPC codes with transversal non-Clifford gates has proven to be a challenging problem, and the constructions of \cite{golowich_quantum_2025-1,zhu_topological_2025} still have far from optimal parameters. It therefore remains an interesting direction of future work to determine if the overall space-time overhead of a fault-tolerance scheme based on these codes would improve upon existing schemes such as in \cite{nguyen_quantum_2025}.

Furthermore, our proofs of a fault-tolerance threshold rely on our small-set flip decoder in Section~\ref{sec:ssflip}. We only construct such a decoder for quantum codes arising as products of classical codes associated to lossless expanders. Yet known product constructions of qLDPC codes with transversal non-Clifford gates (see above) require (at least some of) the classical factor codes to have structure that is incompatible with lossless expansion. Hence our fault-tolerance proofs do not immediately apply to such codes.

Below, we provide some more details on the codes of \cite{golowich_quantum_2025-1,zhu_topological_2025}, which exhibit some of the best known parameters for qLDPC codes with transversal non-Clifford gates. We highlight interesting directions for future work regarding applying our code switching techniques to these codes.


\subsection{Codes with Transversal Non-Clifford Gates}
\label{sec:transnc}

We now provide more details on qLPDC codes with transversal non-Clifford (specifically $CCZ$) gates, to provide context for the question of performing code switching on such codes. In this section, for convenience we will use the non-standard notation that a quantum code is $[[n,k,d]]_{CCZ}$ if it supports a constant-depth physical circuit that induces logical $CCZ$ gates on $k$ disjoint triples of logical qubits.

For many years, the state-of-the-art qLDPC codes with transversal $CCZ$ were given by the $[[n,\Theta(1),\Theta(n^{1/3})]]_{CCZ}$ 3-dimensional color/surface code \cite{bombin_topological_2007} (see also \cite{bombin_exact_2007,bombin_gauge_2015,kubica_unfolding_2015}), which can be viewed as a tensor product of three chain complexes associated to classical repetition codes. 

Recently, for arbitrarily small constant $\epsilon>0$, \cite{golowich_quantum_2025-1} obtained a $[[n,\Theta(n^{1-\epsilon}),\tilde{\Theta}(n^{1/3})]]_{CCZ}$ construction that is nearly LDPC, meaning that the locality is polylogarithmic instead of constant. This almost-linear code dimension $k=\Theta(n^{1-\epsilon})$ was achieved by taking a tensor product of three chain complexes associated to classical algebraic codes, which can be viewed as Reed-Muller codes with sparsified parity-check matrices.

Subsequently, \cite{zhu_topological_2025} gave constructions of $[[n,\Theta(n^{1/3}),\Theta(n^{1/3})]]_{CCZ}$ and $[[n,\Theta(n^{1/2}),\Theta(n^{1/2})]]_{CCZ}$ qLDPC codes. These codes are obtained by ``thickening'' 3-dimensional product codes, meaning that the product code is mapped to a manifold or higher-dimensional CW complex using a modified version of the mapping from \cite{freedman_building_2021}. The underlying classical codes used in the product consist of both classical repetition codes and good classical LDPC codes (based on expander graphs, e.g.~\cite{sipser_expander_1996}). 

Both the construction of \cite{golowich_quantum_2025-1} and the $\Theta(n^{1/3})$-distance construction of \cite{zhu_topological_2025} are based on tensor products of classical LDPC codes, and hence may be amenable to similar code switching techniques as we develop. The $\Theta(n^{1/2})$-distance construction of \cite{zhu_topological_2025} uses the related but more involved lifted/balanced product \cite{breuckmann_balanced_2021,panteleev_asymptotically_2022,leverrier_quantum_2022-1,dinur_good_2023,lin_good_2022}. However, as the input codes in these constructions are not based on lossless expanders (see above), our fault-tolerance proofs do not directly apply. It is an interesting direction for future work to see if a threshold can still be proven for analogues of our code switching gadgets on these codes.

A positive resolution to this question would immediately yield a new scheme for universal fault-tolerant quantum computation. By the code parameters listed above, this scheme would be capable of performing polynomially many logical $CCZ$ gates in a codeblock using a constant-depth physical circuit. Specifically, the codes of \cite{golowich_quantum_2025-1} can perform $n^{1-\epsilon}$ such $CCZ$ gates in parallel for arbitrarily small constant $\epsilon>0$, while the $\Theta(n^{1/3})$-distance codes of \cite{zhu_topological_2025} can perform $\Theta(n^{1/3})$ such $CCZ$ gates in parallel. Note that while the polylogarithmic locality of the codes of \cite{golowich_quantum_2025-1} may present an apparent barrier against establishing a fault-tolerance threshold, this issue may be resolved by concatenating with a small code such as a polylogarithmic-sized surface code (see e.g.~\cite{pattison_hierarchical_2023}).

We emphasize again that our gadgets for Clifford gates in the main text already improve the efficiency of specific components of existing fault-tolerance schemes. If single-shot code switching for the qLDPC codes with $CCZ$ described above could be established, it would be an interesting direction to compare the overall scheme's space-time overhead to that of other methods for performing non-Clifford gates.

\end{document}